\documentclass[12pt]{article}

\usepackage
{hyperref}
\usepackage{amsmath, amssymb, amsthm}
\usepackage{xcolor}
\usepackage[sort&compress,numbers]{natbib}
\usepackage{doi}
\usepackage[margin=0.8in]{geometry}
\usepackage{amsmath, amssymb, amsthm, mathtools}

\newtheorem{theorem}{Theorem}[section]
\newtheorem{lemma}[theorem]{Lemma}

\newtheorem{prop}[theorem]{Proposition}

\theoremstyle{definition}
\newtheorem{defi}[theorem]{Definition}

\theoremstyle{remark}
\newtheorem{remark}[theorem]{Remark}

\numberwithin{equation}{section}

\def\cE{\mathcal E}
\def\cM{\mathcal M}
\def\cS{\mathcal S}
\def\Me{\mathcal M}
\def\Ne{\mathcal N}
\def\cF{\mathcal{F}}
\def\cR{\mathcal{R}}
\def\Tr{\mathrm{tr}}

\def\<{\langle}
\def\>{\rangle}
\def\ffi{\varphi}
\def\1{\mathbf{1}}
\def\eps{\varepsilon}
\def\bN{\mathbb{N}}
\def\bR{\mathbb{R}}
\def\bZ{\mathbb{Z}}
\def\bC{\mathbb{C}}
\def\Re{\mathrm{Re}\,}

\title{
$\alpha$-$z$-R\'enyi divergences in von Neumann algebras: \\
data processing inequality, reversibility, and \\
monotonicity properties in $\alpha,z$}

\author{Fumio Hiai\footnote{Graduate School of Information Sciences, Tohoku University,
Aoba-ku, Sendai 980-8579, Japan}
\footnote{{\it E-mail:} hiai.fumio@gmail.com} \
and Anna Jen\v cov\'a\footnote{Mathematical Institute, Slovak Academy
of Sciences, Bratislava, Slovakia}
\footnote{{\it E-mail:} jencaster@gmail.com}}

\begin{document}

\maketitle

\begin{abstract}
{We study the $\alpha$-$z$-R\'enyi divergences $D_{\alpha,z}(\psi\|\ffi)$ where $\alpha,z>0$ ($\alpha\ne1$)
for normal positive functionals $\psi,\ffi$ on general von Neumann
algebras, introduced in [S.~Kato and Y.~Ueda, arXiv:2307.01790] and [S.~Kato, arXiv:2311.01748].
We prove the variational expressions and the data processing inequality (DPI)
for the $\alpha$-$z$-R\'enyi divergences. We establish the sufficiency theorem for $D_{\alpha,z}(\psi\|\ffi)$,
saying that for $(\alpha,z)$ inside the DPI bounds, the equality
$D_{\alpha,z}(\psi\circ\gamma\|\ffi\circ\gamma)=D_{\alpha,z}(\psi\|\ffi)<\infty$ in the DPI under
a quantum channel (or a normal $2$-positive unital map) $\gamma$ implies the reversibility of
$\gamma$ with respect to $\psi,\ffi$. Moreover, we show the monotonicity properties of
$D_{\alpha,z}(\psi\|\ffi)$ in the parameters $\alpha,z$ and their limits to the normalized relative entropy as
$\alpha\nearrow1$ and $\alpha\searrow1$.}



\end{abstract}


\section{Introduction}\label{Sec-1}

The $\alpha$-$z$-R\'enyi divergences were {first introduced in \cite{jaksic2012entropic}, and
were further studied} in \cite{audenaert2015alpha} as a two-parameter family of quantum generalizations
of the classical R\'enyi divergences. For a pair of density operators (or more generally positive
operators) $\rho$ and $\sigma$ on a finite dimensional Hilbert space, the divergence
$D_{\alpha,z}(\rho\|\sigma)$ is defined by 
\[
D_{\alpha,z}(\rho\|\sigma)=\frac{1}{\alpha-1}\log
{\frac{\mathrm{Tr}\bigl(\rho^{\frac{\alpha}{2z}}\sigma^{\frac{1-\alpha}{z}}
\rho^{\frac{\alpha}{2z}}\bigr)^z}{\mathrm{Tr}\,\rho},}
\]
where $0<\alpha\ne 1$ and $z>0$. This family includes the {Umegaki} relative entropy in the
limit $\alpha\to 1$ and unifies the two most important quantum
versions of R\'enyi divergences, which have already an established operational
significance. Namely, the standard or Petz-type R\'enyi divergences,  based on the Petz
quasi-entropy \cite{petz1985quasi,petz1986quasi}, are obtained for $z=1$,  whereas $z=\alpha$ gives
the sandwiched R\'enyi divergences introduced in \cite{mullerlennert2013onquantum, wilde2014strong}.
These two quantities have important applications, in particular, in describing asymptotic error bounds in
binary quantum state discrimination problems; see \cite{audenaert2007discriminating,jaksic2012quantum,
mosonyi2015quantum,mosonyi2015twoapproaches,nussbaum2009thechernoff} and references therein.

A fundamental property of any information measure is the data processing inequality (DPI), that
is, such a quantity must be non-increasing under quantum channels that are defined as completely
positive and trace-preserving maps. For the {$\alpha$-$z$-R\'enyi} divergences, this means
that we require
\[
D_{\alpha,z}(\Phi(\rho)\|\Phi(\sigma))\le D_{\alpha,z}(\rho\|\sigma)
\]
for any quantum channel $\Phi$ and any pair of states $\rho,\sigma$. This requirement
leads to a restriction on the values of the parameters. For the Petz-type divergences,
 the convexity theorem of Ando \cite{ando1979concavity} shows that  $\alpha$ is restricted to the interval
$(0,1)\cup (1,2]$. The DPI in the sandwiched case holds for $\alpha\in
[1/2,1)\cup(1,\infty)$, as was proved in  \cite{frank2013monotonicity}  and independently in
\cite{beigi2013sandwiched}. The proof in \cite{beigi2013sandwiched} was given only for
the case $\alpha>1$, but is interesting for us, because it utilizes complex interpolation
methods and, as later observed in \cite{mullerhermes2017monotonicity}, it shows that DPI holds also for
maps that are trace-preserving and positive but not necessarily completely positive. Remarkably, by taking
the limit $\alpha\searrow 1$, this proves the same property also for the relative entropy.

The problem of finding the values of $(\alpha,z)$ for which the DPI holds was posed and studied
already in \cite{audenaert2015alpha}. The conjecture stated in
\cite{carlen2018inequalities} suggests the bounds
\[
0<\alpha<1, \quad \max\{\alpha,1-\alpha\}\le z\qquad \text{or}\qquad \alpha>1,\quad
\max\{{\alpha/2},\alpha-1\}\le z\le \alpha.
\]
This conjecture was finally proved in \cite{zhang2020fromwyd}.

Quantum divergences and their properties were mostly studied in the finite dimensional
case, but there is a growing interest in extensions beyond this setting. Von Neumann
algebras form a suitable framework for such extensions, since they are general enough to
cover many finite or infinite dimensional situations and there are also strong mathematical
tools and techniques available to deal with them. Some important quantum divergences were
already extended to this setting, notably the Araki relative entropy \cite{araki1976relative}. The
quasi-entropies were first introduced for states of von Neumann algebras in
\cite{petz1985quasi} and studied in detail in \cite{hiai2018quantum}. The definition of sandwiched
R\'enyi divergences in this setting, based on non-commutative $L^p$-spaces,  was
introduced independently in \cite{berta2018renyi}  using Araki--Masuda's $L^p$-spaces and
in \cite{jencova2018renyi,jencova2021renyi} using Kosaki's $L^p$-spaces. A number of
properties of the divergences including DPI was proved in these works and it was shown in
\cite{jencova2018renyi} that the two definitions are indeed equivalent. See \cite{hiai2021quantum} for
an overview of both the Petz-type and the sandwiched R\'enyi divergences in the
von Neumann algebra setting.

Sufficiency of quantum channels with respect to a set of states was introduced by Petz
\cite{petz1986sufficient,petz1988sufficiency}. This definition was inspired by sufficient statistics
originated in \cite{halmos1949application,kullback1951oninformation}
in classical statistical models. The quantum counterpart in Petz's definition is that all
the states in the given set can be recovered by another channel. It is an important result
that this property is characterized by equality in the DPI for the relative entropy, as well
as for the transition probability (that is basically the Petz-type R\'enyi divergence for
$\alpha=1/2$). This result has been generalized and extended to a large class of
divergences and other quantities satisfying the DPI, in both finite and infinite
dimensional settings. In particular, this property 
for the standard R\'enyi divergence was proved in \cite{hiai2017different} for $\alpha$ in the
interval $(0,1)\cup(1,2)$ and in \cite{jencova2017preservation, jencova2018renyi, jencova2021renyi} for
the sandwiched version with $\alpha\in(1/2,1)\cup(1,\infty)$. See also \cite{hiai2021quantum} for a
summary of these results in von Neumann algebras.

The $\alpha$-$z$-R\'enyi divergences $D_{\alpha,z}$ in von Neumann algebras were very recently
introduced in \cite{kato2023aremark}, {based on the Haagerup $L^p$-spaces}, and a number of
their properties was subsequently proved in \cite{kato2023onrenyi}. In particular, most of the important
properties including the DPI (with respect to normal positive unital maps), monotonicity in the parameter $z$
and a variational expression were proved in the case $0<\alpha<1$. In the infinite dimensional
$\mathcal{B}(\mathcal{H})$ case, $D_{\alpha,z}$ and their properties were studied also in
\cite{mosonyi2023thestrong}.  In the von Neumann algebra setting, an alternative definition for $\alpha\in
[0,1]$ and $z\ge 1/2$ was given in \cite{furuya2023monotonic}, where the DPI with respect
to completely positive normal unital maps was also proved for $z>1$.

Inspired by the papers \cite{kato2023aremark, kato2023onrenyi}, our aim in the present
work is to complete their results on the properties of the $\alpha$-$z$-R\'enyi divergences.
Section \ref{Sec-2} of the present paper is a preliminary on $D_{\alpha,z}$, where we improve the
variational expression for $0<\alpha<1$ in \cite{kato2023onrenyi} and finish the proof of that for $\alpha>1$.
In Sec.~\ref{Sec-3} we prove the DPI, in the same bounds as in the finite dimensional case. Similarly to
Zhang's approach in \cite{zhang2020fromwyd}, we apply the variational expressions to do this,
but, not having the techniques of matrix analysis at our disposal,
we use their relation to the sandwiched R\'enyi divergences for the proof. This allows us
to show the DPI with respect to normal positive unital maps, similarly to the case $0<\alpha<1$
in \cite{kato2023onrenyi} and for the sandwiched R\'enyi divergences in
\cite{jencova2018renyi,jencova2021renyi}. As a consequence, we prove martingale convergence
theorems for $D_{\alpha,z}$. In Sec.~\ref{Sec-4} we next show the most significant result that for the
values of $(\alpha,z)$ inside the DPI bounds, equality in the DPI implies sufficiency of the channel,
where by channel we here mean a normal 2-positive unital map. The proofs again use the variational
expressions and their relation to the sandwiched R\'enyi divergences, together with the properties of
conditional expectations extended to the Haagerup $L^p$-spaces.

We then turn to monotonicity in the parameters $z$ and $\alpha$. In Sec.~\ref{Sec-5} we first show that
in the case of finite von Neumann algebras,  monotonicity in $z$ holds for all values of the parameter. 
In the case of general von Neumann algebras, for $\alpha>1$ we prove monotonicity in $z$
on the interval $[\alpha/2,\infty)$ using a complex interpolation method
based on Kosaki's $L^p$-spaces. In Sec.~\ref{Sec-6} we further prove monotonicity in the parameter
$\alpha$, which seems to be new even in the finite dimensional case, using a real analysis method
for $0<\alpha<1$ and a complex interpolation method for $1<\alpha\le2z$. Finally, we show the limits
to the normalized relative entropy as $\alpha\searrow 1$ and $\alpha\nearrow 1$.

Appendices A--C are brief summaries on the Haagerup $L^p$-spaces, Haagerup's reduction theorem,
and Kosaki's interpolation $L^p$-spaces together with some technical lemmas, which are
utilized in the main body.

\section{Preliminaries}\label{Sec-2}

Let $\Me$ be a von Neumann algebra and let $\Me^+$ be the cone of positive elements in $\Me$. We
denote the predual of $\Me$ by $\Me_*$ and its positive part (consisting of normal positive
functionals on $\Me$) by $\Me_*^+$. For $\psi\in \Me_*^+$, we will denote by $s(\psi)$ the support
projection  of $\psi$. We denote the unit of von Neumann algebras by $\1$.

For $0< p\le \infty$, let $L^p(\Me)$ be the \emph{Haagerup $L^p$-space}
\cite{haagerup1979lpspaces,terp1981lpspaces} over $\Me$ and let $L^p(\Me)^+$ its positive cone. We
{note that $L^\infty(\Me)=\Me$, and use the identification
$\Me_*\ni \psi \leftrightarrow h_\psi\in L^1(\Me)$ (an order isomorphism)},
so that the $\Tr$-functional on $L^1(\Me)$ is defined by $\Tr\,h_\psi:=\psi(\1)$ for $\psi\in \Me_*$. It this way,
$\Me_*^+$ is identified with the positive cone $L^1(\Me)^+$; in particular, the normal states
on $\Me$ are identified with the elements $h\in L^1(\Me)^+$ with $\Tr\,h=1$. Let $\|a\|_p$ denote the
Haagerup (quasi-)norm of $a\in L^p(\Me)$, $0<p\le\infty$. Precise definitions and further details on
the spaces $L^p(\Me)$ can be found in \cite[Chap.~9]{hiai2021lectures}, or in the notes
\cite{terp1981lpspaces}. A short summary on the Haagerup $L^p$-spaces and some technical results
that will be used below can be found in Appendix \ref{Sec-A}.

In \cite{kato2023aremark, kato2023onrenyi}, the
$\alpha$-$z$-R\'enyi divergence for $\psi,\varphi\in \mathcal M_*^+$  was defined as follows:

\begin{defi}\label{defi:renyi} Let $\psi,\varphi\in \Me_*^+$, $\psi\ne 0$, and let
$\alpha,z>0$, $\alpha\ne 1$. The \emph{$\alpha$-$z$-R\'enyi divergence} is defined as 
\[
D_{\alpha,z}(\psi\|\varphi):=\frac1{\alpha-1}\log
\frac{Q_{\alpha,z}(\psi\|\varphi)}{\psi(\1)},
\]
where
\[
Q_{\alpha,z}(\psi\|\varphi):=\begin{dcases} \Tr
\left(h_\varphi^{\frac{1-\alpha}{2z}}h_\psi^{\frac{\alpha}{z}}h_\varphi^{\frac{1-\alpha}{2z}}\right)^z, &
\text{if } 0<\alpha<1,\\[0.3em]
\|x\|_z^z, & \text{if } \alpha>1 \text{ and }
h_\psi^{\frac{\alpha}{z}}=h_\varphi^{\frac{\alpha-1}{2z}}xh_\varphi^{\frac{\alpha-1}{2z}}
\ \text{with}\\ & x\in s(\varphi)L^z(\Me)s(\varphi),\\[0.3em]
\infty,& \text{otherwise}.
\end{dcases}
\]
\end{defi}

\begin{remark}\label{remark:defi}
As mentioned in Appendix \ref{Sec-A}, the Haagerup $L^p(\Me)$ spaces are constructed inside the
crossed product $\cR:=\Me\rtimes_{\sigma^\omega}\bR$ by the modular automorphism group
$\sigma^\omega$ for a faithful normal semi-finite weight $\omega$ on $\Me$. When two such weights
$\omega_j$, $j=0,1$, are given, we have the Haagerup $L^p$-spaces $L^p(\Me)^{(j)}$ constructed in
$\cR_j:=\Me\rtimes_{\sigma^{\omega_j}}\bR$ with the canonical trace $\tau_j$. Then there is an
isomorphism $\kappa:\cR_0\to\cR_1$ (extended to $\kappa:\widetilde\cR_0\to\widetilde\cR_1$, the spaces
of $\tau_j$-measurable operators) such that $\tau_0\circ\kappa^{-1}=\tau_1$,
$\kappa\bigl(L^p(\Me)^{(0)}\bigr)=L^p(\Me)^{(1)}$ for $0<p\le\infty$, and $\Tr_0\circ\kappa^{-1}=\Tr_1$ on
$L^1(\Me)^{(1)}$, where $\Tr_j$ is the $\Tr$-functional on $L^1(\Me)^{(j)}$. (For more details on the isomorphism
$\kappa$, see, e.g., \cite[Remark 9.10]{hiai2021lectures}.) Hence we note that the above definition
of $D_{\alpha,z}(\psi\|\ffi)$ is indeed independent of the choice of the Haagerup $L^p$-spaces. From this
fact, we also note that for $\psi,\ffi\in\Me_*^+$, $\psi\ne0$, with $s(\psi),s(\ffi)\le e$, where $e$ is a projection
in $\Me$, the $\alpha$-$z$-R\'enyi divergence $D_{\alpha,z}(\psi\|\ffi)$ is the same as that defined when
$\psi$ and $\ffi$ regarded as functionals on the reduced von Neumann algebra $e\Me e$.
\end{remark}

\medskip
In the case $\alpha>1$, the following alternative form will be useful.

\begin{lemma}[\mbox{\cite[Lemma 7]{kato2023onrenyi}}]\label{lemma:renyi_2z}
Let $\alpha>1$ and $\psi,\varphi\in \Me_*^+$, $\psi\ne0$. Then $Q_{\alpha,z}(\psi\|\varphi)<\infty$
if and only if there is some $y\in L^{2z}(\Me)s(\varphi)$ such that 
\[
h_\psi^{\frac{\alpha}{2z}}=yh_\varphi^{\frac{\alpha-1}{2z}}.
\]
Moreover, in this case, such $y$ is unique and we have
$Q_{\alpha,z}(\psi\|\varphi)=\|y\|_{2z}^{2z}$. 
\end{lemma}

The \emph{standard} or \emph{Petz-type R\'enyi divergence}
\cite{petz1985quasi,hiai2018quantum,hiai2021quantum} is contained in this range as
$D_\alpha(\psi\|\varphi)=D_{\alpha,1}(\psi\|\varphi)$. Also, the \emph{sandwiched R\'enyi divergence} is
obtained as $\tilde D_\alpha(\psi\|\varphi)=D_{\alpha,\alpha}(\psi\|\varphi)$; see
\cite{berta2018renyi,hiai2021quantum,jencova2018renyi, jencova2021renyi} for some
alternative definitions and properties of $\tilde D_\alpha$. The definitions in
\cite{jencova2018renyi,jencova2021renyi} are based on \emph{Kosaki's interpolation $L^p$-spaces}
$L^p(\Me,\varphi)$ \cite{kosaki1984applications} with respect to $\varphi$. These spaces and the
complex interpolation method are briefly summarized in Appendix \ref{Sec-C}, and will be used
frequently in the present work. Here, we note that for $z=1$ and $z=\alpha$, the above
definition of $D_{\alpha,z}$ proposed in \cite{kato2023aremark, kato2023onrenyi} is compatible with the
previous alternative definitions of $D_\alpha$ and $\tilde D_\alpha$; see \cite[Theorem 3.6]{hiai2021quantum}
for $z=1$ and \cite[Lemma 9]{kato2023aremark} for $z=\alpha$.

As have already been done by Kato in \cite{kato2023onrenyi}, many of the properties of
$D_{\alpha,z}(\psi\|\varphi)$ are extended from the finite dimensional case into the general von Neumann
algebra case. The main purpose of this section is to prove the following
variational expressions. These expressions will provide
 an important tool for our work. When $0<\alpha<1$, in \cite{kato2023onrenyi} Kato showed \eqref{F-2.1} for
$z\ge\max\{\alpha,1-\alpha\}$ as well as the inequality $\le$ in \eqref{F-2.1} for all $z>0$. When $\alpha>1$,
only the inequality $\ge$ in \eqref{F-2.2} for all $\alpha>1$ and $z>0$ was shown in
\cite{kato2023onrenyi}.

\begin{theorem}[Variational expressions]\label{thm:variational} Let $\psi,\varphi\in \Me_*^+$, $\psi\ne 0$. 
\begin{enumerate}
\item[(i)] Let $0<\alpha<1$. For every $z>0$ we have
\begin{align}\label{F-2.1}
Q_{\alpha,z}(\psi\|\varphi)=\inf_{a\in \Me^{++}}\left\{\alpha
\Tr\left((a^{\frac12}h_\psi^{\frac{\alpha}{z}}a^{\frac12})^{\frac{z}{\alpha}}\right)+(1-\alpha)
\Tr\left((a^{-\frac12}h_\varphi^{\frac{1-\alpha}{z}}a^{-\frac12})^{\frac{z}{1-\alpha}}\right) \right\}.
\end{align}

\item[(ii)] Let $\alpha>1$, $z\ge\alpha/2$. Then
\begin{align}\label{F-2.2}
Q_{\alpha,z}(\psi\|\varphi)=\sup_{a\in \Me^+} \left\{\alpha
\Tr\left((a^{\frac12}h_\psi^{\frac{\alpha}{z}}a^{\frac12})^{\frac{z}{\alpha}}\right)-(\alpha-1)
\Tr\left((a^{\frac12}h_\varphi^{\frac{\alpha-1}{z}}a^{\frac12})^{\frac{z}{\alpha-1}}\right) \right\}.
\end{align}
 Moreover,  equality \eqref{F-2.2} holds for all $\alpha>1$ and $z>0$ such that
$D_{\alpha,z}(\psi\|\ffi)<\infty$. In particular, if $\psi\le\lambda\ffi$ for some $\lambda>0$, then \eqref{F-2.2} holds for all
$z\ge\alpha-1>0$.
\end{enumerate}
\end{theorem}

\begin{proof}
(i)\enspace
Let $0<\alpha<1$, $z>0$ and $\psi,\varphi\in\Me_*^+$. We need to prove the inequality $\ge$ in
\eqref{F-2.1}, since the inequality $\le$ is known in \cite[Theorem 1(vi)]{kato2023onrenyi}.

Let $r:=\max\bigl\{{\alpha\over z},{1-\alpha\over z}\bigr\}$, and for any $\eps>0$ define
$\psi_\eps,\varphi_\eps\in\Me_*^+$ by
\[
h_{\psi_\eps}:=(h_\psi^r+\eps h_\varphi^r)^{1/r},\qquad
h_{\varphi_\eps}:=(h_\varphi^r+\eps h_\psi^r)^{1/r}.
\]
Since
\[
h_{\psi_\eps}^r\ge h_\psi^r,\qquad h_{\psi_\eps}^r\ge\eps h_\varphi^r,\qquad
h_{\varphi_\eps}^r\ge h_\varphi^r,\qquad h_{\varphi_\eps}^r\ge\eps h_\psi^r,
\]
we have
\[
h_{\psi_\eps}^r\le(\eps^{-1}+\eps)h_{\varphi_\eps}^r,\qquad
h_{\varphi_\eps}^r\le(\eps^{-1}+\eps)h_{\psi_\eps}^r.
\]
Since ${\alpha\over rz},{1-\alpha\over rz}\le1$, it follows (see \cite[Lemma B.7]{hiai2021quantum} and
\cite[Lemma 3.2]{hiai2021connections}) that
\begin{align}\label{F-2.3}
h_{\psi_\eps}^{\alpha\over z}\ge h_\psi^{\alpha\over z},\qquad
h_{\varphi_\eps}^{1-\alpha\over z}\ge h_\varphi^{1-\alpha\over z},\qquad
(\eps^{-1}+\eps)^{-{\alpha\over rz}}h_{\varphi_\eps}^{\alpha\over z}
\le h_{\psi_\eps}^{\alpha\over z}\le
(\eps^{-1}+\eps)^{\alpha\over rz}h_{\varphi_\eps}^{\alpha\over z}.
\end{align}
By the second paragraph of the proof of \cite[Theorem 1(vi)]{kato2023onrenyi} we thus obtain
\[
Q_{\alpha,z}(\psi_\eps\|\varphi_\eps)=\inf_{a\in\Me^{++}}
\left\{\alpha\Tr\left((a^{1/2}h_{\psi_\eps}^{\alpha\over z}a^{1/2})^{z\over\alpha}\right)
+(1-\alpha)\Tr\left((a^{-{1\over2}}h_{\varphi_\eps}^{1-\alpha\over z}a^{-{1\over2}})^{z\over1-\alpha}
\right)\right\}.
\]
For any $a\in\Me^{++}$, since
$a^{1/2}h_{\psi_\eps}^{\alpha\over z}a^{1/2}\ge a^{1/2}h_\psi^{\alpha\over z}a^{1/2}$ by the first
inequality in \eqref{F-2.3}, it follows from Lemma \ref{lemma:order1} that
\[
\|a^{1/2}h_{\psi_\eps}^{\alpha\over z}a^{1/2}\|_{z\over\alpha}
\ge\|a^{1/2}h_\psi^{\alpha\over z}a^{1/2}\|_{z\over\alpha},
\]
that is,
\[
\Tr\left((a^{1/2}h_{\psi_\eps}^{\alpha\over z}a^{1/2})^{z\over\alpha}\right)
\ge\Tr\left((a^{1/2}h_\psi^{\alpha\over z}a^{1/2})^{z\over\alpha}\right),
\]
and similarly by the second inequality in \eqref{F-2.3},
\[
\Tr\left((a^{-1/2}h_{\varphi_\eps}^{1-\alpha\over z}a^{-1/2})^{z\over1-\alpha}\right)
\ge\Tr\left((a^{-1/2}h_\varphi^{1-\alpha\over z}a^{-1/2})^{z\over1-\alpha}\right).
\]
Therefore, we obtain
\[
Q_{\alpha,z}(\psi_\eps\|\varphi_\eps)\ge\inf_{a\in\Me^{++}}
\left\{\alpha\Tr\left((a^{1/2}h_\psi^{\alpha\over z}a^{1/2})^{z\over\alpha}\right)
+(1-\alpha)\Tr\left((a^{-{1\over2}}h_\varphi^{1-\alpha\over z}a^{-{1\over2}})^{z\over1-\alpha}
\right)\right\}.
\]
Letting $\eps\searrow0$ with \cite[Lemma 6 and Theorem 1(iv)]{kato2023onrenyi} gives the desired
inequality $\ge$ in \eqref{F-2.1}.

\medskip
(ii)\enspace
It was proved in \cite[Theorem 2(vi)]{kato2023onrenyi} that the inequality $\ge$ in
\eqref{F-2.2} holds for all $\alpha>1$ and $z>0$. We now prove the opposite inequality. 

Assume first that $Q_{\alpha,z}(\psi\|\varphi)<\infty$, so that there is some $x\in
s(\varphi)L^z(\Me)^+s(\varphi)$ such that
$h_\psi^{\frac{\alpha}{z}}=h_\varphi^{\frac{\alpha-1}{2z}}xh_\varphi^{\frac{\alpha-1}{2z}}$. Plugging this
into the right-hand side of \eqref{F-2.2}, we obtain
\begin{align}
&\sup_{a\in \Me^+} \left\{\alpha
\Tr\left((a^{\frac12}h_\psi^{\frac{\alpha}{z}}a^{\frac12})^{\frac{z}{\alpha}}\right)-(\alpha-1)
\Tr\left((a^{\frac12}h_\varphi^{\frac{\alpha-1}{z}}a^{\frac12})^{\frac{z}{\alpha-1}}\right) \right\}
\nonumber\\
&\quad=\sup_{a\in \Me^+} \left\{\alpha
\Tr\left((a^{\frac12}h_\varphi^{\frac{\alpha-1}{2z}}xh_\varphi^{\frac{\alpha-1}{2z}}
a^{\frac12})^{\frac{z}{\alpha}}\right)-(\alpha-1)
\Tr\left((a^{\frac12}h_\varphi^{\frac{\alpha-1}{z}}a^{\frac12})^{\frac{z}{\alpha-1}}\right) \right\}
\nonumber\\
&\quad=\sup_{a\in \Me^+} \left\{\alpha
\Tr\left((x^{\frac12}h_\varphi^{\frac{\alpha-1}{2z}}ah_\varphi^{\frac{\alpha-1}{2z}}
x^{\frac12})^{\frac{z}{\alpha}}\right)-(\alpha-1)
\Tr\left((h_\varphi^{\frac{\alpha-1}{2z}}a h_\varphi^{\frac{\alpha-1}{2z}}
)^{\frac{z}{\alpha-1}}\right)\right\}\nonumber\\
&\quad=\sup_{w\in L^{\frac{z}{\alpha-1}}(\Me)^+} \left\{\alpha
\Tr\left((x^{\frac12}wx^{\frac12})^{\frac{z}{\alpha}}\right)-(\alpha-1)
\Tr\left(w^{\frac{z}{\alpha-1}}\right)
\right\}, \label{F-2.4}
\end{align}
where we have used the fact that $\Tr\left((h^*h)^p\right)=\Tr\left((hh^*)^p\right)$ for
$p>0$, $h\in L^{\frac{p}{2}}(\Me)$, and Lemma \ref{lemma:cone}.
Putting $w=x^{\alpha-1}$ we get
\[
\sup_{w\in L^{\frac{z}{\alpha-1}}(\Me)^+} \left\{\alpha
\Tr\left((x^{\frac12}wx^{\frac12})^{\frac{z}{\alpha}}\right)-(\alpha-1)
\Tr\left(w^{\frac{z}{\alpha-1}}\right)
\right\}\ge \Tr(x^z)=\|x\|_z^z= Q_{\alpha,z}(\psi\|\varphi).
\]
This finishes the proof of \eqref{F-2.2} in the case that $Q_{\alpha,z}(\psi\|\varphi)<\infty$. Note that
this holds if $z\ge \alpha-1>0$ and $\psi\le \lambda\varphi$ for some $\lambda>0$. Indeed, since
$\frac{\alpha-1}{z}\in (0,1]$ by the assumption, we then have 
$h_\psi^{\frac{\alpha-1}{z}}\le \lambda^{\frac{\alpha-1}{z}}h_\varphi^{\frac{\alpha-1}{z}}$.
Hence by \cite[Lemma A.58]{hiai2021quantum} there is some $b\in \Me$ such that 
\[
h_\psi^{\frac{\alpha-1}{2z}}=bh_\varphi^{\frac{\alpha-1}{2z}},
\]
so that $h_\psi^{\alpha\over2z}=yh_\ffi^{\alpha-1\over2z}$, where
 $y:=h_\psi^{\frac{1}{2z}}b\in L^{2z}(\Me)$. By Lemma \ref{lemma:renyi_2z},
$Q_{\alpha,z}(\psi\|\varphi)=\|y\|_{2z}^{2z}<\infty$. This shows the  last assertion too.
(Note that this follows also from the order relations in \cite[Theorem 2(iii)]{kato2023onrenyi}.)

Next, assume that $z\ge\alpha/2$. For any $\eps>0$ define $\varphi_\eps\in\Me_*^+$ by
\[
h_{\varphi_\eps}:=\left(
h_\varphi^{\alpha-1\over z}+\eps h_\psi^{\alpha-1\over z}\right)^{z\over\alpha-1}.
\]
Then we have
\begin{align}\label{F-2.5}
h_\varphi^{\alpha-1\over z}\le h_{\varphi_\eps}^{\alpha-1\over z},\qquad
h_\psi^{\alpha-1\over z}\le\eps^{-1}h_{\varphi_\eps}^{\alpha-1\over z}.
\end{align}
From the last inequality we have $Q_{\alpha,z}(\psi\|\varphi)<\infty$ by a similar argument to
the last part of the above paragraph. It thus follows from the case proved above that
\begin{align*}
Q_{\alpha,z}(\psi\|\varphi_\varepsilon)&=\sup_{a\in \Me^+} \left\{\alpha
\Tr\left((a^{\frac12}h_\psi^{\frac{\alpha}{z}}a^{\frac12})^{\frac{z}{\alpha}}\right)
-(\alpha-1)\Tr\left((a^{\frac12}h_{\varphi_\eps}^{\frac{\alpha-1}{z}}
a^{\frac12})^{\frac{z}{\alpha-1}}\right)\right\}\\
&\le\sup_{a\in \Me^+} \left\{\alpha
\Tr\left((a^{\frac12}h_\psi^{\frac{\alpha}{z}}a^{\frac12})^{\frac{z}{\alpha}}\right)-(\alpha-1)
\Tr\left((a^{\frac12}h_\varphi^{\frac{\alpha-1}{z}}a^{\frac12})^{\frac{z}{\alpha-1}}\right)
\right\},
\end{align*}
where the inequality above follows from the first inequality in \eqref{F-2.5} and Lemma \ref{lemma:order1}.
Furthermore, since $z\ge\alpha/2$, from the lower semi-continuity of $Q_{\alpha,z}$ in
\cite[Theorem 2(iv)]{kato2023onrenyi} we have
\[
Q_{\alpha,z}(\psi\|\varphi)\le \liminf_{\varepsilon\searrow 0}
Q_{\alpha,z}(\psi\|\varphi_\varepsilon),
\]
so that the desired inequality is obtained.

\end{proof}


\begin{remark}\label{remark:variational}
In the finite dimensional setting, the variational expressions in \eqref{F-2.1} and \eqref{F-2.2} were proved
in \cite[Theorem 3.3]{zhang2020fromwyd} with no restriction on the parameter $z>0$. Also, in
\cite[Lemma 3.23]{mosonyi2023thestrong}, the variational expression in \eqref{F-2.2} was proved in the
infinite dimensional $\mathcal{B}(\mathcal{H})$ case with no restriction on $z$.
However, we need the assumption $z\ge\alpha/2$ to prove \eqref{F-2.2} of
Theorem \ref{thm:variational}(ii), which is though a part of the DPI bounds shown in
Theorem \ref{thm:dpi} below.
\end{remark}

\medskip
We finish this section by investigation of the properties of the variational expression for
$0<\alpha<1$, in the case when $\lambda^{-1}\ffi\le \psi\le \lambda \ffi$ for some
$\lambda>0$. This case will be denoted as $\psi\sim \ffi$. 

\begin{lemma}\label{lemma:variational_majorized}
Assume that $\psi\sim\ffi$ and let $0<\alpha<1$,
$\max\{\alpha,1-\alpha\}\le z$. Then the infimum in \eqref{F-2.1} of Theorem \ref{thm:variational}(i) is
attained at a unique element $\bar a\in \Me^{++}$. This element satisfies
\begin{align}
h_\psi^{\alpha\over2z}\bar ah_\psi^{\alpha\over2z}
&=\Bigl(h_\psi^{\alpha\over2z}h_\ffi^{1-\alpha\over z}h_\psi^{\alpha\over2z}\Bigr)^\alpha,
\label{eq:minimizer1}\\
h_\ffi^{1-\alpha\over2z}\bar a^{-1}h_\ffi^{1-\alpha\over2z}
&=\Bigl(h_\ffi^{1-\alpha\over2z}h_\psi^{\alpha\over
z}h_\ffi^{1-\alpha\over2z}\Bigr)^{1-\alpha}.
\label{eq:minimizer2}
\end{align}
\end{lemma}

\begin{proof} We may assume that $\varphi$ and hence also $\psi$ are faithful. Following
 the proof of \cite[Theorem 1(vi)]{kato2023onrenyi},  we may use the
assumptions and \cite[Lemma A.58]{hiai2021quantum} to show  that there are $b,c\in\cM$ such that
\begin{align}\label{eq:bc}
h_\ffi^{1-\alpha\over2z}
=b\Bigl(h_\ffi^{1-\alpha\over2z}h_\psi^{\alpha\over z}h_\ffi^{1-\alpha\over2z}\Bigr)^{1-\alpha\over2},\qquad
\Bigl(h_\ffi^{1-\alpha\over2z}h_\psi^{\alpha\over z}h_\ffi^{1-\alpha\over2z}\Bigr)^{1-\alpha\over2}
=ch_\ffi^{1-\alpha\over2z}.
\end{align}
With $\bar a:=bb^*\in\cM^{++}$ we have  $\bar a^{-1}=c^*c$ and $\bar a$ is indeed a minimizer of
\eqref{F-2.2}, so that
\begin{align}\label{eq:infimum}
Q_{\alpha,z}(\psi\|\ffi)
=\Big\|h_\psi^{\alpha\over2z}\bar ah_\psi^{\alpha\over2z}\Big\|_{z\over\alpha}^{z\over\alpha}
=\alpha\Big\|h_\psi^{\alpha\over2z}\bar ah_\psi^{\alpha\over2z}\Big\|_{z\over\alpha}^{z\over\alpha}
+(1-\alpha)\Big\|h_\ffi^{1-\alpha\over2z}\bar a^{-1}h_\ffi^{1-\alpha\over2z}\Big\|_{z\over
1-\alpha}^{z\over 1-\alpha}.
\end{align}
(See the proof of \cite[Theorem 1(vi)]{kato2023onrenyi} for more detail.)
We next observe that such a minimizer is unique. Indeed, suppose that the infimum is
attained  at some $a_1,a_2\in \Me^{++}$. Let $a_0:=(a_1+a_2)/2$. Since the map 
$L^{p}(\cM)\ni k\mapsto\|k\|_{p}^{p}$ is convex for any $p\ge 1$, we have 
\[
\Big\|h_\psi^{\alpha\over2z}a_0h_\psi^{\alpha\over2z}\Big\|_{z\over\alpha}^{z\over\alpha}
\le{1\over2}\biggl\{\Big\|h_\psi^{\alpha\over2z}a_1
h_\psi^{\alpha\over2z}\Big\|_{z\over\alpha}^{z\over\alpha}
+\Big\|h_\psi^{\alpha\over2z}a_2h_\psi^{\alpha\over2z}\Big\|_{z\over\alpha}^{z\over\alpha}\biggr\}.
\]
Moreover, since $a_0^{-1}\le(a_1^{-1}+a_2^{-1})/2$, using Lemma \ref{lemma:order1}
we have
\begin{align*}
\Big\|h_\ffi^{1-\alpha\over2z}a_0^{-1}h_\ffi^{1-\alpha\over2z}\Big\|_{z\over1-\alpha}^{z\over1-\alpha}
&\le\Big\|h_\ffi^{1-\alpha\over2z}\biggl({a_1^{-1}+a_2^{-1}\over2}\biggr)
h_\ffi^{1-\alpha\over2z}\Big\|_{z\over1-\alpha}^{z\over1-\alpha} \\
&\le{1\over2}\biggl\{\Big\|h_\ffi^{1-\alpha\over2z}a_1^{-1}
h_\ffi^{1-\alpha\over2z}\Big\|_{z\over1-\alpha}^{z\over1-\alpha}
+\Big\|h_\ffi^{1-\alpha\over2z}a_2^{-1}
h_\ffi^{1-\alpha\over2z}\Big\|_{z\over1-\alpha}^{z\over1-\alpha}\biggr\}.
\end{align*}
Hence the assumption of $a_1,a_2$ being a minimizer gives
\[
\Big\|h_\ffi^{1-\alpha\over2z}a_0^{-1}h_\ffi^{1-\alpha\over2z}\Big\|_{z\over1-\alpha}
=\Big\|h_\ffi^{1-\alpha\over2z}\biggl({a_1^{-1}+a_2^{-1}\over2}\biggr)
h_\ffi^{1-\alpha\over2z}\Big\|_{z\over1-\alpha},
\]
which implies that $a_0^{-1}={a_1^{-1}+a_2^{-1}\over2}$, as easily verified by Lemma \ref{lemma:order1}
again. This gives $a_1=a_2$ from the strict operator convexity of $t\mapsto t^{-1}$ on
$(0,\infty)$, as easily seen since $\bigl({a_1+a_2\over2}\bigr)^{-1}={a_1^{-1}+a_2^{-1}\over2}$ gives
$\Bigl({\1+a_1^{-1/2}a_2a_1^{-1/2}\over2}\Bigr)^{-1}={\1+a_1^{-1/2}a_2a_1^{-1/2}\over2}$.

The equality  \eqref{eq:minimizer2} is obvious from the second equality in \eqref{eq:bc} and
$\bar a^{-1}=c^*c$. Since $Q_{\alpha,z}(\psi\|\ffi)=Q_{1-\alpha,z}(\ffi\|\psi)$, we see by uniqueness that
the minimizer of the infimum expression for $Q_{1-\alpha,z}(\ffi\|\psi)$ (instead of \eqref{eq:infimum}) is
$\bar a^{-1}$ (instead of $\bar a$). This says that \eqref{eq:minimizer1} is the equality corresponding to
\eqref{eq:minimizer2} when $\psi,\ffi,\alpha$ are replaced with $\ffi,\psi,1-\alpha$, respectively. 
\end{proof}

To make the next lemma more readable, we will use the following notations:
\[
p:={z\over\alpha},\qquad r:={z\over 1-\alpha},\qquad
\xi_p(a):=h_\psi^{1\over2p}ah_\psi^{1\over 2p},\qquad
\eta_r(a):=h_\ffi^{1\over2r}a^{-1}h_\ffi^{1\over 2r}.
\]
We will also denote the function under the infimum in the variational expression in
Theorem \ref{thm:variational}(i) by $f$, that is,
\begin{equation}\label{func-variational}
f(a)=\alpha\|\xi_p(a)\|_p^p
+(1-\alpha)\|\eta_r(a)\|_r^r,\qquad a\in \Me^{++}.
\end{equation}
When $p\in(1,\infty)$, recall that $L^p(\Me)$ is uniformly convex (see
\cite{haagerup1979lpspaces}, \cite[Theorem 4.2]{kosaki1984applications}), so that the norm
$\|\cdot\|_p$ is uniformly Fr\'echet differentiable (see, e.g.,
\cite[Part 3, Chap.~II]{beauzamy1982introduction}). Hence if $p,r>1$, $a\mapsto\|\xi_p(a)\|_p^p$
and $a\mapsto\|\eta_r(a)\|_r^r$ are Fr\'echet differentiable on $\Me^{++}$. Since differentiability of
these functions is obvious when $p=1$ and $r=1$, we see that the function $f$ is Fr\'echet differentiable
on $\Me^{++}$ for any $p,r\ge1$, whose Fr\'echet derivative at $a$ will be denoted by $\nabla f(a)$.

\begin{lemma}\label{lemma:variational_majorized2}
Assume that $\psi\sim\ffi$  and let $0<\alpha<1$, $\max\{\alpha,1-\alpha\}\le z$. {Let
$\bar a\in\Me^{++}$ be as given in Lemma \ref{lemma:variational_majorized}.} If $p>1$, then for
every $C\ge Q_{\alpha,z}(\psi\|\ffi)$ and $\eps>0$ there is some $\delta>0$ such that whenever
{$b\in\Me^{++}$ satisfies}
$\|\xi_p(b)\|^p_p\le C$ and $\|\xi_p(b)-\xi_p(\bar a)\|_p\ge \eps$, we have
\[
f(b)-Q_{\alpha,z}(\psi\|\ffi)\ge \delta.
\]
A similar statement holds if $r>1$ too.
\end{lemma}

\begin{proof} By assumptions, $p,r\ge 1$.  For
$a,b\in \Me^{++}$ and $s\in (0,1/2)$, we have
\begin{align*}
\|\xi_p(sb+(1-s)a)\|_p^p&=\|s\xi_p(b)+(1-s)\xi_p(a)\|_p^p\\
&=\Big\|(1-2s)\xi_p(a)+2s\frac12(\xi_p(a)+\xi_p(b))\Big\|_p^p\\
&\le (1-2s)\|\xi_p(a)\|_p^p+2s\Big\|\frac12(\xi_p(a)+\xi_p(b))\Big\|_p^p.
\end{align*}
Similarly,
\[
\|\eta_r(sb+(1-s)a)\|_r^r\le
(1-2s)\|\eta_r(a)\|_r^r+2s\Big\|\frac12(\eta_r(a)+\eta_r(b))\Big\|_r^r,
\]
where we have also used the fact that $(t a+(1-t)b)^{-1}\le t
a^{-1}+(1-t)b^{-1}$ for $t\in (0,1)$ and Lemma \ref{lemma:order1}. It follows that
\begin{align*}
&\<\nabla f(a),b-a\>\\
&\quad=\lim_{s\searrow0} {1\over s}[ f(sb+(1-s)a)-f(a)]\\
&\quad\le 2\alpha\biggl[\Big\|\frac12(\xi_p(a)+\xi_p(b))\Big\|_p^p-\|\xi_p(a)\|_p^p\biggr]
+2(1-\alpha)\biggl[\Big\|\frac12(\eta_r(a)+\eta_r(b))\Big\|_r^r-\|\eta_r(a)\|_r^r\biggr]\\
&\quad=f(b)-f(a)-2\Biggl\{\alpha\biggl[\frac12 \|\xi_p(a)\|_p^p
+\frac12 \|\xi_p(b)\|_p^p-\Big\|\frac12(\xi_p(a)+\xi_p(b))\Big\|_p^p\biggr]\\
&\qquad\qquad\qquad\qquad\quad+(1-\alpha)\biggl[\frac12 \|\eta_r(a)\|_r^r+\frac12 \|\eta_r(b)\|_r^r
-\Big\|\frac12(\eta_r(a)+\eta_r(b))\Big\|_r^r\biggr]\Biggr\}.
\end{align*}
Since $p,r\ge 1$, both terms in brackets in the last expression above are non-negative.
Let $\bar a\in \Me^{++}$ be the minimizer as in Lemma \ref{lemma:variational_majorized};
then $f(\bar a)=Q_{\alpha,z}(\psi\|\varphi)$ and $\nabla f(\bar a)=0$, so that we get
\begin{align*}
f(b)-Q_{\alpha,z}(\psi\|\varphi)\ge 2\alpha\biggl[\frac12 \|\xi_p(\bar a)\|_p^p+\frac12 \|\xi_p(b)\|_p^p
-\Big\|\frac12(\xi_p(\bar a)+\xi_p(b))\Big\|_p^p\biggr].
\end{align*}
Since $L^p(\Me)$ is uniformly convex, note (see, e.g., \cite[Theorem 3.7.7]{zalinescu2002convex})
that the function $h\mapsto \|h\|_p^p$ is uniformly convex on each bounded subset of $L^p(\Me)$.
Hence for each $C>0$ and $\varepsilon>0$ there is some $\delta>0$ such that for every $h,k$
with $\|h\|_p^p,\|k\|_p^p\le C$ and $\|h-k\|_p\ge \varepsilon$, we have
\[
\frac12\|h\|_p^p+\frac12\|k\|_p^p-\Big\|\frac12(h+k)\Big\|_p^p\ge \delta
\]
(see \cite[p.~288, Exercise 3.3]{zalinescu2002convex}). Since $\|\xi_p(\bar
a)\|_p^p=Q_{\alpha,z}(\psi\|\varphi)$ by \eqref{eq:infimum}, this proves the statement.
The proof in the case $r>1$ is similar. 
\end{proof}

\section{Data processing inequality and martingale convergence}\label{Sec-3}

Let  $\gamma: \Ne\to \Me$ be a normal positive unital map. Then the predual of $\gamma$ defines a 
positive linear map $\gamma_*: L^1(\Me)\to L^1(\Ne)$ that preserves the $\Tr$-functional,
acting as
\[
L^1(\Me)\ni h_\rho\mapsto h_{\rho\circ\gamma} \in L^1(\Ne).
\]
The support
of $\gamma$ will be denoted by $s(\gamma)$. Recall that this is defined as the smallest projection
$e\in \Ne$ such that $\gamma(e)=\1$ and in this case, $\gamma(a)=\gamma(eae)$ for any $a\in
\Ne$. For any $\rho\in \Me_*^+$ we clearly have
$s(\rho\circ\gamma)\le s(\gamma)$, with equality if $\rho$ is faithful. 
It follows that $\gamma_*$ maps $L^1(\Me)$ to $s(\gamma)L^1(\Ne)s(\gamma)\equiv
L^1(s(\gamma)\Ne s(\gamma))$.  For any $\rho\in \Me_*^+$, $\rho\ne 0$, the map
\[
\gamma_0: s(\rho\circ\gamma)\Ne s(\rho\circ\gamma)\to s(\rho)\Me s(\rho),
\qquad a\mapsto s(\rho)\gamma(a)s(\rho)
\]
is a faithful normal positive unital (i.e., $\gamma_0(s(\rho\circ\gamma))=s(\rho)$) map;
see \cite[Remark 6.7]{hiai2021quantum}. Moreover, for any $\ffi\in \Me_*^+$ such that
$s(\ffi)\le s(\rho)$, we have for any $a\in \Ne$,
\[
\ffi(\gamma_0(s(\rho\circ\gamma)as(\rho\circ\gamma)))=\ffi(s(\rho)\gamma(a)s(\rho))=\ffi(\gamma(a)).
\]
Replacing $\gamma$ by $\gamma_0$ and $\rho$ by the restriction
$\rho|_{s(\rho)\Me s(\rho)}$, we may  assume that both $\rho$ and $\rho\circ\gamma$ are faithful,
as far as we are concerned with $\ffi\in\Me_*^+$ and $\ffi\circ\gamma\in\Ne_*^+$ with $s(\ffi)\le s(\rho)$.

The \emph{Petz dual} of $\gamma$ with respect to $\rho\in \Me_*^+$ is a map
$\gamma_\rho^*:\Me\to \Ne$, introduced in \cite{petz1988sufficiency} when $\rho$ and
$\rho\circ\gamma$ are faithful (hence so is $\gamma$ as well). It was proved that
$\gamma_\rho^*$ is again normal, positive and unital, and in addition, it is $n$-positive whenever
$\gamma$ is. More generally, even though none of $\rho$, $\rho\circ\gamma$ and
$\gamma$ is faithful, letting $e:=s(\rho)$ and $e_0:=s(\rho\circ\gamma)$, we can use the restriction
$\gamma_0$ as mentioned above to define the Petz dual
$\gamma^*_\rho: e\Me e\to e_0\Ne e_0$. As explained in \cite{jencova2018renyi} (also in
\cite[Lemma 8.3]{hiai2021quantum}), in this general setting, $\gamma^*_\rho$ is determined by
the equality
\[
h_{\rho\circ\gamma}^{1/2}\gamma_\rho^*(a)h_{\rho\circ\gamma}^{1/2}
=\gamma_*\bigl(h_\rho^{1/2}ah_\rho^{1/2}\bigr),
\qquad a\in e\Me e,
\]
equivalently,
\begin{equation}\label{eq:petzdual}
(\gamma^*_\rho)_*\bigl(h_{\rho\circ\gamma}^{1/2}bh_{\rho\circ\gamma}^{1/2}\bigr)
=h_\rho^{1/2}\gamma(b)h_\rho^{1/2},\qquad b\in\Ne^+,
\end{equation}
where $\gamma_*$ and $(\gamma^*_\rho)_*$ are the predual maps of $\gamma$ and $\gamma^*_\rho$,
respectively.
We also have
\begin{equation}\label{eq:petzdual2}
\rho\circ\gamma\circ\gamma^*_\rho=\rho,\qquad (\gamma_\rho^*)_{\rho\circ\gamma}^*=\gamma_0.
\end{equation}
In the special case where $\gamma$ is the
inclusion map $\gamma: \Ne\hookrightarrow \Me$ for a subalgebra $\Ne\subseteq \Me$ and
$\rho\in\Me_*^+$ is faithful, the Petz dual is the \emph{generalized conditional expectation}
$\cE_{\Ne,\rho}:\Me\to \Ne$, as introduced in \cite{accardi1982conditional}; see, e.g.,
\cite[Proposition 6.5]{hiai2021quantum}. Hence $\cE_{\Ne,\rho}$ is a normal completely positive
unital map with range in $\Ne$ and such that 
\[
\rho\circ \cE_{\Ne,\rho}=\rho.
\]

\subsection{Data processing inequality}\label{Sec-3.1}

In this subsection we prove the \emph{data processing inequality} (\emph{DPI}) for $D_{\alpha,z}$
with respect to normal positive unital maps. For standard R\'enyi divergence $D_\alpha$,
that is, for $z=1$, {the} DPI is known to hold for $\alpha\in (0,1)\cup (1,2]$ under stronger positivity
assumptions \cite{hiai2018quantum}. In the case of the sandwiched divergences $\tilde D_\alpha$
with {$\alpha\in[1/2,1)\cup(1,\infty)$,} the DPI was proved in
\cite{jencova2018renyi,jencova2021renyi}; see also \cite{berta2018renyi} for an alternative proof
in the case when the maps are assumed completely positive. In the finite dimensional case, the DPI for
$D_{\alpha,z}$ under completely positive maps was proved in \cite{zhang2020fromwyd}, for $\alpha,z$
in the range specified as in Theorem \ref{thm:dpi} below.

The first part of the next lemma was essentially shown in
\cite[Proposition 3.12]{jencova2018renyi}, while we give the proof for the convenience of the reader.

\begin{lemma}\label{lemma:pcontraction} Let $\gamma:\Ne\to \Me$ be a normal positive unital map.
Let $\rho\in \Me_*^+$, {$\rho\ne0$,} $e:=s(\rho)$ and $e_0:=s(\rho\circ\gamma)$. For any
$p\ge 1$, the map $\gamma^*_{\rho,p}:L^p(e_0 \Ne e_0)\to L^p(e\Me e)$, determined by
\begin{align}\label{gamma-rho-p1}
\gamma^*_{\rho,p}\Bigl(h_{\rho\circ\gamma}^{1\over 2p}bh^{1\over2p}_{\rho\circ\gamma}\Bigr)
=h_\rho^{1\over 2p}\gamma(b) h_\rho^{1\over 2p},\qquad b\in \Ne,
\end{align}
is a contraction such that 
\begin{align}\label{gamma-rho-p2}
(\gamma^*_\rho)_*\Bigl(h_{\rho\circ\gamma}^{p-1\over 2p}k
h_{\rho\circ\gamma}^{p-1\over2p}\Bigr)
=h_\rho^{p-1\over 2p}\gamma^*_{\rho,p}(k)h_\rho^{p-1\over 2p},\qquad
k\in L^p(e_0 \Ne e_0).
\end{align}
Moreover, if $\rho_n\in \Me_*^+$ are such that $s(\rho)\le s(\rho_n)$ and
$\|\rho_n-\rho\|_1\to 0$, then for any $k\in L^p(e_0 \Ne e_0)$ we have
$\gamma^*_{\rho_n,p}(k)\to \gamma^*_{\rho,p}(k)$ in $L^p(\Me)$.

\end{lemma}

\begin{proof}
We use Kosaki's symmetric $L^p$-spaces $L^p(e_0\Ne e_0,\rho\circ\gamma)$ and
$L^p(e\Me e,\rho)$ (see \eqref{F-C.4} in Appendix \ref{Sec-C}). The map
$(\gamma_\rho^*)_*:L^1(e_0\Ne e_0)\to L^1(e\Me e)$ is contractive
with respect to $\|\cdot\|_1$. Its restriction to $h_{\rho\circ\gamma}^{1/2}\Ne h_{\rho\circ\gamma}^{1/2}$
($\subseteq L^1(e_0\Ne e_0)$) is given by \eqref{eq:petzdual}, which is also contractive with respect
to $\|\cdot\|_{\infty,\rho\circ\gamma}$ and $\|\cdot\|_{\infty,\rho}$. Hence it follows from the Riesz--Thorin
theorem that $(\gamma_\rho^*)_*$ is a contraction from $L^p(e_0\Ne e_0,\rho\circ\gamma)$ to
$L^p(e\Me e,\rho)$ for any $p\in(1,\infty)$. By \eqref{F-C.4} note that we have isometric isomorphisms
\begin{align*}
k\in L^p(e_0\Ne e_0)&\mapsto h_{\rho\circ\gamma}^{p-1\over2p}kh_{\rho\circ\gamma}^{p-1\over2p}
\in L^p(e_0\Ne e_0,\rho\circ\gamma), \\
h\in L^p(e\Me e)&\mapsto h_\rho^{p-1\over2p}hh_\rho^{p-1\over2p}\in L^p(e\Me e,\rho).
\end{align*}
Hence we can define a contraction $\gamma_{\rho,p}^*:L^p(e_0\Ne e_0)\to L^p(e\Me e)$ by
\eqref{gamma-rho-p2}. Then, for $k=h_{\rho\circ\gamma}^{1\over2p}bh_{\rho\circ\gamma}^{1\over2p}$
with $b\in\Ne$ we have
\[
h_\rho^{p-1\over2p}\gamma_{\rho,p}^*\Bigl(h_{\rho\circ\gamma}^{1\over2p}b
h_{\rho\circ\gamma}^{1\over2p}\Bigr)h_\rho^{p-1\over2p}
=(\gamma_\rho^*)_*\Bigl(h_{\rho\circ\gamma}^{1\over2}bh_{\rho\circ\gamma}^{1\over2}\Bigr)
=h_\rho^{1\over2}\gamma(b)h_\rho^{1\over2},
\]
so that \eqref{gamma-rho-p1} holds.
Since $h_{\rho\circ\gamma}^{1\over 2p}\Ne h^{1\over
2p}_{\rho\circ\gamma}$ is dense in $L^p(e_0 \Ne e_0)$, this proves the first part of the statement.

Let $\rho_n$ be a sequence as required and let $k\in L^p(e_0 \Ne e_0)$. By the
assumptions on the supports, $\gamma^*_{\rho_n,p}$ is well defined on $k$ for all $n$.
Further, we may assume that $k=h_{\rho\circ\gamma}^{1\over 2p}b
h_{\rho\circ\gamma}^{1\over 2p}$ for some $b\in \Ne$, since the set of such elements is
dense in $L^p(e_0 \Ne e_0)$ and all the maps are contractions. 
Put $k_n:=h_{\rho_n\circ\gamma}^{1\over 2p}b
h_{\rho_n\circ\gamma}^{1\over 2p}$, then we have
\[
\gamma^*_{\rho,p}(k)={h_\rho^{1\over2p}\gamma(b)h_\rho^{1\over2p},}\qquad
\gamma^*_{\rho_n,p}(k_n)= {h_{\rho_n}^{1\over2p}\gamma(b)h_{\rho_n}^{1\over2p},}
\]
and we have $k_n\to k$ in $L^p(\Ne)$ and  $\gamma^*_{\rho_n,p}(k_n)\to
\gamma^*_{\rho,p}(k)$ in $L^p(\Me)$. Indeed, this follows by the H\"older inequality and the
continuity of the map $L^1(\Me)^+\ni h\mapsto {h^{1\over2p}\in L^{2p}(\Me)^+}$; see
\cite[Lemma 3.4]{kosaki1984applicationsuc}. Therefore
\[
\|\gamma^*_{\rho_n,p}(k)-\gamma^*_{\rho,p}(k)\|_p\le
\|\gamma^*_{\rho_n,p}(k-k_n)\|_p+\|\gamma^*_{\rho_n,p}(k_n)-\gamma^*_{\rho,p}(k)\|_p\to 0,
\]
showing the latter statement.
\end{proof}

\begin{lemma}\label{lemma:dpi} Let $\gamma:\Ne\to \Me$ be a normal positive unital map, and
let $\rho\in \Me_*^+$, {$\rho\ne0$, and} $b\in \Ne^+$. 
\begin{enumerate}
\item[(i)]  If $p\in [1/2,1)$, then 
\[
\Big\|h_{\rho\circ\gamma}^{\frac{1}{2p}}bh_{\rho\circ\gamma}^{\frac{1}{2p}}\Big\|_p\le
\Big\|h_{\rho}^{\frac{1}{2p}}\gamma(b)h_{\rho}^{\frac{1}{2p}}\Big\|_p.
\]

\item[(ii)]  If $p\in [1,\infty]$, the inequality reverses.
\end{enumerate}
\end{lemma}

\begin{proof} Let us denote $\beta:=\gamma_\rho^*$ and let $\omega\in {\Ne_*^+}$ be such
that 
$h_\omega:=h_{\rho\circ\gamma}^{\frac12}bh_{\rho\circ\gamma}^{\frac12}\in L^1(\Ne)^+$. Then
$\beta$ is a normal positive unital map, and {by \eqref{eq:petzdual}} we have
\[
\beta_*(h_\omega)=h_\rho^{\frac12}\gamma(b)h_\rho^{\frac12},\qquad
\beta_*(h_{\rho\circ\gamma})=h_\rho.
\]
Let $p\in [1/2,1)$, then  
\begin{align*}
\Big\|h_{\rho}^{\frac{1}{2p}}\gamma(b)h_{\rho}^{\frac{1}{2p}}\Big\|^p_p
&=\Big\|h_\rho^{\frac{1-p}{2p}}\beta_*(h_\omega)h_\rho^{\frac{1-p}{2p}}\Big\|_p^p
=Q_{p,p}(\beta_*(h_\omega)\|h_\rho)=Q_{p,p}(\beta_*(h_\omega)\|\beta_*(h_{\rho\circ\gamma}))\\
&\ge  Q_{p,p}(h_\omega\|h_{\rho\circ\gamma})=\Big\|h_{\rho\circ\gamma}^{\frac{1-p}{2p}}h_\omega
h_{\rho\circ\gamma}^{\frac{1-p}{2p}}\Big\|_p^p
=\Big\|h_{\rho\circ\gamma}^{\frac{1}{2p}}bh_{\rho\circ\gamma}^{\frac{1}{2p}}\Big\|^p_p.
\end{align*}
Here we have used the DPI for the sandwiched R\'enyi  divergence $D_{\alpha,\alpha}$ for
$\alpha\in [1/2,1)$; see \cite[Theorem 4.1]{jencova2021renyi}. This proves (i).
The case (ii) is immediate from Lemma \ref{lemma:pcontraction}. This was proved also 
in \cite{kato2023onrenyi} (see Eq.~(A3) therein), by using the same argument.
\end{proof}

\begin{theorem}[DPI] \label{thm:dpi}
Let $\psi,\varphi\in \Me_*^+$, $\psi\ne 0$ and let $\gamma:\Ne\to \Me$ be a normal positive
unital map. Assume either of the following conditions:
\begin{enumerate}
\item[(i)] $0<\alpha<1$, $\max\{\alpha,1-\alpha\}\le z$,
\item[(ii)] $\alpha>1$, $\max\{\alpha/2,\alpha-1\}\le z\le \alpha$.
\end{enumerate}
Then we have
\[
D_{\alpha,z}(\psi\circ\gamma\|\varphi\circ\gamma)\le D_{\alpha,z}(\psi\|\varphi).
\]
\end{theorem}

\begin{proof}
Under the conditions (i), the DPI was proved in \cite[Theorem 1(viii)]{kato2023onrenyi}.
Since parts of this proof will be used below, we repeat it here.

Assume the conditions in (i) and put $p:=\frac{z}{\alpha}$, $r:=\frac{z}{1-\alpha}$, so that $p,r\ge 1$. 
For any $b\in \Ne^{++}$, we have by  the Choi inequality \cite{choi1974aschwarz} 
that  $\gamma(b)^{-1}\le \gamma(b^{-1})$, so that  by Lemmas \ref{lemma:order1} and
\ref{lemma:dpi}(ii), we have 
\begin{equation}\label{eq:ineq}
\Big\|h_\varphi^{\frac{1}{2r}}\gamma(b)^{-1}h_\varphi^{\frac{1}{2r}}\Big\|_r\le
\Big\|h_\varphi^{\frac{1}{2r}}\gamma(b^{-1})h_\varphi^{\frac{1}{2r}}\Big\|_r\le
\Big\|h_{\varphi\circ\gamma}^{\frac{1}{2r}}b^{-1}h_{\varphi\circ\gamma}^{\frac{1}{2r}}\Big\|_r.
\end{equation}
Using the variational expression in Theorem \ref{thm:variational}(i), we have
\begin{align*}
Q_{\alpha,z}(\psi\|\varphi)&\le \alpha\Big\|h_\psi^{\frac{1}{2p}}\gamma(b)h_\psi^{\frac{1}{2p}}\Big\|_p^p+
(1-\alpha)\Big\|h_\varphi^{\frac{1}{2r}}\gamma(b)^{-1}h_\varphi^{\frac{1}{2r}}\Big\|_r^r\\
&\le \alpha\Big\|h_{\psi\circ\gamma}^{\frac{1}{2p}}bh_{\psi\circ\gamma}^{\frac{1}{2p}}\Big\|_p^p+
(1-\alpha)\Big\|h_{\varphi\circ\gamma}^{\frac{1}{2r}}b^{-1}h_{\varphi\circ\gamma}^{\frac{1}{2r}}\Big\|_r^r.
\end{align*}
 Since this holds for all
$b\in \Ne^{++}$, it follows that $Q_{\alpha,z}(\psi\|\varphi)\le
Q_{\alpha}(\psi\circ\gamma\|\varphi\circ\gamma)$, which proves the DPI in this case.

Assume next the condition (ii), and put $p:=\frac{z}{\alpha}$, $q:=\frac{z}{\alpha-1}$, so
that $p\in [1/2,1)$ and $q\ge 1$. Using Theorem
\ref{thm:variational}(ii), we get for any $b\in \Ne^+$,
\begin{align*}
Q_{\alpha,z}(\psi\|\varphi)&\ge
\alpha\Big\|h_\psi^{\frac{1}{2p}}\gamma(b)h_\psi^{\frac{1}{2p}}\Big\|_p^p-
(\alpha-1)\Big\|h_\varphi^{\frac{1}{2q}}\gamma(b)h_\varphi^{\frac{1}{2q}}\Big\|_q^q\\
&\ge \alpha\Big\|h_{\psi\circ\gamma}^{\frac{1}{2p}}bh_{\psi\circ\gamma}^{\frac{1}{2p}}\Big\|_p^p-
(\alpha-1)\Big\|h_{\varphi\circ\gamma}^{\frac{1}{2q}}bh_{\varphi\circ\gamma}^{\frac{1}{2q}}\Big\|_q^q,
\end{align*}
here we used both (i) and (ii) in Lemma \ref{lemma:dpi}. Again, since this holds for all
$b\in \Ne^+$, we get the desired inequality.
\end{proof}

Once the DPI has been proved as above, the joint concavity/convexity of $Q_{\alpha,z}(\psi\|\ffi)$ in $\psi,\ffi$
follows immediately by taking a unital positive map $\gamma:\Me\to\Me\oplus\Me$, $\gamma(a):=a\oplus a$,
as was shown in \cite[Theorem 1(ix)]{kato2023onrenyi}. Here we state the next proposition for completeness.

\begin{prop}\label{prop:jointconv}
If $\alpha$ and $z$ satisfy condition (i) (resp., (ii)) in Theorem \ref{thm:dpi}, then
$(\psi,\ffi)\mapsto Q_{\alpha,z}(\psi\|\ffi)$ is jointly concave (resp., jointly convex) on $\Me_*^+\times\Me_*^+$.
\end{prop}

\subsection{Martingale convergence}\label{Sec-3.2}

An important consequence of DPI is the martingale convergence property that will be
proved in this subsection. Here assume that $\cM$ is a $\sigma$-finite von Neumann algebra.

\begin{theorem}\label{thm:martingale}
Let $\psi,\varphi\in \Me_*^+$, $\psi\ne 0$, and let $\{\cM_i\}$ be an increasing net of von Neumann
subalgebras of $\cM$ containing the unit of $\cM$ such that $\cM=\bigl(\bigcup_i\cM_i\bigr)''$.
Assume that $\alpha$ and $z$ satisfy the DPI bounds (that is, condition (i) or (ii) in
Theorem \ref{thm:dpi}). Then we have
\begin{align}\label{eq:martingale}
D_{\alpha,z}(\psi\|\ffi)=\lim_iD_{\alpha,z}(\psi|_{\cM_i}\|\ffi|_{\cM_i})
\quad\mbox{increasingly}.
\end{align}
\end{theorem}

\begin{proof}
Let $\ffi_i:=\ffi|_{\cM_i}$ and $\psi_i:=\psi|_{\cM_i}$. From Theorem \ref{thm:dpi} it follows that
$D_{\alpha,z}(\psi\|\ffi)\ge D_{\alpha,z}(\psi_i\|\ffi_i)$ for all $i$ and
$i\mapsto D_{\alpha,z}(\psi_i\|\ffi_i)$ is increasing. Hence, to show \eqref{eq:martingale}, it suffices to
prove that
\begin{align}\label{eq:martingale1}
D_{\alpha,z}(\psi\|\ffi)\le\sup_iD_{\alpha,z}(\psi_i\|\ffi_i).
\end{align}
To do this, we may assume that $\ffi$ is faithful. Indeed, assume that
\eqref{eq:martingale1} has been shown when $\ffi$ is
faithful. For general $\ffi\in\cM_*^+$, from the assumption of $\cM$ being $\sigma$-finite, there exists
a $\ffi_0\in\cM_*^+$ with $s(\ffi_0)=\1-s(\ffi)$. Let $\ffi^{(n)}:=\ffi+n^{-1}\ffi_0$ and
$\ffi_i^{(n)}:=\ffi^{(n)}|_{\cM_i}$ for each $n\in\bN$. Thanks to the lower semi-continuity in
\cite[Theorems 1(iv) and 2(iv)]{kato2023onrenyi} and the order relation
\cite[Theorems 1(iii) and 2(iii)]{kato2023onrenyi} we have
\begin{align*}
D_{\alpha,z}(\psi\|\ffi)&\le\liminf_{n\to\infty}D_{\alpha,z}(\psi\|\ffi^{(n)}) \\
&\le\liminf_{n\to\infty}\sup_iD_{\alpha,z}(\psi_i\|\ffi_i^{(n)}) \\
&\le\sup_iD_{\alpha,z}(\psi_i\|\ffi_i),
\end{align*}
proving \eqref{eq:martingale1} for general $\ffi$. Below we assume the faithfulness of $\ffi$ and write
$\cE_{\cM_i,\ffi}$ for the generalized conditional expectation from $\cM$ to $\cM_i$ with respect to $\ffi$. 
Then from the martingale convergence given in \cite[Theorem 3]{hiai1984strong}, it follows that
\begin{align}\label{eq:martingaleHT}
\psi_i\circ\cE_{\cM_i,\ffi}=\psi\circ\cE_{\cM_i,\ffi}\to\psi\quad\mbox{in the norm},
\end{align}
as well as
\begin{align}\label{eq:condexp}
\ffi_i\circ\cE_{\cM_i,\ffi}=\ffi\circ\cE_{\cM_i,\ffi}=\ffi.
\end{align}
Using lower semi-continuity and DPI, we obtain
\[
D_{\alpha,z}(\psi\|\ffi)\le \liminf_{i}
D_{\alpha,z}(\psi_i\circ\cE_{\cM_i,\varphi}\|\ffi)\le \liminf_i
D_{\alpha,z}(\psi_i\|\ffi_i)\le \sup_i D_{\alpha,z}(\psi_i\|\ffi_i).
\]
\end{proof}

The next proposition is another martingale type convergence.

\begin{prop}\label{prop:mart-proj}
Let $\psi,\varphi\in \Me_*^+$, $\psi\ne 0$, and let $\{e_i\}$ be an increasing net of projections in $\Me$
such that $e_i\nearrow\1$. Either if $0<\alpha<1$ and $z>0$, or if $\alpha>1$ and
$\max\{\alpha/2,\alpha-1\}\le z\le \alpha$ (that is, $\alpha$ and $z$ satisfy the DPI bounds when
$\alpha>1$), then
\begin{align}\label{eq:mart-proj}
D_{\alpha,z}(\psi\|\ffi)=\lim_iD_{\alpha,z}(e_i\psi e_i\|e_i\ffi e_i),
\end{align}
where $e_i\psi e_i:=\psi|_{e_i\Me e_i}$ and $e_i\ffi e_i:=\ffi|_{e_i\Me e_i}$.
\end{prop}

\begin{proof}
Let $\psi_i,\ffi_i\in\cM_*^+$ be such that $\psi_i:=\psi(e_i\cdot e_i)$ and $\ffi_i:=\ffi(e_i\cdot e_i)$.
Since $s(\psi_i)\le e_i$ and $e_i\psi e_i=\psi_i|_{e_i\Me e_i}$ and similarly for $\ffi_i$, we note
(see Remark \ref{remark:defi}) that
\[
D_{\alpha,z}(e_i\psi e_i\|e_i\ffi e_i)
=D_{\alpha,z}(\psi_i|_{e_i\Me e_i}\|\ffi_i|_{e_i\Me e_i})=D_{\alpha,z}(\psi_i\|\ffi_i)
\]
for all $\alpha,z>0$, $\alpha\ne1$. Moreover, we have
\begin{align*}
\|\psi-\psi_i\|&=\|h_\psi-e_ih_\psi e_i\|_1\le\|(\1-e_i)h_\psi\|_1+\|e_ih_\psi(\1-e_i)\|_1 \\
&\le\|(\1-e_i)h_\psi^{1/2}\|_2\|h_\psi^{1/2}\|_2+\|e_ih_\psi^{1/2}\|_2\|h_\psi^{1/2}(\1-e_i)\|_2 \\
&=\psi(\1-e_i)^{1/2}\psi(\1)^{1/2}+\psi(e_i)^{1/2}\psi(\1-e_i)^{1/2}\to0\quad\mbox{as $i\to\infty$},
\end{align*}
and similarly $\|\ffi-\ffi_i\|_1\to0$. Hence, when $0<\alpha<1$ and $z>0$,
the joint continuity of $Q_{\alpha,z}$ in \cite[Theorem 1(iv)]{kato2023onrenyi} gives \eqref{eq:mart-proj}.

Next, assume that $\alpha>1$ and $\alpha,z$ satisfy the DPI bounds. Let
$\Me_i:=e_i\Me e_i\oplus\bC(\1-e_i)$; then $\{\Me_i\}$ is an increasing net of von Neumann
subalgebras of $\Me$ containing the unit of $\Me$ with $\Me=\bigl(\bigcup_i\Me_i\bigr)''$. Since
$\psi|_{\Me_i}=e_i\psi e_i\oplus\psi(\1-e_i)$ and similarly for $\ffi|_{\Me_i}$, it follows from
Theorem \ref{thm:martingale} and \cite[Theorems 1(ii) and 2(ii)]{kato2023onrenyi} that
\[
Q_{\alpha,z}(\psi\|\ffi)
=\lim_i\bigl[Q_{\alpha,z}(e_i\psi e_i\|e_i\ffi e_i)
+\psi(\1-e_i)^\alpha\ffi(\1-e_i)^{1-\alpha}\bigr].
\]
Here, $\psi(\1-e_i)^\alpha\ffi(\1-e_i)^{1-\alpha}$ is defined with the usual conventions
\[
0^{1-\alpha}:=\begin{cases}0 & (0<\alpha<1), \\ \infty & (\alpha>1),\end{cases}\qquad
\lambda\cdot\infty:=\begin{cases}0 & (\lambda=0), \\ \infty & (\lambda>0).\end{cases}
\]
Then \eqref{eq:mart-proj} holds if we show the following:
\begin{itemize}
\item[(1)] If $Q_{\alpha,z}(\psi\|\ffi)=\infty$, then
$\lim_iQ_{\alpha,z}(e_i\psi e_i\|e_i\ffi e_i)=\infty$.
\item[(2)] If $Q_{\alpha,z}(\psi\|\ffi)<\infty$, then $\lim_i\psi(\1-e_i)^\alpha\ffi(\1-e_i)^{1-\alpha}=0$.
\end{itemize}
These two facts can be shown in the same way as in the proof of \cite[Theorem 4.5]{hiai2018quantum},
whose details are omitted here.
\end{proof}

\begin{remark}\label{remark:martingale}
The martingale convergence in Theorem \ref{thm:martingale} is applied, in particular, when $\Me$ is an
injective (or approximately finite dimensional) von Neumann algebra. It may also be useful since one can
reduce studying $D_{\alpha,z}$ in general von Neumann algebras to the finite von Neumann algebra
case by Haagerup's reduction theorem; see, e.g., the proof of Lemma \ref{L-B.2} in Appendix
\ref{Sec-B}. The convergence in Proposition \ref{prop:mart-proj} is of quite use as well.
For example, in the infinite dimensional $\mathcal{B}(\mathcal{H})$ setting, this proposition was given in
\cite[Proposition 3.40]{mosonyi2023thestrong} and it was frequently utilized in \cite{mosonyi2023thestrong}
to prove results related to $D_{\alpha,z}$ (in particular, $D_{\alpha,\alpha}$) from those in the finite
dimensional case.
\end{remark}

\section{Equality in DPI and reversibility of channels}\label{Sec-4}

In what follows, a (quantum) \emph{channel} is a normal \emph{2-positive} unital
map $\gamma: \Ne\to \Me$. We use this term for the sake of convenience, while in most literature,
a quantum channel is assumed to be completely positive.

\begin{defi}\label{defi:reversible}
Let $\gamma:\Ne\to \Me$ be a channel and let $\mathcal S \subset
\Me_*^+$. We say that $\gamma$ is \emph{reversible} (or \emph{sufficient}) with respect to
$\mathcal S$ if there exists a channel $\beta:\Me\to \Ne$ such that
\[
\rho\circ\gamma\circ\beta=\rho\quad\mbox{for all}\ \rho\in \mathcal S.
\]
\end{defi}

The notion of sufficient channels was introduced by Petz
\cite{petz1986sufficient,petz1988sufficiency}, who also obtained a number of conditions
characterizing this situation. It particular, it was proved in \cite{petz1988sufficiency}
(also \cite{jencova2006sufficiency}) that sufficient channels can be characterized by
equality in the DPI for the \emph{relative entropy} $D(\psi\|\varphi)$: if $D(\psi\|\varphi)<\infty$,
then a channel $\gamma$ is sufficient with respect to $\{\psi,\varphi\}$ if and only if 
\[
D(\psi\circ\gamma\|\varphi\circ\gamma)=D(\psi\|\varphi). 
\]
This characterization has been proved for a number of other divergence measures, including the
standard R\'enyi divergences $D_{\alpha,1}$ with $0<\alpha<2$  and the sandwiched
R\'enyi divergences $D_{\alpha,\alpha}$ for $\alpha>1/2$
(\cite{hiai2021quantum,jencova2018renyi,jencova2021renyi}).
Another important result of \cite{petz1988sufficiency} shows that the Petz dual $\gamma_\varphi^*$
is a universal recovery map, in the sense given in the proposition below. 

\begin{prop}\label{prop:universal}
Let $\gamma:\Ne\to \Me$ be a channel and let $\varphi\in \Me_*^+$ be such that both $\ffi$ and
$\ffi\circ\gamma$ are faithful. Then the following statements hold:
\begin{itemize}
\item[(i)] For any $\psi\in \Me_*^+$, $\gamma$ is sufficient with respect to $\{\psi,\varphi\}$ if and only
if $\psi\circ\gamma\circ\gamma_\varphi^*=\psi$.

\item[(ii)]
There is a faithful normal conditional expectation $\mathcal E$ from $\Me$ onto a von Neumann
subalgebra of $\Me$ such that $\varphi\circ \mathcal E=\varphi$, and $\gamma$ is sufficient with
respect to $\{\psi,\varphi\}$ if and only if also $\psi\circ\mathcal E=\psi$.
\end{itemize}
\end{prop}

See \cite[Lemma 4.3]{jencova2018renyi}, for a proof of the statement (ii).
Note (see also \cite[Theorem 2]{petz1988sufficiency} and the proof of \cite[Theorem 3]{petz1988sufficiency})
that the range of the conditional expectation $\mathcal E$ is the set of fixed points of the channel
$\gamma\circ\gamma_\varphi^*$.

Our aim in this section is to prove that equality in the DPI for $D_{\alpha,z}$ with 
values of the parameters (strictly) contained in the DPI bounds of Theorem \ref{thm:dpi}
characterizes sufficiency of channels. Throughout this section, we use the
notations $\psi_0:=\psi\circ\gamma$ and $\ffi_0:=\ffi\circ\gamma$. We also denote
\[
p:=\frac{z}{\alpha},\qquad r:=\frac{z}{1-\alpha},\qquad q:=-r=\frac{z}{\alpha-1}. 
\]

\subsection{The case $\alpha\in (0,1)$}\label{Sec-4.1}

Here we study equality in the DPI for $D_{\alpha,z}$ with $\alpha\in (0,1)$,  for a pair of positive
normal functionals $\psi,\ffi\in \Me_*^+$ and a normal positive unital map $\gamma:\Ne\to \Me$. We first
prove some equality conditions  in the case $\psi\sim \ffi$. 

\begin{prop}\label{prop:DPI_equality}
Let $0<\alpha<1$ and $\max\{\alpha,1-\alpha\}\le z$.  Assume that $\psi\sim \ffi$ and  let
$\gamma:\Ne\to \Me$ be a normal positive unital map. Let $\bar a\in \Me^{++}$  
be the unique minimizer as in Lemma \ref{lemma:variational_majorized} for $Q_{\alpha,z}(\psi\|\ffi)$
and let $\bar a_0\in \Ne^{++}$ be the minimizer for $Q_{\alpha,z}(\psi_0\|\ffi_0)$.
The following conditions are equivalent:
\begin{itemize}
\item[(i)] $D_{\alpha,z}(\psi_0\|\ffi_0)=D_{\alpha,z}(\psi\|\ffi)$, i.e.,
$Q_{\alpha,z}(\psi_0\|\ffi_0)=Q_{\alpha,z}(\psi\|\ffi)$.
\item[(ii)] $\gamma(\bar a_0)=\bar a$ and
$\Big\|h_\psi^{1\over2p}\gamma(\bar a_0)h_\psi^{1\over2p}\Big\|_{p}
=\Big\|h_{\psi_0}^{1\over2p}\bar a_0h_{\psi_0}^{1\over2p}\Big\|_{p}$.
\item[(iii)] $\Big\|h_\psi^{1\over2p}\bar ah_\psi^{1\over2p}\Big\|_{p}
=\Big\|h_{\psi_0}^{1\over2p}\bar a_0h_{\psi_0}^{1\over2p}\Big\|_{p}$.
\item[(iv)] $\gamma(\bar a_0^{-1})=\bar a^{-1}$ and
$\Big\|h_\ffi^{1\over2r}\gamma(\bar a_0^{-1})h_\ffi^{1\over2r}\Big\|_{r}
=\Big\|h_{\ffi_0}^{1\over2r}\bar a_0^{-1}h_{\ffi_0}^{1\over2r}\Big\|_{r}$.
\item[(v)] $\Big\|h_\ffi^{1\over2r}\bar a^{-1}h_\ffi^{1\over2r}\Big\|_{r}
=\Big\|h_{\ffi_0}^{1\over2r}\bar a_0^{-1}h_{\ffi_0}^{1\over2r}\Big\|_{r}$.
\end{itemize}
\end{prop}

\begin{proof}
{Since $\psi\sim\ffi$ by assumption and hence $\psi_0\sim \ffi_0$, we have
$s(\psi)=s(\ffi)$ and $s(\psi_0)=s(\ffi_0)$.} Using restrictions explained in the beginning
of Sec.~\ref{Sec-3}, we may assume that all $\psi,\ffi,\psi_0,\ffi_0$ are faithful.

(i)$\implies$(ii) \& (iv).\enspace
By Lemma \ref{lemma:dpi}(ii)
\begin{align}
\Big\|h_\psi^{1\over2p}\gamma(\bar a_0)h_\psi^{1\over2p}\Big\|_{p}
\le\Big\|h_{\psi_0}^{1\over2p}\bar a_0h_{\psi_0}^{1\over2p}\Big\|_{p},
\label{eq:ineq2}
\end{align}
and by \eqref{eq:ineq} we have
\begin{align}
\Big\|h_\ffi^{1\over2r}\gamma(\bar a_0)^{-1}h_\ffi^{1\over2r}\Big\|_{r}
\le\Big\|h_\ffi^{1\over2r}\gamma(\bar a_0^{-1})h_\ffi^{1\over2r}\Big\|_{r}
\le\Big\|h_{\ffi_0}^{1\over2r}\bar
a_0^{-1}h_{\ffi_0}^{1\over2r}\Big\|_{r}. \label{eq:ineq3}
\end{align}
From \eqref{eq:ineq2} and \eqref{eq:ineq3} it follows that
\begin{align*}
Q_{\alpha,z}(\psi\|\ffi)&\le \alpha\Big\|h_\psi^{1\over2p}\gamma(\bar
a_0)h_\psi^{1\over2p}\Big\|^p_{p}
+(1-\alpha)\Big\|h_\ffi^{1\over2r}\gamma(\bar a_0)^{-1}h_\ffi^{1\over2r}\Big\|^r_{r} \\
&\le\alpha\Big\|h_{\psi_0}^{1\over2p}\bar a_0h_{\psi_0}^{1\over2p}\Big\|^p_{p}
+(1-\alpha)\Big\|h_{\ffi_0}^{1\over2r}\bar a_0^{-1}h_{\ffi_0}^{1\over2r}\Big\|^r_{r}
=Q_{\alpha,z}(\psi_0\|\ffi_0)=Q_{\alpha,z}(\psi\|\ffi).
\end{align*}
By uniqueness in Lemma \ref{lemma:variational_majorized} we find that $\gamma(\bar
a_0)=\bar a$ and all the inequalities in \eqref{eq:ineq2} and \eqref{eq:ineq3} must
become equalities. Since $\gamma(\bar a_0^{-1})\ge\gamma(\bar a_0)^{-1}$, we see by
Lemma \ref{lemma:order1} that the equality in
\eqref{eq:ineq3} yields $\gamma(\bar a_0^{-1})=\gamma(\bar a_0)^{-1}=\bar a^{-1}$. Therefore,
(ii) and (iv) hold.

The implications (ii)$\implies$(iii) and (iv)$\implies$(v) are obvious.

(iii)$\implies$(i).\enspace
By (iii) with the equality \eqref{eq:minimizer1} in Lemma \ref{lemma:variational_majorized} we have
\begin{align*}
Q_{\alpha,z}(\psi\|\ffi)
&=\Tr\Bigl(h_\psi^{1\over2p}h_\ffi^{1\over r}h_\psi^{1\over2p}\Bigr)^z
=\Tr\Bigl(h_\psi^{1\over2p}\bar ah_\psi^{1\over2p}\Bigr)^{p} \\
&=\Tr\Bigl(h_{\psi_0}^{1\over2p}\bar a_0h_{\psi_0}^{1\over2p}\Bigr)^{p}
=\Tr\Bigl(h_{\psi_0}^{\1\over2p}h_{\ffi_0}^{1\over r}h_{\psi_0}^{1\over2p}\Bigr)^z
=Q_{\alpha,z}(\psi_0\|\ffi_0).
\end{align*}

(v)$\implies$(i).\enspace
Similarly, by (v) with the equality \eqref{eq:minimizer2} in Lemma
\ref{lemma:variational_majorized} we have
\begin{align*}
Q_{\alpha,z}(\psi\|\ffi)
&=\Tr\Bigl(h_\ffi^{1\over2r}h_\psi^{1\over p}h_\ffi^{1\over2r}\Bigr)^z
=\Tr\Bigl(h_\ffi^{1\over2r}\bar a^{-1}h_\ffi^{1\over2r}\Bigr)^{r} \\
&=\Tr\Bigl(h_{\ffi_0}^{1\over2r}\bar a_0^{-1}h_{\ffi_0}^{1\over2r}\Bigr)^{r}
=\Tr\Bigl(h_{\ffi_0}^{1\over2r}h_{\psi_0}^{1\over p}h_{\ffi_0}^{1\over2r}\Bigr)^z
=Q_{\alpha,z}(\psi_0\|\ffi_0).
\end{align*}
\end{proof}

\begin{remark}\label{rem:conditions} Note that the above conditions extend the results
obtained in \cite{leditzky2017data} and \cite{zhang2020equality} in the finite dimensional case.
Indeed, as seen from \eqref{eq:minimizer2}, the first condition in (ii) with $\alpha=z$ is 
equivalent to the condition in \cite[Theorem 1]{leditzky2017data} (obtained under the more
general assumption that $s(\psi)\le s(\ffi)$). Since $p=\frac{z}{\alpha}=1$ in
\cite{leditzky2017data} and 
$\big\|h_\psi^{1\over2}\gamma(a_0)h_\psi^{1\over2}\big\|_{1}=\Tr
\bigl(\gamma(\alpha_0)h_\psi\bigr)=\big\|h_{\psi_0}^{1\over2}a_0h_{\psi_0}^{1\over2}\big\|_{1}$
for any $a_0\in \Ne$, the second condition in (ii) is automatic.
Moreover, (ii) extends the necessary condition in \cite[Theorem 1.2(2)]{zhang2020equality}
to a necessary and sufficient one. While in both of these works $\gamma$ was required
to be completely positive, only positivity   is assumed in Proposition
\ref{prop:DPI_equality}.

\end{remark}

\begin{theorem}\label{thm:suffle1} Let $0<\alpha<1$ and $\max\{\alpha,1-\alpha\}\le
z$. Let $\psi,\varphi\in \Me_*^+$, $\psi\ne0$, and assume either that $\alpha<z$ and $s(\ffi)\le
s(\psi)$, or that $1-\alpha<z$ and $s(\psi)\le s(\ffi)$. 
Then a channel $\gamma:\Ne \to \Me$ is reversible with respect to
$\{\psi,\varphi\}$ if and only if
\[
D_{\alpha,z}(\psi\circ\gamma\|\varphi\circ\gamma)=D_{\alpha,z}(\psi\|\varphi).
\]
\end{theorem}

\begin{proof} This proof is a modification of the proof of \cite[Theorem
5.1]{jencova2021renyi}.  
We will assume that $s(\ffi)\le s(\psi)$ and $\alpha<z$, that is, $p>1$. In the other case we may
exchange the roles of $p$, $r$ and of $\psi$, $\ffi$ by the equality
$Q_{\alpha,z}(\psi\|\varphi)=Q_{1-\alpha,z}(\varphi\|\psi)$. As before, we may assume that
both $\psi$ and $\psi_0$ are faithful.

The strategy of the proof is to use known results
in \cite{jencova2018renyi} for the sandwiched R\'enyi divergence $D_{p,p}$ with $p>1$. For this, let
$\mu,\omega\in\Me_*^+$ be such that
\[
h_{\mu}=\Big|h_{\ffi}^{1\over 2r}h_{\psi}^{1\over 2p}\Big|^{2z},\qquad
h_\omega=h_\psi^{p-1\over 2p}h_\mu^{\frac1p}h_\psi^{p-1\over 2p},
\]
and notice that
\[
Q_{\alpha, z}(\psi\|\ffi)=Q_{p,p}(\omega\|\psi).
\]
Let $\mu_0,\omega_0\in \Ne_*^+$ be similar functionals obtained from $\psi_0,\ffi_0$. Then
we have the equality
\begin{equation}\label{eq:sandwiched}
Q_{p,p}(\omega_0\|\psi_0)=Q_{\alpha,z}(\psi_0\|\ffi_0)=Q_{\alpha,z}(\psi\|\ffi)=Q_{p,p}(\omega\|\psi).
\end{equation}
Our first goal is to show that $\omega_0=\omega\circ\gamma$, which implies by
\cite[Theorem 4.6]{jencova2018renyi} that $\gamma$ is sufficient with respect to $\{\omega,\psi\}$. 
 We then apply Proposition \ref{prop:universal},  and the properties of the extensions of the
conditional expectation $\cE$ to the Haagerup $L^p$-spaces proved in
\cite{junge2003noncommutative}; see also Appendix \ref{Sec-A}.

Let us remark here that if $\psi\sim \ffi$, it follows from \eqref{eq:minimizer1} that
$h_\omega=h_\psi^{\frac12}\bar ah_\psi^{\frac12}$ and
$h_{\omega_0}=h_{\psi_0}^{\frac12}\bar a_0h_{\psi_0}^{\frac12}$. Hence from \eqref{eq:petzdual}
and condition (ii) in Proposition \ref{prop:DPI_equality}, we immediately have
\[
(\gamma^*_{\psi})_*(h_{\omega_0})=h_\psi^{\frac12}\gamma(\bar
a_0)h_\psi^{\frac12}=h_\omega,\quad
{\mbox{i.e.},\quad\omega_0\circ\gamma_\psi^*=\omega,}
\]
as well as $\psi_0\circ\gamma_\psi^*=\psi$ by \eqref{eq:petzdual2}.
These and \eqref{eq:sandwiched} show that $\gamma^*_\psi$ is sufficient with respect to
$\{\omega_0,\psi_0\}$. By Proposition \ref{prop:universal}(i) and the fact that the
Petz dual  $(\gamma_\psi^*)_{\psi_0}^*$ is $\gamma$ itself, this implies the desired
equality
\[
\omega\circ\gamma= \omega_0\circ \gamma_\psi^*\circ\gamma=\omega_0.
\]

In the case $\psi\not\sim \ffi$ some more work is required. Let $\psi_n:=\psi+\frac1n \ffi$ and
$\ffi_n:=\ffi+\frac1n \psi$. Then all $\psi_n$, $\ffi_n$ are faithful,  $\psi_n\to \psi$, $\ffi_n\to \ffi$ in
$\Me_*^+$, and moreover,  $\psi_n\sim \ffi_n$ for all $n$. Then
$\psi_n\circ\gamma\sim \ffi_n\circ\gamma$, $\psi_n\circ\gamma\to \psi_0$,
$\ffi_n \circ \gamma\to \ffi_0$ and by the joint continuity of $Q_{\alpha,z}$ in
\cite[Theorem 1(iv)]{kato2023onrenyi}, we have
\[
\lim_n
Q_{\alpha,z}(\psi_n\circ\gamma\|\ffi_n\circ\gamma)=Q_{\alpha,z}(\psi_0\|\ffi_0)
=Q_{\alpha,z}(\psi\|\ffi)=\lim_nQ_{\alpha,z}(\psi_n\|\ffi_n).
\]
Let $\bar b_{n}\in \Ne^{++}$ be the minimizer for the variational expression for
$Q_{\alpha,z}(\psi_n\circ\gamma\|\ffi_n\circ\gamma)$ {given in \eqref{F-2.1}.}
Let also $\bar a_n$ be the minimizer for $Q_{\alpha,z}(\psi_n\|\ffi_n)$, and let
$f_n:\Me^{++}\to \mathbb R^+$ be the function minimized in the expression for
$Q_{\alpha,z}(\psi_n\|\ffi_n)$ (see \eqref{func-variational}). We then have 
\begin{align*}
Q_{\alpha,z}(\psi_n\circ\gamma\|\ffi_n\circ\gamma)-f_n(\gamma(\bar
b_{n}))&=\alpha\left(\Big\|h_{\psi_n\circ\gamma}^{\frac1{2p}}\bar b_{n}
h_{\psi_n\circ\gamma}^{\frac1{2p}}\Big\|_p^p-\Big\|h_{\psi_n}^{1\over 2p}\gamma(\bar
b_n)h_{\psi_n}^{1\over 2p}\Big\|_p^p\right)\\
&\quad + (1-\alpha)\left(\Big\|h_{\ffi_n\circ\gamma}^{\frac1{2r}}\bar b_{n}^{-1}
h_{\ffi_n\circ\gamma}^{\frac1{2r}}\Big\|_r^r-\Big\|h_{\ffi_n}^{1\over 2r}\gamma(\bar
b_n)^{-1}h_{\ffi_n}^{1\over 2r}\Big\|_r^r\right)\ge 0,
\end{align*}
where the inequality follows from Lemma \ref{lemma:dpi}(ii) and \eqref{eq:ineq}. We
obtain
\begin{equation}\label{eq:qfn}
Q_{\alpha,z}(\psi_n\circ\gamma\|\ffi_n\circ\gamma)-Q_{\alpha,z}(\psi_n\|\ffi_n)\ge f_n(\gamma(\bar
b_{n}))-Q_{\alpha,z}(\psi_n\|\ffi_n)\ge 0.
\end{equation}

Now let $\mu_{n,0}\in \Ne_*^+$ and $\mu_n\in \Me_*^+$ be such that (using \eqref{eq:minimizer1} in
Lemma \ref{lemma:variational_majorized})
\begin{align*}
h_{\mu_{n,0}}^{\frac 1p}=
\Big|h_{\ffi_n\circ\gamma}^{1\over 2r}h_{\psi_n\circ\gamma}^{1\over 2p}\Big|^{2\alpha}=
h_{\psi_n\circ\gamma}^{\frac 1{2p}}\bar b_{n}
h_{\psi_n\circ\gamma}^{\frac1{2p}},\qquad
h_{\mu_{n}}^{\frac1p}=\Big|h_{\ffi_n}^{1\over 2r}h_{\psi_n}^{1\over 2p}\Big|^{2\alpha}=
h_{\psi_n}^{\frac 1{2p}}\bar a_{n}h_{\psi_n}^{\frac1{2p}}. 
\end{align*}
Then $h_{\mu_{n,0}}^{\frac1p}\to h_{\mu_0}^{\frac1p}$ in $L^p(\Ne)$, which follows from the H\"older
inequality and the fact \cite{kosaki1984applicationsuc} that the map
$L^{2z}(\Ne)\ni h\mapsto |h|^{2\alpha}\in L^p(\Ne)$ is continuous in the norm. Similarly,
$h_{\mu_n}^{\frac1p}\to h_\mu^{\frac1p}$ in $L^p(\Me)$. 
Next, note that since
$Q_{\alpha,z}(\psi_n\circ\gamma\|\ffi_n\circ\gamma)$ and $Q_{\alpha,z}(\psi_n\|\ffi_n)$
have the same limit, we see from \eqref{eq:qfn} that
$f_n(\gamma(\bar b_{n}))-Q_{\alpha,z}(\psi_n\|\ffi_n)\to0$.
Moreover, by Lemma \ref{lemma:dpi}(ii) and \eqref{eq:infimum} note that
\[
\sup_n\Big\|h_{\psi_n}^{\frac 1{2p}}\gamma(\bar b_{n})h_{\psi_n}^{1\over 2p}\Big\|_p^p
\le\sup_n\Big\|h_{\psi_n\circ\gamma}^{\frac 1{2p}}\bar b_{n}
h_{\psi_n\circ\gamma}^{\frac 1{2p}}\Big\|_p^p
=\sup_nQ_{\alpha,z}(\psi_n\circ\gamma\|\ffi_n\circ\gamma)<\infty.
\]
Therefore, since
$\big\|h_{\psi_n}^{\frac 1{2p}}\gamma(\bar b_{n})h_{\psi_n}^{1\over 2p}-h_{\mu_n}^{\frac1p}\big\|_p$
means $\|\xi_p(\gamma(\bar b_n))-\xi_p(\bar a_n)\|_p$ defined for $\psi_n$ (in place of $\psi$),
it follows from Lemma \ref{lemma:variational_majorized2} that 
$h_{\psi_n}^{\frac 1{2p}}\gamma(\bar b_{n})h_{\psi_n}^{1\over 2p}-h_{\mu_n}^{\frac1p}\to
0$ in $L^p(\Me)$.  On the other hand, let $\gamma^*_{\psi_n,p}, \gamma^*_{\psi,p}$ be
the contractions defined in Lemma \ref{lemma:pcontraction}. We
then have 
\[
h_{\psi_n}^{\frac 1{2p}}\gamma(\bar b_{n})h_{\psi_n}^{1\over
2p}=\gamma_{\psi_{n},p}^*\bigl(h_{\mu_{n,0}}^{\frac1p}\bigr)
\]
and since $\gamma^*_{\psi_{n},p}(k)\to \gamma^*_{\psi,p}(k)$ in $L^p(\Me)$ for any
$k\in L^p(s(\psi\circ\gamma)\Ne s(\psi\circ\gamma))$ by Lemma \ref{lemma:pcontraction}, we have  
\[
\Big\|\gamma^*_{\psi,p}\bigl(h_{\mu_0}^{\frac1p}\bigr)-
\gamma_{\psi_{n},p}^*\bigl(h_{\mu_{n,0}}^{\frac1p}\bigr)\Big\|_p\le
\Big\|(\gamma^*_{\psi,p}-\gamma^*_{\psi_{n},p})\bigl(h_{\mu_0}^{\frac1p}\bigr)\Big\|_p+
\Big\|\gamma^*_{\psi_{n},p}\bigl(h_{\mu_0}^{\frac1p}-h_{\mu_{n,0}}^{\frac1p}\bigr)\Big\|_p\to 0.
\]
Putting all together, we obtain that 
\[
h_\mu^{\frac1p}=\lim_n h_{\mu_n}^{\frac1p}=\lim_n
\gamma^*_{\psi_{n},p}(h_{\mu_{n,0}}^{\frac1p})=\gamma^*_{\psi,p}(h_{\mu_0}^{\frac1p}).
\]
It follows from \eqref{gamma-rho-p2} that 
\[
(\gamma^*_{\psi})_*(h_{\omega_0})=
h_{\psi}^{\frac{p-1}{2p}}\gamma^*_{\psi,p}(h_{\mu_0}^{\frac1p})h_{\psi}^{\frac{p-1}{2p}}=
h_{\psi}^{\frac{p-1}{2p}}h_\mu^{\frac1p}h_{\psi}^{\frac{p-1}{2p}}=h_\omega.
\]
As we have seen in the case $\psi\sim \ffi$ above, this and \eqref{eq:sandwiched} imply that
\[
\omega\circ\gamma= \omega_0\circ \gamma_\psi^*\circ\gamma=\omega_0.
\]
Therefore, we have shown that $\gamma$ is sufficient with respect to $\{\omega,\psi\}$.

Next, let $\mathcal E$ be the faithful  normal conditional expectation onto the set of fixed points of
$\gamma\circ\gamma^*_\psi$ (see a note after Proposition \ref{prop:universal}).
Then by Proposition \ref{prop:universal}(ii), $\mathcal E$ preserves both $\psi$ and $\omega$.
Since $\psi$ is assumed faithful, we may construct the Haagerup $L^p$-spaces with respect
to $\psi$; see Remark \ref{remark:defi}. We then have the extensions of $\cE$ to any Haagerup
$L^p$-space as in Appendix \ref{Sec-A}, such that  by the bimodule property
\eqref{eq:cond}, 
\[
h_\psi^{p-1\over 2p}h_\mu^{\frac1p}h_\psi^{p-1\over 2p}=h_\omega=\mathcal
E_*(h_\omega)=h_\psi^{p-1\over 2p}\cE_p\bigl(h_\mu^{\frac1p}\bigr)h_\psi^{p-1\over 2p}.
\]
It follows that $\big|h_\ffi^{1\over 2r}h_\psi^{1\over 2p}\big|^{2\alpha}= h_\mu^{\frac1p}\in
L^p(\cE(\Me))$ and consequently $\big|h_\varphi^{\frac1{2r}}h_\psi^{\frac1{2p}}\big|=h_\mu^{1\over
2z}\in L^{2z}(\mathcal E(\Me))^+$. Note that by the  assumptions we must have $2z>1$. 
Let  $h_\varphi^{\frac1{2r}}h_\psi^{\frac1{2p}}=u\big|h_\varphi^{\frac1{2r}}h_\psi^{\frac1{2p}}\big|$ be the
polar decomposition in $L^{2z}(\Me)$; then we have by using \eqref{eq:cond} again,
\[
u^*h_\varphi^{\frac1{2r}}h_\psi^{\frac1{2p}}=\mathcal
E_{2z}\bigl(u^*h_\varphi^{\frac1{2r}}h_\psi^{\frac1{2p}}\bigr)=\mathcal
E_{2r}\bigl(u^*h_\varphi^{\frac1{2r}}\bigr)h_\psi^{\frac1{2p}},
\]
which implies that $u^*h_\varphi^{\frac1{2r}}\in L^{2r}(\cE(\Me))$. Since $\psi$ is
faithful, we have

\[
\ker\Bigl(\bigl(h_\ffi^{1\over2r}h_\psi^{1\over2p}\bigr)^*\Bigr)
=\ker\bigl(h_\psi^{1\over2p}h_\ffi^{1\over2r}\bigr)
=\ker h_\ffi^{1\over2r}=\ker h_\ffi,
\]
which implies that $uu^*=s(\ffi)$.
Hence by uniqueness of the polar decomposition in $L^{2r}(\Me)$ and $L^{2r}(\cE(\Me))$,
we obtain that $h_\ffi^{1\over 2r}\in L^{2r}(\cE(\Me))^+$ and $u\in \cE(\Me)$. 
Therefore, we must have $h_\ffi\in L^1(\cE(\Me))$, so that $\ffi\circ\mathcal E=\ffi$ and $\gamma$ is
sufficient with respect to $\{\psi,\varphi\}$ by Proposition \ref{prop:universal}(ii) again.
\end{proof}

\subsection{The case $\alpha>1$}\label{Sec-4.2}

We now turn to the case $\alpha>1$. Assume that $(\alpha,z)$ is within the DPI bounds,
that is, $\alpha>1$, $\max\{\alpha/2,\alpha-1\}\le z\le\alpha$.  Then
we have $p=\frac z\alpha\in [1/2,1]$ and $q=\frac z{\alpha-1}\ge 1$, and we further assume that $q>1$.
Here we need to assume that $D_{\alpha,z}(\psi\|\ffi)<\infty$, so that by Lemma \ref{lemma:renyi_2z}
there is some (unique) $y\in L^{2z}(\Me)s(\ffi)$ such that
\[
h_{\psi}^{1\over 2p}=y h_\ffi^{1\over 2q}.
\]
By the proof of Theorem \ref{thm:variational}(ii), we have the following variational expression
\begin{align}\label{eq:variationalq}
Q_{\alpha,z}(\psi\|\varphi) =\sup_{w\in
L^q(\Me)^+}\bigl\{\alpha\Tr\bigl((ywy^*)^p\bigr)-(\alpha-1)\Tr\bigl(w^q\bigr)\bigr\}.
\end{align}
Indeed, we note that $x$ in the proof of Theorem \ref{thm:variational}(ii) is $y^*y$ and
$\Tr\bigl((x^{1\over2}wx^{1\over2})^p\bigr)$ in expression \eqref{F-2.4} is rewritten as
$\Tr\bigl((|y|w|y|)^p\bigr)=\Tr\bigl((ywy^*)^p\bigr)$.
The supremum is attained at a unique point $\bar
w=(y^*y)^{\alpha-1}\in L^q(\Me)^+$, where uniqueness follows from strict concavity of the
function $w\mapsto \alpha\Tr\bigl((ywy^*)^p\bigr)-(\alpha-1)\Tr\bigl( w^q\bigr)$ thanks to $p\le1$ and $q>1$.

By DPI, we have $D_{\alpha,z}(\psi_0\|\varphi_0)\le D_{\alpha,z}(\psi\|\varphi)<\infty$,
so that there is some (unique) $y_0\in L^{2z}(\Ne)s(\ffi_0)$ such that 
\[
h_{\psi_0}^{\frac1{2p}}=y_0h_{\varphi_0}^{\frac1{2q}}.
\]

Since $D_{\alpha,z}(\psi\|\ffi)<\infty$ implies that $s(\psi)\le s(\ffi)$, we may assume that both $\ffi$ and $\ffi_0$
are faithful.

\begin{lemma}\label{lemma:le}
Let $(\alpha,z)$ be as assumed above, and assume further that $\alpha<z+1$. Let
$\gamma:\Ne\to\Me$ be a normal positive unital map.
Let $\gamma^*_{\ffi,q}:L^q(\Ne)\to L^q(\Me)$ be the contraction as in Lemma \ref{lemma:pcontraction}.
Let $\bar w:=(y^*y)^{\alpha-1}\in L^q(\Me)$ and
$\bar w_0:=(y_0^*y_0)^{\alpha-1}\in L^q(\Ne)$. Then equality in the DPI for $D_{\alpha,z}(\psi\|\ffi)$ holds
if and only if
\begin{equation}\label{eq:dpiw}
\bar w=\gamma^*_{\ffi,q}(\bar w_0)\quad \text{and}\quad  
\Tr\bigl(\bar w_0^q\bigr)=\Tr\bigl(\gamma^*_{\varphi,q}(\bar w_0)^q\bigr).
\end{equation}
\end{lemma}

\begin{proof} We first show that  for any $w_0\in L^q(\Ne)^+$,
\begin{equation}\label{eq:lemma}
\Tr\bigl((y\gamma^*_{\ffi,q}(w_0)y^*)^p\bigr)\ge
\Tr\bigl((y_0w_0y_0^*)^p\bigr).
\end{equation}
Let us first assume that
$w_0=h_{\varphi_0}^{\frac1{2q}}bh_{\varphi_0}^{\frac1{2q}}$ for some $b\in \Ne^+$. Then 
$\gamma^*_{\varphi,q}(w_0)=h_{\varphi}^{\frac1{2q}}\gamma(b)h_{\varphi}^{\frac1{2q}}$.
Therefore
\begin{align*}
\Tr\bigl((y\gamma^*_{\varphi,q}(w_0)y^*)^p\bigr)
&=\Tr\Bigl(\bigl(yh_{\varphi}^{\frac1{2q}}\gamma(b)h_{\varphi}^{\frac1{2q}}y^*\bigr)^p\Bigr)=
\Tr\Bigl(\bigl(h_\psi^{\frac1{2p}}\gamma(b)h_\psi^{\frac1{2p}}\bigr)^p\Bigr)\ge
\Tr\Bigl(\bigl(h_{\psi_0}^{\frac1{2p}}bh_{\psi_0}^{\frac1{2p}}\bigr)^p\Bigr)\\
&=\Tr\Bigl(\bigl(y_0h_{\varphi_0}^{\frac1{2q}}bh_{\varphi_0}^{\frac1{2q}}y_0^*\bigr)^p\Bigr)
=\Tr\bigl((y_0w_0y_0^*)^p\bigr),
\end{align*}
where the inequality is from Lemma \ref{lemma:dpi}(i). The proof of inequality
\eqref{eq:lemma} is finished by Lemma \ref{lemma:cone}.
By using this and the fact that $\gamma^*_{\ffi,q}$ is a contraction, it follows from the variational
expression in \eqref{eq:variationalq} that
\begin{align*}
Q_{\alpha,z}(\psi\|\varphi)
&\ge \alpha\Tr\bigl((y\gamma^*_{\varphi,q}(\bar w_0)y^*)^p\bigr)-
(\alpha-1)\Tr\bigl(\gamma^*_{\varphi,q}(\bar w_0)^q\bigr)\\
&\ge \alpha\Tr\bigr((y_0\bar w_0 y_0^*)^p\bigr)-(\alpha-1)\Tr\bigl(\bar
w_0^q\bigr)=Q_{\alpha,z}(\psi_0\|\varphi_0).
\end{align*}
Suppose that $Q_{\alpha,z}(\psi_0\|\varphi_0)=Q_{\alpha,z}(\psi\|\varphi)$, then both the
inequalities must be equalities. Since $\bar w\in L^q(\Me)^+$ and $\bar w_0\in L^q(\Ne)^+$ are
the unique elements such that the suprema
in the variational expressions in \eqref{eq:variationalq} for $Q_{\alpha,z}(\psi\|\varphi)$ and
$Q_{\alpha,z}(\psi_0\|\varphi_0)$ are attained, this proves \eqref{eq:dpiw}.
Conversely, if the equalities in \eqref{eq:dpiw} hold, then
\[
Q_{\alpha,z}(\psi_0\|\ffi_0)=\Tr\bigl((y_0^*y_0)^z\bigr)=\Tr\bigl(\bar w_0^q\bigr)
=\Tr\bigl(\bar w^q\bigr)=\Tr\bigl((y^*y)^z\bigr)=Q_{\alpha,z}(\psi\|\ffi).
\]
\end{proof}

\begin{theorem}\label{thm:suffge1}
{Let $\alpha>1$ and $\max\{\alpha/2,\alpha-1\}\le z\le\alpha$, and assume further that
$\alpha<z+1$.} Let $\gamma:\Ne\to \Me$ be a channel and let $\psi,\varphi\in \Me_*^+$,
$\psi\ne0$, be such that $D_{\alpha,z}(\psi\|\varphi)<\infty$ (hence $s(\psi)\le s(\ffi)$).
Then $D_{\alpha,z}(\psi\circ\gamma\|\ffi\circ\gamma)=D_{\alpha,z}(\psi\|\varphi)$ if and only if
$\gamma$ is sufficient with respect to $\{\psi,\ffi\}$.
\end{theorem}

\begin{proof}
As before we may assume that both $\ffi$ and $\ffi_0$ are faithful. Let $\bar w$ and $\bar w_0$
be as in Lemma \ref{lemma:le}. Let $\omega\in\Me_*^+$ and $\omega_0\in \Ne_*^+$ be such that
\[
h_\omega=h_\ffi^{\frac{q-1}{2q}}\bar wh_\ffi^{\frac{q-1}{2q}},\qquad
h_{\omega_0}=h_{\ffi_0}^{\frac{q-1}{2q}}\bar w_0h_{\ffi_0}^{\frac{q-1}{2q}}.
\]
Then by Lemma \ref{lemma:le} we have
\[
Q_{q,q}(\omega_0\|\varphi_0)=\Tr\bigl(\bar w_0^q\bigr)=\Tr\bigl(\bar
w^q\bigr)=Q_{q,q}(\omega\|\varphi).
\]
and using also Lemma \ref{lemma:pcontraction}, we have
\[
(\gamma^*_\ffi)_*(h_{\omega_0})=h_\ffi^{\frac1{2\alpha}}\gamma^*_{\ffi,q}(\bar
w_0)h_\ffi^{\frac1{2\alpha}}=h_\omega.
\]
Similarly to the proof of Theorem \ref{thm:suffle1} (note that we
have $q>1$ by the assumption), this shows that $\gamma$ is sufficient
with respect to $\{\omega,\ffi\}$. Hence $\omega\circ \cE=\omega$, where $\cE$ is the
conditional expectation onto the fixed points of $\gamma\circ\gamma^*_\ffi$.
Using again the extensions of $\cE$ to the Haagerup $L^p$-spaces with respect to $\ffi$
and their properties in Appendix \ref{Sec-A},
\[
h_\varphi^{q-1\over2q}\bar
wh_\varphi^{q-1\over 2q}=h_\omega=\cE(h_\omega)=h_\varphi^{q-1\over 2q}\cE(\bar
w)h_\varphi^{q-1\over 2q},
\]
which implies that $\bar w=\cE(\bar w)\in L^q(\cE(\Me))^+$. But then we also have
\[
|y|=\bar w^{\frac1{2(\alpha-1)}}=\bar w^{\frac{q}{2z}}\in L^{2z}(\cE(\Me))^+.
\]
Let $y=u|y|$ be the polar decomposition of $y$; then we obtain from the definition of $y$ that
$uu^*=s(yy^*)=s(\psi)$. Furthermore, since
\[
u^*h_\psi^{\frac1{2p}}=|y|h_\varphi^{\frac1{2q}}\in L^{2p}(\cE(\Me)),
\]
by uniqueness of the polar decomposition in $L^{2p}(\Me)$ and $L^{2p}(\cE(\Me))$, we
obtain that $h_{\psi}^{\frac1{2p}}\in L^{2p}(\cE(\Me))^+$ and $u\in \cE(\Me)$. Hence we must
have $h_\psi\in L^1(\cE(\Me))$ so that $\psi\circ\cE=\psi$ and $\gamma$ is sufficient with
respect to $\{\psi,\ffi\}$ by Proposition \ref{prop:universal}(ii). The converse is clear from DPI.
\end{proof}

\begin{remark}
We note that an extra assumption $\alpha<z+1$ is essential in Theorem \ref{thm:suffge1}.
In fact, for the Petz-type R\'enyi divergence $D_{2,1}$ (i.e., in the case $\alpha=2$ and $z=1$),
equality in the DPI does not necessarily imply the sufficiency of completely positive channels,
as exemplified in \cite[Example 2.2]{jencova2010markov} and \cite[Example 4.8]{hiai2017different}
in the finite dimensional case.
\end{remark}

\section{Monotonicity in the parameter $z$}\label{Sec-5}

It is well known \cite{berta2018renyi,hiai2018quantum,jencova2018renyi} that the standard R\'enyi divergence
$D_{\alpha,1}(\psi\|\ffi)$ is monotone increasing in $\alpha\in(0,1)\cup(1,\infty)$ and the sandwiched R\'enyi
divergence $D_{\alpha,\alpha}(\psi\|\ffi)$ is monotone increasing in $\alpha\in[1/2,1)\cup(1,\infty)$. It is also
known \cite{berta2018renyi,hiai2018quantum,jencova2018renyi} that
\begin{align}\label{F-5.1}
\lim_{\alpha\nearrow1}D_{\alpha,1}(\psi\|\ffi)=\lim_{\alpha\nearrow1}D_{\alpha,\alpha}(\psi\|\ffi)
=D_1(\psi\|\ffi),
\end{align}
and if $D_{\alpha,1}(\psi\|\ffi)<\infty$ (resp., $D_{\alpha,\alpha}(\psi\|\ffi)<\infty$) for some $\alpha>1$, then
\begin{align}\label{F-5.2}
\lim_{\alpha\searrow1}D_{\alpha,1}(\psi\|\ffi)=D_1(\psi\|\ffi)\quad
\Bigl(\mbox{resp.,}\ \lim_{\alpha\searrow1}D_{\alpha,\alpha}(\psi\|\ffi)=D_1(\psi\|\ffi)\Bigr),
\end{align}
where $D_1(\psi\|\ffi):=D(\psi\|\ffi)/\psi(\1)$, the normalized relative entropy.
In the rest of the paper we will discuss similar monotonicity properties {of $D_{\alpha,z}$ in
both parameters $\alpha$, $z$ and limits for $D_{\alpha,z}(\psi\|\ffi)$ as
$\alpha\nearrow1$, $\alpha\searrow1$.} We consider monotonicity in the parameter $z$ in Sec.~\ref{Sec-5}
and monotonicity in the parameter $\alpha$ in Sec.~\ref{Sec-6}.

The next theorem in \cite{kato2023onrenyi} shows the monotonicity of $D_{\alpha,z}$ in the
parameter $z>0$ when $0<\alpha<1$.

\begin{theorem}[\mbox{\cite[Theorem 1(x)]{kato2023onrenyi}}]\label{T-5.1}
For every $\psi,\ffi\in\cM_*^+$, $\psi\ne0$, and $0<\alpha<1$, the function
$z\mapsto D_{\alpha,z}(\psi\|\ffi)$ is monotone increasing on $(0,\infty)$.
\end{theorem}

\subsection{The finite von Neumann algebra case}\label{Sec-5.1}

In this subsection we show monotonicity of $D_{\alpha,z}$ in the parameter $z$ in the finite
von Neumann algebra setting. Recall (see, e.g., \cite[Example 9.11]{hiai2021lectures}) that if $(\Me,\tau)$
is a semi-finite von Neumann algebra $\Me$ with a faithful normal semi-finite trace $\tau$, then the
Haagerup $L^p$-space $L^p(\Me)$ is identified with the $L^p$-space $L^p(\Me,\tau)$ with respect to
$\tau$. Hence one can define $Q_{\alpha,z}(\psi\|\ffi)$ for $\psi,\ffi\in\Me_*^+$ by replacing, in
Definition \ref{defi:renyi}, $L^p(\Me)$ with $L^p(\Me,\tau)$ and $h_\psi\in L^1(\Me)^+$ with the
Radon--Nikodym derivative $d\psi/d\tau\in L^1(\Me,\tau)^+$. Below we use the symbol $h_\psi$ to
denote $d\psi/d\tau$ as well. Note that $\tau$ on $\Me^+$ is naturally extended to the positive part
$\widetilde\Me^+$ of the space $\widetilde\Me$ of $\tau$-measurable operators. By
\cite[Proposition 2.7]{fack1986generalized} (also \cite[Proposition 4.20]{hiai2021lectures}) we then have
\begin{align}\label{F-5.3}
\tau(a)=\int_0^\infty\mu_s(a)\,ds,\qquad a\in\widetilde\Me^+,
\end{align}
where $\mu_s(a)$ is the generalized $s$-number of $a$, defined by
\[
\mu_s(a):=\inf\{\|ae\|:\mbox{$e$ is a projection in $\Me$ with $\tau(\1-e)\le s$}\}
\]
for $s>0$ (see \cite{fack1986generalized} for more details).

In the rest of this subsection, we assume that $(\Me,\tau)$ is a finite von Neumann algebra
with a faithful normal finite trace $\tau$. Note that in this setting, $\widetilde\Me^+$ consists of all positive
self-adjoint operators affiliated with $\Me$.

\begin{lemma}\label{L-5.2}
Let $(\Me,\tau)$ be as mentioned above.
For every $\psi,\ffi\in\Me_*^+$ with $\psi\ne0$ and for any $\alpha,z>0$ with $\alpha\ne1$,
\[
D_{\alpha,z}(\psi\|\ffi)=\lim_{\eps\searrow0}D_{\alpha,z}(\psi\|\ffi+\eps\tau)\quad\mbox{increasingly},
\]
and hence $D_{\alpha,z}(\psi\|\ffi)=\sup_{\eps>0}D_{\alpha,z}(\psi\|\ffi+\eps\tau)$.
\end{lemma}

\begin{proof}
{\it Case $0<\alpha<1$}.\enspace
We need to prove that
\begin{align}\label{F-5.4}
Q_{\alpha,z}(\psi\|\ffi)&=\lim_{\eps\searrow0}Q_{\alpha,z}(\psi\|\ffi+\eps\tau)\quad\mbox{decreasingly}.
\end{align}
In the present setting we have by \eqref{F-5.3}
\begin{align}\label{F-5.5}
Q_{\alpha,z}(\psi\|\ffi)
=\tau\Bigl(\bigl(h_\psi^{\alpha/2z}h_\ffi^{1-\alpha\over z}h_\psi^{\alpha/2z}\bigr)^z\Bigr)
=\int_0^\infty\mu_s\Bigl(h_\psi^{\alpha/2z}h_\ffi^{1-\alpha\over z}h_\psi^{\alpha/2z}\Bigr)^z\,ds,
\end{align}
and similarly
\[
Q_{\alpha,z}(\psi\|\ffi+\eps\tau)
=\int_0^\infty\mu_s\Bigl(h_\psi^{\alpha/2z}h_{\ffi+\eps\tau}^{1-\alpha\over z}h_\psi^{\alpha/2z}\Bigr)^z\,ds.
\]
Since $h_{\ffi+\eps\tau}^{1-\alpha\over z}=(h_\ffi+\eps\1)^{1-\alpha\over z}$ decreases to
$h_\ffi^{1-\alpha\over z}$ in the measure topology as $\eps\searrow0$, it follows that
$h_\psi^{\alpha/2z}h_{\ffi+\eps\tau}^{1-\alpha\over z}h_\psi^{\alpha/2z}$ decreases to
$h_\psi^{\alpha/2z}h_\ffi^{1-\alpha\over z}h_\psi^{\alpha/2z}$ in the measure topology. Hence by
\cite[Lemma 3.4]{fack1986generalized} we have
$\mu_s\bigl(h_\psi^{\alpha/2z}h_{\ffi+\eps\tau}^{1-\alpha\over z}h_\psi^{\alpha/2z}\bigr)
\searrow\mu_s\bigl(h_\psi^{\alpha/2z}h_\ffi^{1-\alpha\over z}h_\psi^{\alpha/2z}\bigr)$
as $\eps\searrow0$ for almost every $s>0$. Since
$s\mapsto\mu_s\bigl(h_\psi^{\alpha/2z}h_{\ffi+\tau}^{1-\alpha\over z}h_\psi^{\alpha/2z}\bigr)^z$ is
integrable on $(0,\infty)$, the Lebesgue convergence theorem yields \eqref{F-5.4}.

{\it Case $\alpha>1$}.\enspace
We need to prove that
\begin{align}\label{F-5.6}
Q_{\alpha,z}(\psi\|\ffi)&=\lim_{\eps\searrow0}Q_{\alpha,z}(\psi\|\ffi+\eps\tau)\quad\mbox{increasingly}.
\end{align}
For any $\eps>0$, since $h_{\ffi+\eps\tau}=h_\ffi+\eps\1$ has the bounded inverse
$h_{\ffi+\eps\tau}^{-1}=(h_\ffi+\eps\1)^{-1}\in\Me^+$, one can define
$x_\eps:=(h_\ffi+\eps\1)^{-{\alpha-1\over2z}}h_\psi^{\alpha/z}(h_\ffi+\eps\1)^{-{\alpha-1\over2z}}
\in\widetilde\Me^+$ so that
\[
h_\psi^{\alpha/z}=(h_\ffi+\eps\1)^{\alpha-1\over2z}x_\eps(h_\ffi+\eps\1)^{\alpha-1\over2z}.
\]
In the present setting one can write by \eqref{F-5.3}
\begin{align}\label{F-5.7}
Q_{\alpha,z}(\psi\|\ffi+\eps\tau)=\tau(x_\eps^z)=\int_0^\infty\mu_s(x_\eps)^z\,ds\quad(\in[0,\infty]).
\end{align}
Let $0<\eps\le\eps'$. Since $(h_\ffi+\eps\1)^{-{\alpha-1\over z}}\ge(h_\ffi+\eps'\1)^{-{\alpha-1\over z}}$,
one has $\mu_s(x_\eps)\ge\mu_s(x_{\eps'})$ for all $s>0$, so that
\[
Q_{\alpha,z}(\psi\|\ffi+\eps\tau)\ge Q_{\alpha,z}(\psi\|\ffi+\eps'\tau).
\]
Hence $\eps>0\mapsto D_{\alpha,z}(\psi\|\ffi+\eps\tau)$ is decreasing.

First, assume that $s(\psi)\not\le s(\ffi)$. Then
$\mu_{s_0}\Bigl(h_\psi^{\alpha/2z}s(\ffi)^\perp h_\psi^{\alpha/2z}\Bigr)>0$ for some $s_0>0$; indeed,
otherwise, $h_\psi^{\alpha/2z}s(\ffi)^\perp h_\psi^{\alpha/2z}=0$ so that $s(\psi)\le s(\ffi)$. Hence we have
\[
\mu_s(x_\eps)=\mu_s\Bigl(h_\psi^{\alpha/2z}(h_\ffi+\eps\1)^{-{\alpha-1\over z}}h_\psi^{\alpha/2z}\Bigr)
\ge\eps^{-{\alpha-1\over z}}\mu_s\Bigl(h_\psi^{\alpha/2z}s(\ffi)^\perp h_\psi^{\alpha/2z}\Bigr)
\nearrow\infty\quad\mbox{as $\eps\searrow0$}
\]
for all $s\in(0,s_0]$. Therefore, it follows from \eqref{F-5.7} that
$Q_{\alpha,z}(\psi\|\ffi+\eps\tau)\nearrow\infty=Q_{\alpha,z}(\psi\|\ffi)$.

Next, assume that $s(\psi)\le s(\ffi)$. Take the spectral decomposition $h_\ffi=\int_0^\infty t\,de_t$ and
define $y,x\in\widetilde\Me^+$ by
\[
y:=h_\ffi^{-{\alpha-1\over z}}s(\ffi)=\int_{(0,\infty)}t^{-{\alpha-1\over z}}\,de_t,
\qquad x:=y^{1/2}h_\psi^{\alpha/z}y^{1/2}.
\]
Since
\[
h_\psi^{\alpha/z}=s(\ffi)h_\psi^{\alpha/z}s(\ffi)
=h_\ffi^{\alpha-1\over2z}y^{1/2}h_\psi^{\alpha/z}y^{1/2}h_\ffi^{\alpha-1\over2z}
=h_\ffi^{\alpha-1\over2z}xh_\ffi^{\alpha-1\over2z},
\]
one has, similarly to \eqref{F-5.7},
\begin{align}\label{F-5.8}
Q_{\alpha,z}(\psi\|\ffi)=\tau(x^z)=\int_0^\infty\mu_s(x)^z\,ds.
\end{align}
We write $(h_\ffi+\eps\1)^{-{\alpha-1\over z}}s(\ffi)=\int_{(0,\infty)}(t+\eps)^{-{\alpha-1\over z}}\,de_t$,
and for any $\delta>0$ choose a $t_0>0$ such that $\tau(e_{(0,t_0)})<\delta$. Then, since
$\int_{[t_0,\infty)}(t+\eps)^{-{\alpha-1\over z}}\,de_t\to\int_{[t_0,\infty)}t^{-{\alpha-1\over z}}\,de_t$
in the operator norm as $\eps\searrow0$, we obtain $(h_\ffi+\eps\1)^{-{\alpha-1\over z}}s(\ffi)\nearrow y$
in the measure topology (see \cite[1.5]{fack1986generalized}), so that
$h_\psi^{\alpha/2z}(h_\ffi+\eps\1)^{-{\alpha-1\over z}}h_\psi^{\alpha/2z}
\nearrow h_\psi^{\alpha/2z}yh_\psi^{\alpha/2z}$ in the measure topology as $\eps\searrow0$. Hence
we have by \cite[Lemma 3.4]{fack1986generalized}
\[
\mu_s(x_\eps)=\mu_s\Bigl(h_\psi^{\alpha/2z}(h_\ffi+\eps\1)^{-{\alpha-1\over z}}h_\psi^{\alpha/2z}\Bigr)
\nearrow\mu_s\Bigl(h_\psi^{\alpha/2z}yh_\psi^{\alpha/2z}\Bigr)=\mu_s(x)
\]
for all $s>0$. Therefore, by \eqref{F-5.7} and \eqref{F-5.8} the monotone convergence theorem yields
\eqref{F-5.6}.
\end{proof}

\begin{lemma}\label{L-5.3}
Let $(\Me,\tau)$ and $\psi,\ffi$ be as in Lemma \ref{L-5.2}, and let $0<z\le z'$. Then
\[
\begin{cases}
D_{\alpha,z}(\psi\|\ffi)\le D_{\alpha,z'}(\psi\|\ffi), & \text{$0<\alpha<1$},\\
D_{\alpha,z}(\psi\|\ffi)\ge D_{\alpha,z'}(\psi\|\ffi), & \text{$\alpha>1$}.
\end{cases}
\]
\end{lemma}

\begin{proof}
The case $0<\alpha<1$ is in Theorem \ref{T-5.1} for general von Neumann algebras.
For the case $\alpha>1$, by Lemma \ref{L-5.2} it suffices to show that, for every $\eps>0$,
\[
\tau\Bigl(\Bigl(y_\eps^{\alpha-1\over2z}h_\psi^{\alpha/z}y_\eps^{\alpha-1\over2z}\Bigr)^z\Bigr)
\ge\tau\Bigl(\Bigl(y_\eps^{\alpha-1\over2z'}h_\psi^{\alpha/z'}y_\eps^{\alpha-1\over2z'}\Bigr)^{z'}\Bigr),
\]
where $y_\eps:=(h_\ffi+\eps\1)^{-1}\in\Me^+$. The above is equivalently written as
\[
\tau\Bigl(\big|(h_\psi^{\alpha/2z'})^r(y_\eps^{(\alpha-1)/2z'})^r\big|^{2z}\Bigr)
\ge\tau\Bigl(\big|h_\psi^{\alpha/2z'} y_\eps^{(\alpha-1)/2z'}\big|^{2zr}\Bigr),
\]
where $r:=z'/z\ge1$. Hence the desired inequality follows from Kosaki's Araki--Lieb--Thirring inequality
\cite[Corollary 3]{kosaki1992aninequality}.
\end{proof}

When $(\Me,\tau)$ and $\psi,\ffi$ are as in Lemma \ref{L-5.2}, one can define, thanks to
Lemma \ref{L-5.3}, for any $\alpha\in(0,\infty)\setminus\{1\}$,
\begin{align}
Q_{\alpha,\infty}(\psi\|\ffi)&:=\lim_{z\to\infty}Q_{\alpha,z}(\psi\|\ffi)
=\inf_{z>0}Q_{\alpha,z}(\psi\|\ffi), \nonumber\\
D_{\alpha,\infty}(\psi\|\ffi)&:={1\over\alpha-1}\log{Q_{\alpha,\infty}(\psi\|\ffi)\over\psi(\1)} \nonumber\\
&\ =\lim_{z\to\infty}D_{\alpha,z}(\psi\|\ffi)
=\begin{cases}\sup_{z>0}D_{\alpha,z}(\psi\|\ffi), & \text{$0<\alpha<1$},\\
\inf_{z>0}D_{\alpha,z}(\psi\|\ffi), & \text{$\alpha>1$}.\end{cases} \label{F-5.9}
\end{align}
If $h_\psi,h_\ffi\in\Me^{++}$ (i.e., $\lambda^{-1}\tau\le\psi,\ffi\le\lambda\tau$ for some $\lambda>0$), then
the Lie--Trotter formula gives
\begin{align}\label{F-5.10}
Q_{\alpha,\infty}(\psi\|\ffi)=\tau\bigl(\exp(\alpha\log h_\psi+(1-\alpha)\log h_\ffi)\bigr).
\end{align}

\begin{lemma}\label{L-5.4}
Let $(\Me,\tau)$ and $\psi,\ffi$ be as in Lemma \ref{L-5.2}. Then for any $z>0$,
\[
\begin{cases}
D_{\alpha,z}(\psi\|\ffi)\le D_1(\psi\|\ffi), & \text{$0<\alpha<1$},\\
D_{\alpha,z}(\psi\|\ffi)\ge D_1(\psi\|\ffi), & \text{$\alpha>1$}.
\end{cases}
\]
\end{lemma}

\begin{proof}
First assume that $h_\psi,h_\ffi\in\Me^{++}$. Set self-adjoint $H:=\log h_\psi$ and
$K:=\log h_\ffi$ in $\Me$ and define $F(\alpha):=\log\tau\bigl(e^{\alpha H+(1-\alpha)K}\bigr)$ for $\alpha>0$.
Then by \eqref{F-5.10}, $F(\alpha)=\log Q_{\alpha,\infty}(\psi\|\ffi)$ for all $\alpha\in(0,\infty)\setminus\{1\}$,
and we compute
\begin{align*}
F'(\alpha)&={\tau\bigl(e^{\alpha H+(1-\alpha)K}(H-K)\bigr)\over\tau\bigl(e^{\alpha H+(1-\alpha)K}\bigr)}, \\
F''(\alpha)&={\tau\bigl(e^{\alpha H+(1-\alpha)K}(H-K)^2\bigr)\tau\bigl(e^{\alpha
H+(1-\alpha)K}\bigr)-\bigl\{\tau\bigl(e^{\alpha H+(1-\alpha)K}(H-K)\bigr)\bigr\}^2
\over\bigl\{\tau\bigl(e^{\alpha H+(1-\alpha)K}\bigr)\bigr\}^2}.
\end{align*}
Since $F''(\alpha)\ge0$ on $(0,\infty)$ thanks to the Schwarz inequality, we see that $F(\alpha)$ is
convex on $(0,\infty)$ and hence
\[
D_{\alpha,\infty}(\psi\|\ffi)={F(\alpha)-F(1)\over\alpha-1}
\]
is increasing in $\alpha\in(0,\infty)$, where for $\alpha=1$ the above right-hand side is understood as
\[
F'(1)={\tau(e^H(H-K))\over\tau(e^H)}={\tau\bigl(h_\psi(\log h_\psi-\log h_\ffi)\bigr)\over\tau(h_\psi)}
=D_1(\psi\|\ffi).
\]
Hence by \eqref{F-5.9} the assertion holds when $h_\psi,h_\ffi\in\Me^{++}$. Below we extend this to
general $\psi,\ffi\in\Me_*^+$.

{\it Case $0<\alpha<1$}.\enspace
Let $\psi,\ffi\in\Me_*^+$ and $z>0$. From \cite[Theorem 1(iv)]{kato2023onrenyi} and
\cite[Corollary 2.8(3)]{hiai2021quantum} we have
\begin{align*}
D_{\alpha,z}(\psi\|\ffi)&=\lim_{\eps\searrow0}D_{\alpha,z}(\psi+\eps\tau\|\ffi+\eps\tau), \\
D_1(\psi\|\ffi)&=\lim_{\eps\searrow0}D_1(\psi+\eps\tau\|\ffi+\eps\tau),
\end{align*}
so that we may assume that $\psi,\ffi\ge\eps\tau$ for some $\eps>0$. Take the spectral decompositions
$h_\psi=\int_0^\infty t\,de_t^\psi$ and $h_\ffi=\int_0^\infty t\,de_t^\ffi$, and define
$e_n:=e_n^\psi\wedge e_n^\ffi$ for each $n\in\bN$. Then
$\tau(e_n^\perp)\le\tau((e_n^\psi)^\perp)+\tau((e_n^\ffi)^\perp)\to0$ as $n\to\infty$,
so that $e_n\nearrow\1$. Hence by Proposition \ref{prop:mart-proj} one has
$D_{\alpha,z}(e_n\psi e_n\|e_n\ffi e_n)\to D_{\alpha,z}(\psi\|\ffi)$. On the other hand,
by \cite[Proposition 2.10]{hiai2021quantum} one has
$D_1(e_n\psi e_n\|e_n\ffi e_n)\to D_1(\psi\|\ffi)$ as well. Since
$D_{\alpha,z}(e_n\psi e_n\|e_n\ffi e_n)\le D_1(e_n\psi e_n\|e_n\ffi e_n)$
holds by the above shown case, we obtain the desired inequality for general $\psi,\ffi\in\Me_*^+$.

{\it Case $\alpha>1$}.\enspace
We show the extension to general $\psi,\ffi\in\Me_*^+$ by dividing four steps as follows, where
$h_\psi=\int_0^\infty t\,e_t^\psi$ and $h_\ffi=\int_0^\infty t\,de_t^\ffi$ are the spectral decompositions.

(1)\enspace
Assume that $h_\psi\in\Me^+$ and $h_\ffi\in\Me^{++}$. Set $\psi_n\in\Me_*^+$ by
$h_{\psi_n}=(1/n)e_{[0,1/n]}^\psi+\int_{(1/n,\infty)}t\,de_t^\psi$ ($\in\Me^{++}$). Since
$h_{\psi_n}^{\alpha/z}\searrow h_\psi^{\alpha/z}$ in the operator norm, we have by \eqref{F-5.5} and
\cite[Lemma 3.4]{fack1986generalized}
\begin{align}
Q_{\alpha,z}(\psi\|\ffi)&=\int_0^\infty\mu_s\Bigl((h_\ffi^{-1})^{\alpha-1\over2z}h_\psi^{\alpha/z}
(h_\ffi^{-1})^{\alpha-1\over2z}\Bigr)^z\,ds \nonumber\\
&=\lim_{n\to\infty}\int_0^\infty\mu_s\Bigl((h_\ffi^{-1})^{\alpha-1\over2z}h_{\psi_n}^{\alpha/z}
(h_\ffi^{-1})^{\alpha-1\over2z}\Bigr)^z\,ds
=\lim_{n\to\infty}Q_{\alpha,z}(\psi_n\|\ffi). \label{F-5.11}
\end{align}
From this and the lower semi-continuity of $D_1$ the extension holds in this case.

(2)\enspace
Assume that $h_\psi\in\Me^+$ and $h_\ffi\ge\delta\1$ for some $\delta>0$. Set $\ffi_n\in\Me_*^+$
by $h_{\ffi_n}=\int_{[\delta,n]}t\,de_t^\ffi+ne_{(n,\infty)}^\ffi$ ($\in\Me^{++}$). Since
$h_{\ffi_n}^{-{\alpha-1\over z}}\searrow h_\ffi^{-{\alpha-1\over z}}$ in the operator norm, we have by
\eqref{F-5.5} and \cite[Lemma 3.4]{fack1986generalized} again
\begin{align*}
Q_{\alpha,z}(\psi\|\ffi)&=\int_0^\infty\mu_s\Bigl(h_\psi^{\alpha/2z}h_\ffi^{-{\alpha-1\over z}}
h_\psi^{\alpha/2z}\Bigr)^z\,ds \\
&=\lim_{n\to\infty}\int_0^\infty\mu_s\Bigl(h_\psi^{\alpha/2z}h_{\ffi_n}^{-{\alpha-1\over z}}
h_\psi^{\alpha/2z}\Bigr)^z\,ds
=\lim_{n\to\infty}Q_{\alpha,z}(\psi\|\ffi_n).
\end{align*}
From this and (1) above the extension holds in this case too.

(3)\enspace
Assume that $\psi$ is general and $\ffi\ge\delta\tau$ for some $\delta>0$. Set $\psi_n\in\Me_*^+$
by $h_{\psi_n}=\int_{[0,n]}t\,de_t^\psi+ne_{(n,\infty)}^\ffi$ ($\in\Me^+$). Since
$h_{\psi_n}^{\alpha/z}\nearrow h_\psi^{\alpha/z}$ in the measure topology, one can argue as in \eqref{F-5.11}
with use of the monotone convergence theorem to see from (2) that the extension holds in this case too.

(4)\enspace
Finally, from (3) with Lemma \ref{L-5.2} and \cite[Corollary 2.8(3)]{hiai2021quantum} it follows that
the desired extension holds for general $\psi,\ffi\in\Me_*^+$.
\end{proof}

Note that the first inequality in (i) of the next proposition is a particular case of
Theorem \ref{T-5.1}, while we include it here to make the statement complete.

\begin{prop}\label{P-5.5}
Assume that $\Me$ is a finite von Neumann algebra with a faithful normal finite trace $\tau$.
Let $\psi,\ffi\in\Me_*^+$, $\psi\ne0$. Then we have the following:
\begin{itemize}
\item[(i)] If $0<\alpha<1<\alpha'$ and $0<z\le z'\le\infty$, then
\[
D_{\alpha,z}(\psi\|\ffi)\le D_{\alpha,z'}(\psi\|\ffi)\le D_1(\psi\|\ffi)
\le D_{\alpha',z'}(\psi\|\ffi)\le D_{\alpha',z}(\psi\|\ffi).
\]
\item[(ii)] For any $z\in[1,\infty]$,
\begin{align}\label{F-5.12}
\lim_{\alpha\nearrow1}D_{\alpha,z}(\psi\|\ffi)=D_1(\psi\|\ffi).
\end{align}
\item[(iii)] If $D_{\alpha,\alpha}(\psi\|\ffi)<\infty$ for some $\alpha>1$, then for any $z\in(1,\infty]$,
\begin{align}\label{F-5.13}
\lim_{\alpha\searrow1}D_{\alpha,z}(\psi\|\ffi)=D_1(\psi\|\ffi).
\end{align}
\end{itemize}
\end{prop}

\begin{proof}
(i) follows from Lemmas \ref{L-5.3} and \ref{L-5.4}.

(ii)\enspace
For every $z\in[1,\infty]$ and $\alpha\in(0,1)$, it follows from (i) above that
\[
D_{\alpha,1}(\psi\|\ffi)\le D_{\alpha,z}(\psi\|\ffi)\le D_1(\psi\|\ffi).
\]
Hence \eqref{F-5.12} follows since it holds for $D_{\alpha,1}$ by
\cite[Proposition 5.3(3)]{hiai2018quantum}, as stated in \eqref{F-5.1}.

(iii)\enspace
Assume that $D_{\alpha,\alpha}(\psi\|\ffi)<\infty$ for some $\alpha>1$. For every $z\in(1,\infty]$
and $\alpha\in(1,z]$, it follows from (i) that
\[
D_1(\psi\|\ffi)\le D_{\alpha,z}(\psi\|\ffi)\le D_{\alpha,\alpha}(\psi\|\ffi).
\]
Hence \eqref{F-5.13} follows since it holds for $D_{\alpha,\alpha}$ by
\cite[Proposition 3.8(ii)]{jencova2018renyi}, as stated in \eqref{F-5.2}.
\end{proof}

In this subsection, in the specialized setting of finite von Neumann algebras, we have considered the
monotonicity of $D_{\alpha,z}$ in the parameter $z$ in an essentially similar way to the finite dimensional
case \cite{lin2015investigating,mosonyi2023somecontinuity}. By Theorem
\ref{T-5.1}, for $0<\alpha<1$ this monotonicity holds in general von Neumann algebras. In the next
subsection, we will extend it to this general setting also for $\alpha>1$, under a
restriction on $z$, by using the Haagerup reduction theorem and the complex interpolation method. Furthermore,
we will extend the limits given in \eqref{F-5.12} and \eqref{F-5.13} in Sec.~\ref{Sec-6.3}.

\subsection{The general von Neumann algebra case}\label{Sec-5.2}

From now on let $\Me$ be again a general von Neumann algebra. In the next proposition, we first
extend Proposition \ref{P-5.5}(i) to general von Neumann algebras, based on Haagerup's reduction
theorem \cite{haagerup2010areduction}, that is briefly explained in Appendix \ref{Sec-B} for
the convenience of the reader.

\begin{prop}\label{P-5.6}
For every $\psi,\ffi\in\Me_*^+$, $\psi\ne0$, we have the following:
\begin{itemize}
\item[(i)] If $0<\alpha<1$ and $0<z\le z'$, then
\[
D_{\alpha,z}(\psi\|\ffi)\le D_{\alpha,z'}(\psi\|\ffi)\le D_1(\psi\|\ffi).
\]
\item[(ii)] If $\alpha>1$ and $\max\{\alpha/2,\alpha-1\}\le z\le z'\le\alpha$, then
\[
D_1(\psi\|\ffi)\le D_{\alpha,z'}(\psi\|\ffi)\le D_{\alpha,z}(\psi\|\ffi).
\]
\end{itemize}
\end{prop}

\begin{proof}
Since $D_{\alpha,z}(\psi\|\ffi)$ is the same as that for $\psi,\ffi$ restricted to $s(\psi+\ffi)\Me s(\psi+\ffi)$
(see Remark \ref{remark:defi}) and similarly for $D_1(\psi\|\ffi)$, we may assume that $\Me$ is
$\sigma$-finite. If $0<\alpha<1$ and $\max\{\alpha,1-\alpha\}\le z$, then we have
$D_{\alpha,z}(\psi\|\ffi)\le D_1(\psi\|\ffi)$ by Proposition \ref{P-5.5}(i) and Lemma \ref{L-B.2} in
Appendix \ref{Sec-B}. Combining this and Theorem \ref{T-5.1} yields (i). Statement (ii) is
again a consequence of Proposition \ref{P-5.5}(i) and Lemma \ref{L-B.2}.
\end{proof}

For the proof of Proposition \ref{P-5.6}(ii) by use of Lemma \ref{L-B.2} (based on Theorem \ref{T-B.1})
it is inevitable to restrict the parameter $z$ to the DPI bounds when $\alpha>1$. However, in the next
theoem we show the $\alpha>1$ counterpart of Theorem \ref{T-5.1}, showing the monotonicity of
$D_{\alpha,z}$ in $z\ge\alpha/2$ when $\alpha>1$. Nevertheless, we note that the inequalities between
$D_{\alpha,z}$ and $D_1$ in (i) and (ii) of Proposition \ref{P-5.6} are not included in
Theorems \ref{T-5.1} and \ref{T-5.7}.
The proof below is based on Kosaki's interpolation $L^p$-spaces \cite{kosaki1984applications},
briefly explained in Appendix \ref{Sec-C}.

\begin{theorem}\label{T-5.7}
For every $\psi,\ffi\in\Me_*^+$, $\psi\ne0$, and $\alpha>1$, the function $z\mapsto D_{\alpha,z}(\psi\|\varphi)$
is monotone decreasing on $[\alpha/2,\infty)$.
\end{theorem}

\begin{proof}
Let $\alpha>1$ and $z,z'\in[\alpha/2,\infty)$ be such that $z<z'$. We need to prove that
$Q_{\alpha,z}(\psi\|\ffi)\ge Q_{\alpha,z'}(\psi\|\ffi)$. To do this, we may assume that
$Q_{\alpha,z}(\psi\|\ffi)<\infty$. Hence by Lemma \ref{lemma:renyi_2z}, there is some $y\in L^{2z}(\Me)$
such that
\[
h_\psi^{\alpha\over2z}=yh_\ffi^{\alpha-1\over2z},\qquad
Q_{\alpha,z}(\psi\|\ffi)=\|y\|_{2z}^{2z}.
\]
In particular, $e:=s(\psi)\le s(\ffi)$, so that we may assume that $\ffi$ is faithful. Let $\sigma\in\Me_*^+$
be such that $s(\sigma)=\1-e$, and set $\psi_0:=\psi+\sigma$, so that $\psi_0$ is faithful too. Let us use
for simplicity the notation $L^p_L$ for Kosaki's left $L^p$-space $L^p(\Me,\ffi)_L$ for $1\le p\le\infty$;
see \eqref{F-C.2} in Appendix \ref{Sec-C}.

Consider the function
\begin{align}\label{F-5.16}
f(w):=h_{\psi_0}^{{\alpha\over2z}w}eh_\ffi^{1-{\alpha\over2z}w},\qquad w\in S,
\end{align}
where $S:=\{w\in\bC:0\le\Re w\le1\}$. Then, for $w=s+it$ with $0\le s\le1$ and $t\in\bR$, since
\[
f(s+it)=h_{\psi_0}^{{\alpha\over2z}it}h_{\psi_0}^{{\alpha\over2z}s}e
h_\ffi^{1-{\alpha\over2z}s}h_\ffi^{-{\alpha\over2z}it},
\]
we have
\begin{align*}
&f(it)=h_{\psi_0}^{{\alpha\over2z}it}eh_\ffi^{-{\alpha\over2z}it}h_\ffi\in L^\infty_L, \\
&\|f(it)\|_{L^\infty_L}=\Big\|h_{\psi_0}^{{\alpha\over2z}it}eh_\ffi^{-{\alpha\over2z}it}\Big\|=1,
\qquad t\in\bR,
\end{align*}
and
\begin{align*}
&f(1+it)=h_{\psi_0}^{{\alpha\over2z}it}h_\psi^{\alpha\over2z}
h_\ffi^{1-{\alpha\over2z}}h_\ffi^{-{\alpha\over2z}it}
=\Bigl(h_{\psi_0}^{{\alpha\over2z}it}yh_\ffi^{-{\alpha\over2z}it}\Bigr)h_\ffi^{2z-1\over2z}\in L^{2z}_L, \\
&\|f(1+it)\|_{L^{2z}_L}=\Big\|h_{\psi_0}^{{\alpha\over2z}it}yh_\ffi^{-{\alpha\over2z}it}\Big\|_{2z}=\|y\|_{2z},
\qquad t\in\bR,
\end{align*}
where the last equality follows from \cite[Lemma 10.1]{kosaki1984applications}. Furthermore,
we observe that $f$ is a bounded continuous function on $S$ into $L^1(\Me)$ and it is analytic in the
interior (see, e.g., \cite[Lemma 9.19 and Theorem 9.18(c)]{hiai2021lectures}). Here, the continuity of
$f$ on $S$ can be verified, for instance, by using \cite[Theorem A.7 and Proposition 10.6]{hiai2021lectures},
while we omit the details. Therefore, $f$ belongs to the
set $\cF'(L^\infty_L,L^{2z}_L)$ of $L^1(\Me)$-valued functions given in
\cite[Definition 1.4]{kosaki1984applications}. Since $L^{2z}_L$ is reflexive thanks to $1<\alpha\le2z<\infty$,
it follows from \cite[Theorems 1.5 and Remark 3.4]{kosaki1984applications} that the set
$\cF'(L^\infty_L,L^{2z}_L)$ defines the interpolation space $C_\theta=C_\theta(L^\infty_L,L^{2z}_L)$ in
\cite[Definition 1.1]{kosaki1984applications}. Hence for any $\theta\in(0,1)$, we have $f(\theta)\in C_\theta$
and
\[
\|f(\theta)\|_{C_\theta}\le\biggl(\sup_{t\in\bR}\|f(it)\|_{L^\infty_L}\biggr)^{1-\theta}
\biggl(\sup_{t\in\bR}\|f(1+it)\|_{L^{2z}_L}\biggr)^\theta=\|y\|_{2z}^\theta,
\]
where the inequality above is known from \cite[Lemma 4.3.2(ii)]{bergh1976interpolation}. By
\cite[Theorem 1.9]{kosaki1984applications} and the reiteration theorem (see
\cite{cwikel1978complex} {for the best result on that}), $C_\theta=L^{2z/\theta}_L$
with equal norms, so that putting $\theta=z/z'$ we have
\[
f(z/z')=h_\psi^{\alpha\over2z'}h_\ffi^{1-{\alpha\over2z'}}=y'h_\ffi^{2z'-1\over2z'}
\]
for some $y'\in L^{2z'}(\Me)$, and $\|y'\|_{2z'}\le\|y\|_{2z}^{z/z'}$. This implies that
$h_\psi^{\alpha\over2z'}=y'h_\ffi^{\alpha-1\over2z'}$ so that
$Q_{\alpha,z'}(\psi\|\ffi)=\|y'\|_{2z'}^{2z'}\le\|y\|_{2z}^{2z}$, and the assertion follows.
\end{proof}

\section{Monotonicity in the parameter $\alpha$}\label{Sec-6}

In this section we show the monotonicity of $D_{\alpha,z}$ in the parameter $\alpha$ as well as limits of
$D_{\alpha,z}$ as $\alpha\nearrow1$ and $\alpha\searrow1$.

\subsection{The case $0<\alpha<1$ and all $z>0$}\label{Sec-6.1}

The aim of this subsection is to prove the next theorem.

\begin{theorem}\label{T-6.1}
Let $\psi,\ffi\in\Me_*^+$ and $z>0$. Then we have
\begin{itemize}
\item[(1)] $\alpha\mapsto\log Q_{\alpha,z}(\psi\|\ffi)$ is convex on $(0,1)$,
\item[(2)] $\alpha\mapsto D_{\alpha,z}(\psi\|\ffi)$ is monotone increasing on $(0,1)$.
\end{itemize}
\end{theorem}

To prove this, we obtain a certain ``log-majorization'' result for positive $\tau$-measurable
operators, which might be meaningful in its own. Assume that $(\cR,\tau)$ be a semi-finite von Neumann
algebra $\cR$ with a faithful normal semi-finite trace $\tau$. Let $\widetilde\cR$ denote the space of
$\tau$-measurable operators affiliated with $\cR$. We consider operators $a\in\widetilde\cR$ satisfying
\begin{align}\label{F-6.1}
\mbox{$a\in\cR$\ \ or\ \ $\mu_t(a)\le Ct^{-\gamma}$ ($t>0$) for some $C,\gamma>0$,}
\end{align}
where $\mu_t(a)$, $t>0$, is the ($t$th) generalized $s$-number of $a$. 
For each $a\in\widetilde\cR$ satisfying \eqref{F-6.1} we define (see \cite{fack1986generalized})
\[
\Lambda_t(a):=\exp\int_0^t\log\mu_s(a)\,ds,\qquad t>0.
\]
Note \cite[]{fack1986generalized} that $\Lambda_t(a)\in[0,\infty)$, $t>0$, are well defined whenever $a$
satisfies \eqref{F-6.1}. Also, note that if $a,b\in\widetilde\cR$ satisfy \eqref{F-6.1}, then $|a|^p$ $(p>0$)
and $ab$ satisfy \eqref{F-6.1} too, as it is clear since $\mu_t(ab)\le\|a\|\mu_t(b)$ for $a\in\cR$,
$\mu_t(|a|^p)=\mu_t(a)^p$, and $\mu_t(ab)\le\mu_{t/2}(a)\mu_{t/2}(b)$; see
\cite[Lemma 2.5]{fack1986generalized}.

\begin{prop}\label{P-6.2}
Let $a_j,b_j\in\widetilde\cR^+$, $j=1,2$, satisfying \eqref{F-6.1} and assume that $a_1a_2=a_2a_1$
and $b_1b_2=b_2b_1$. Then for every $\theta\in(0,1)$ and any $t>0$,
\begin{align}\label{F-6.2}
\Lambda_t\Bigl((a_1^\theta a_2^{1-\theta})^{1/2}(b_1^\theta b_2^{1-\theta})
(a_1^\theta a_2^{1-\theta})^{1/2}\Bigr)
&\le\Lambda_t\bigl(a_1^\theta b_1^\theta\bigr)\Lambda_t\bigl(a_2^{1-\theta}b_2^{1-\theta}\bigr).
\end{align}
In particular,
\begin{align}\label{F-6.3}
\Lambda_t\Bigl((a_1^{1/2}a_2^{1/2})^{1/2}(b_1^{1/2}b_2^{1/2})(a_1^{1/2}a_2^{1/2})^{1/2}\Bigr)
\le\Lambda_t\bigl(a_1^{1/2}b_1a_1^{1/2}\bigr)^{1/2}\Lambda_t\bigl(a_2^{1/2}b_2a_2^{1/2}\bigr)^{1/2}.
\end{align}
\end{prop}

\begin{proof}
Let $\theta\in(0,1)$ and $t>0$ be arbitrary. For any $k\in\bN$ we note that
\begin{align*}
&\bigl((a_1^\theta a_2^{1-\theta})^{1/2}(b_1^\theta b_2^{1-\theta})
(a_1^\theta a_2^{1-\theta})^{1/2}\bigr)^k \\
&\quad=(a_1^\theta a_2^{1-\theta})^{1/2}(b_1^\theta b_2^{1-\theta})(a_1^\theta a_2^{1-\theta})
(b_1^\theta b_2^{1-\theta})\cdots(a_1^\theta a_2^{1-\theta})(b_1^\theta b_2^{1-\theta})
(a_1^\theta a_2^{1-\theta})^{1/2} \\
&\quad=(a_1^\theta a_2^{1-\theta})^{1/2}b_2^{1-\theta}(b_1^\theta a_1^\theta)
(a_2^{1-\theta}b_2^{1-\theta})\cdots(b_1^\theta a_1^\theta)(a_2^{1-\theta}b_2^{1-\theta})
b_1^\theta(a_1^\theta a_2^{1-\theta})^{1/2}.
\end{align*}
Since $(a_1^\theta a_2^{1-\theta})^{1/2}b_2^{1-\theta}$, $b_1^\theta a_1^\theta$,
etc.\ are $\tau$-measurable operators satisfying \eqref{F-6.1}, we have
by \cite[Theorem 4.2(ii)]{fack1986generalized}
\begin{align*}
&\Lambda_t\bigl((a_1^\theta a_2^{1-\theta})^{1/2}(b_1^\theta b_2^{1-\theta})
(a_1^\theta a_2^{1-\theta})^{1/2}\bigr)^k \\
&\quad\le\Lambda_t\bigl((a_1^\theta a_2^{1-\theta})^{1/2}b_2^{1-\theta}\bigr)
\Lambda_t(b_1^\theta a_1^\theta)^{k-1}\Lambda_t\bigl(a_2^{1-\theta}b_2^{1-\theta}\bigr)^{k-1}
\Lambda_t\bigl(b_1^\theta(a_1^\theta a_2^{1-\theta})^{1/2}\bigr),
\end{align*}
so that
\begin{align*}
&\Lambda_t\bigl((a_1^\theta a_2^{1-\theta})^{1/2}(b_1^\theta b_2^{1-\theta})
(a_1^\theta a_2^{1-\theta})^{1/2}\bigr) \\
&\quad\le\Lambda_t\bigl((a_1^\theta a_2^{1-\theta})^{1/2}b_2^{1-\theta}\bigr)^{1/k}
\Lambda_t(b_1^\theta a_1^\theta)^{1-{1\over k}}\Lambda_t\bigl(a_2^{1-\theta}b_2^{1-\theta}\bigr)^{1-{1\over k}}
\Lambda_t\bigl(b_1^\theta(a_1^\theta a_2^{1-\theta})^{1/2}\bigr)^{1/k}.
\end{align*}
Letting $k\to\infty$ gives \eqref{F-6.2}. When $\theta=1/2$, \eqref{F-6.2}  is rewritten as \eqref{F-6.3}.
\end{proof}

\begin{remark}\label{R-6.3}
Since $\Lambda_t(a_j^rb_j^r)\le\Lambda_t(a_jb_j)^r$ for any $r\in(0,1)$ by \cite{kosaki1992aninequality},
inequality \eqref{F-6.2} implies that
\[
\Lambda_t\bigl((a_1^\theta a_2^{1-\theta})^{1/2}(b_1^\theta b_2^{1-\theta})
(a_1^\theta a_2^{1-\theta})^{1/2}\bigr)
\le\Lambda_t(a_1b_1)^\theta\Lambda_t(a_2b_2)^{1-\theta}
=\Lambda_t\bigl(a_1b_1^2a_1\bigr)^{\theta\over2}\Lambda_t\bigl(a_2b_2^2a_2)^{1-\theta\over2}.
\]
We are indeed interested in whether a stronger inequality
\[
\Lambda_t\bigl((a_1^\theta a_2^{1-\theta})^{1/2}(b_1^\theta b_2^{1-\theta})
(a_1^\theta a_2^{1-\theta})^{1/2}\bigr)
\le\Lambda_t\bigl(a_1^{1/2}b_1a_1^{1/2}\bigr)^\theta
\Lambda_t\bigl(a_2^{1/2}b_2a_2^{1/2}\bigr)^{1-\theta}
\]
hold or not in the situation of Proposition \ref{P-6.2}. The last inequality is known to hold in the
finite dimensional setting, as shown in \cite[Theorem 2.1]{hiai2024log-majorization}.
\end{remark}

\begin{lemma}\label{L-6.4}
Let $\psi,\ffi\in\Me_*^+$ with $\psi\ne0$, and assume that $s(\psi)\not\perp s(\ffi)$. Then for every $z>0$,
$Q_{\alpha,z}(\psi\|\ffi)>0$ for all $\alpha\in(0,1)$, and $\alpha\mapsto Q_{\alpha,z}(\psi\|\ffi)$ is continuous
on $(0,1)$.
\end{lemma}

\begin{proof}
Assume that $Q_{\alpha,z}(\psi\|\ffi)=0$ for some $z>0$ and $\alpha\in(0,1)$. Then
$h_\psi^{\alpha/2z}h_\ffi^{(1-\alpha)/2z}=0$ as a $\tau$-measurable operator affiliated with
$\cR:=\Me\rtimes_{\sigma^\omega}\bR$ (see Appendix \ref{Sec-A}). 
Since $s(\psi)=s(h_\psi^{\alpha/2z})$ and $s(\ffi)=s(h_\ffi^{(1-\alpha)/2z})$, it is easy to see that
$s(\psi)\perp s(\ffi)$. Hence the first assertion follows.

Next, since $p>0\mapsto a^p\in\widetilde\cR$ is differentiable in the measure topology for any
$a\in\widetilde\cR^+$ (see, e.g., \cite[Lemma 9.19]{hiai2021lectures}), we see that
$\alpha\mapsto h_\psi^{\alpha/2z}h_\ffi^{(1-\alpha)/z}h_\psi^{\alpha/2z}$ is differentiable (hence continuous)
on $(0,1)$ in the measure topology. Hence by \cite[Lemma 9.14]{hiai2021lectures}, the function
$\alpha\mapsto Q_{\alpha,z}(\psi\|\ffi)=\|h_\psi^{\alpha/2z}h_\ffi^{(1-\alpha)/z}h_\psi^{\alpha/2z}\|_z^z$ is
continuous on $(0,1)$. Here, when $z<1$, note \cite[Theorem 4.9(iii)]{fack1986generalized} that
$|\,\|a\|_z^z-\|b\|_z^z|\le\|a-b\|_z^z$ for $a,b\in L^z(\Me)$.
\end{proof}

\noindent
{\bf Proof of Theorem \ref{T-6.1}.}\enspace
We may assume that $s(\psi)\not\perp s(\ffi)$; otherwise, $Q_{\alpha,z}(\psi\|\ffi)=0$ for all
$\alpha\in(0,1)$. Then by Lemma \ref{L-6.4}, $Q_{\alpha,z}(\psi\|\ffi)\in(0,\infty)$ for all $\alpha\in(0,1)$,
and $\alpha\mapsto Q_{\alpha,z}(\psi\|\ffi)$ is continuous on $(0,1)$. Hence
$\alpha\mapsto D_{\alpha,z}(\psi\|\ffi)$ is continuous on $(0,1)$ too.

Let $\alpha_1,\alpha_2\in(0,1)$ and $z>0$. Let $(\cR,\tau)$ be as in the proof of Lemma \ref{L-6.4}. Consider
$a_j:=h_\psi^{\alpha_j/z}$ and $b_j:=h_\ffi^{(1-\alpha_j)/z}$ in $\widetilde\cR^+$. Since
$a_j\in L^{z/\alpha_j}(\Me)$, we note by \cite[Lemma 4.8]{fack1986generalized} that
$\mu_t(a_j)=t^{-\alpha_j/z}\|a_j\|_{z/\alpha_j}$, $t>0$, and hence $a_j$ satisfies \eqref{F-6.1}. Similarly, $b_j$
does so. Therefore, we can apply \eqref{F-6.3} to $a_j,b_j$ with $t=1$ to obtain
\begin{align}
&\int_0^1\log\mu_s\Bigl(h_\psi^{\alpha_1+\alpha_2\over4z}
h_\ffi^{2-\alpha_1-\alpha_2\over2z}h_\psi^{\alpha_1+\alpha_2\over4z}\Bigr)\,ds \nonumber\\
&\quad\le{1\over2}\biggl[\int_0^1\log\mu_s\Bigl(h_\psi^{\alpha_1\over2z}
h_\ffi^{1-\alpha_1\over z}h_\psi^{\alpha_1\over2z}\Bigr)\,ds
+\int_0^1\log\mu_s\Bigl(h_\psi^{\alpha_2\over2z}
h_\ffi^{1-\alpha_2\over z}h_\psi^{\alpha_2\over2z}\Bigr)\,ds\biggr]. \label{F-6.4}
\end{align}
Since $h_\psi^{\alpha_1+\alpha_2\over4z}h_\ffi^{2-\alpha_1-\alpha_2\over2z}h_\psi^{\alpha_1+\alpha_2\over4z}$
is in $L^z(\Me)$, we have by \cite[Lemma 4.8]{fack1986generalized} again
\[
\mu_s\Bigl(h_\psi^{\alpha_1+\alpha_2\over4z}
h_\ffi^{2-\alpha_1-\alpha_2\over2z}h_\psi^{\alpha_1+\alpha_2\over4z}\Bigr)
=s^{-1/z}\Big\|h_\psi^{\alpha_1+\alpha_2\over4z}
h_\ffi^{2-\alpha_1-\alpha_2\over2z}h_\psi^{\alpha_1+\alpha_2\over4z}\Big\|_z
\]
so that
\begin{align}\label{F-6.5}
\log\mu_s\Bigl(h_\psi^{\alpha_1+\alpha_2\over4z}
h_\ffi^{2-\alpha_1-\alpha_2\over2z}h_\psi^{\alpha_1+\alpha_2\over4z}\Bigr)^z
=-\log s+\log Q_{{\alpha_1+\alpha_2\over2},z}(\psi\|\ffi).
\end{align}
Similarly,
\begin{align}\label{F-6.6}
\log\mu_s\Bigl(h_\psi^{\alpha_j\over2z}h_\ffi^{1-\alpha_j\over z}h_\psi^{\alpha_j\over2z}\Bigr)^z
=-\log s+\log Q_{\alpha_j,z}(\psi\|\ffi),\qquad j=1,2.
\end{align}
Multiply $z$ to both sides of \eqref{F-6.4} and insert \eqref{F-6.5} and \eqref{F-6.6} into it. Since
$\int_0^1(-\log s)\,ds=1$, we then arrive at
\[
1+\log Q_{{\alpha_1+\alpha_2\over2},z}(\psi\|\ffi)
\le{1\over2}\bigl[2+\log Q_{\alpha_1,z}(\psi\|\ffi)+\log Q_{\alpha_2,z}(\psi\|\ffi)\bigr],
\]
which implies that $\alpha\mapsto\log Q_{\alpha,z}(\psi\|\ffi)$ is midpoint convex on $(0,1)$.
Since midpoint convexity implies convexity for continuous functions, (1) holds. Moreover, by
\cite[Theorem 1(vii)]{kato2023onrenyi} we find that
$\lim_{\alpha\nearrow1}Q_{\alpha,z}(\psi\|\ffi)\le\psi(\1)$. Therefore, (1) implies (2) from the defining
formula of $D_{\alpha,z}$ in Definition \ref{defi:renyi}.\qed

\begin{remark}\label{remark:2nd-proof}
Theorem \ref{T-6.1} can also be proved, though restricted to $z>1/2$, by use of the complex
interpolation method based on Kosaki's interpolation $L^p$-spaces. It is worthwhile to sketch the
proof here. It suffices to prove (1) of Theorem \ref{T-6.1} as (2) is immediate from (1) as in the last
part of the above proof of Theorem \ref{T-6.1}. Assume that $z>1/2$ and let $p:=2z$ and
$q:={2z\over2z-1}$ with $1/p+1/q=1$. We may assume that $\Me$ is $\sigma$-finite, by restricting
$\psi,\ffi$ to $s(\psi+\ffi)\Me s(\psi+\ffi)$; see Remark \ref{remark:defi}. Choose $\rho,\sigma\in\Me_*^+$
with $s(\rho)=\1-s(\psi)$ and $s(\sigma)=\1-s(\ffi)$, and put $\psi_0:=\psi+\rho$ and $\ffi_0:=\ffi+\sigma$,
faithful functionals in $\Me_*^+$. For $p\in(1,\infty)$ and $\eta\in(0,1)$ we write
\[
L^p_L:=L^p(\Me,\ffi_0)_L,\qquad L^p_R:=L^p(\Me,\psi_0)_R,\qquad
L^p_\eta:=L^p(\Me,\psi_0,\ffi_0)_\eta,
\]
(see \eqref{F-C.1}--\eqref{F-C.3} in Appendix \ref{Sec-C}). Then by \eqref{F-C.6},
\begin{align}\label{F-6.7}
L^p_\eta=C_\eta(L^p_L,L^p_R)
\end{align}
with equal norms. We divide the proof into the cases $z\ge1$ and $1/2<z<1$.

For $z\ge1$, let $h_0:=h_\psi^{1/2}h_\ffi^{1/2}\in L^1(\Me)$. For each $\alpha\in(0,1)$ put
$\eta:={z-\alpha\over2z-1}$; then we have $0<\eta<1$. Since
\[
h_0=h_\psi^{\eta\over q}\Bigl(h_\psi^{\alpha\over2z}h_\ffi^{1-\alpha\over2z}\Bigr)
h_\ffi^{1-\eta\over q}
=h_{\psi_0}^{\eta\over q}\Bigl(h_\psi^{\alpha\over2z}h_\ffi^{1-\alpha\over2z}\Bigr)
h_{\ffi_0}^{1-\eta\over q},
\]
we have $h_0\in L^p_\eta$ and $\|h_0\|_{p,\psi_0,\ffi_0,\eta}^p=Q_{\alpha,z}(\psi\|\ffi)$. 
Now let $\alpha_1,\alpha_2\in(0,1)$. For each $\theta\in(0,1)$ put
$\alpha:=(1-\theta)\alpha_1+\theta\alpha_2$ and $\eta_j:={z-\alpha_j\over2z-1}$, $j=1,2$, so that
$\eta:={z-\alpha\over2z-1}=(1-\theta)\eta_1+\theta\eta_2$. From \eqref{F-6.7} and the reiteration
theorem, note that
\begin{align}\label{F-6.8}
L^p_\eta=C_\theta(L^p_{\eta_1},L^p_{\eta_2}).
\end{align}
Since $h_0\in L^p_{\eta_1}\cap L^p_{\eta_2}$ as shown above, $\|h_0\|_{p,\psi_0,\ffi_0,\eta}\le
\|h_0\|_{p,\psi_0,\ffi_0,\eta_1}^{1-\theta}\|h_0\|_{p,\psi_0,\ffi_0,\eta_2}^\theta$, which implies that
\begin{align}\label{F-6.9}
Q_{\alpha,z}(\psi\|\ffi)\le Q_{\alpha_1,z}(\psi\|\ffi)^{1-\theta}Q_{\alpha_2,z}(\psi\|\ffi)^\theta,
\end{align}
as desired.

Next, for $1/2<z<1$, denote $\Sigma=\Sigma(L^p_L,L^p_R):=L^p_L+L^p_R$ and let
$\tilde \cF(L^p_L,L^p_R)$ be the set of functions $f:S:=\{w\in\bC:0\le\Re w\le1\}\to\Sigma$
satisfying
\begin{itemize}
\item[(i)] $f$ is bounded, continuous on $S$ and analytic in the interior of $S$ (with respect to the norm
in $\Sigma$),
\item[(ii)] $f(it)\in L^p_L$ and $f(1+it)\in L^p_R$ for all $t\in\bR$,
\item[(iii)] the maps $t\in\bR\mapsto f(it)\in L^p_L$ and $t\in\bR\mapsto f(1+it)\in L^p_R$ are continuous
and
\[
\max\biggl\{\sup_{t\in\bR}\|f(it)\|_{L^p_L},\sup_{t\in\bR}\|f(1+it)\|_{L^p_R}\biggr\}<\infty.
\]
\end{itemize}
It is known \cite[Theorems 1.5 and 11.1]{kosaki1984applications} (see also \eqref{F-C.6} in
Appendix \ref{Sec-C}) that the set $\tilde \cF(L^p_L,L^p_R)$ defines the interpolation space
$L^p_\eta=C_\eta(L^p_L,L^p_R)$ with equal norms.
Consider the function $f:S\to L^1(\Me)$ defined by
\begin{align}\label{F-6.10}
f(w):=h_\psi^{{w\over q}+{1-w\over p}}h_\ffi^{{1-w\over q}+{w\over p}},
\qquad w\in S.
\end{align}
Then we can manage to show that $f\in\tilde \cF(L^p_L,L^p_R)$ and
for each $\eta\in(0,1)$ and $t\in\bR$, $f(\eta+it)\in L^p_\eta$ and
$\|f(\eta+it)\|_{p,\ffi_0,\psi_0,\eta}^p=Q_{1-\eta,z}(\psi\|\ffi)$.
Now let $\alpha_1,\alpha_2\in(0,1)$. For each $\theta\in(0,1)$ let
$\alpha:=(1-\theta)\alpha_1+\theta\alpha_2$ and $\eta_j:=1-\alpha_j$, $j=1,2$, so that
$\eta:=1-\alpha=(1-\theta)\eta_1+\theta\eta_2$. With $f$ given in \eqref{F-6.10}
define $f_1(w):=f((1-w)\eta_1+w\eta_2)$; then we have $f_1\in\tilde \cF(L^p_{\eta_1},L^p_{\eta_2})$.
Since $L^p_\eta=C_\theta(L^p_{\eta_1},L^p_{\eta_2})$ by the reiteration theorem,
\[
\|f(\eta)\|_{p,\ffi_0,\psi_0,\eta}=\|f_1(\theta)\|_{C_\theta(L^p_{\eta_1},L^p_{\eta_2})}
\le\biggl(\sup_{t\bR}\|f_1(it)\|_{L^p_{\eta_1}}\biggr)^{1-\theta}
\biggl(\sup_{t\in\bR}\|f_1(1+it)\|_{L^p_{\eta_2}}\biggr)^\theta,
\]
which implies that $Q_{1-\eta,z}(\psi\|\ffi)\le
Q_{1-\eta_1,z}(\psi\|\ffi)^{1-\theta}Q_{1-\eta_2,z}(\psi\|\ffi)^\theta$, as desired.
\end{remark}

\subsection{The case $1<\alpha\le2z$}\label{Sec-6.2}

In this subsection let us show monotonicity of $D_{\alpha,z}$ in the parameter $\alpha\in(1,2z]$ when
$z>1/2$, based on the complex interpolation in a similar way {to arguments} in
Remark \ref{remark:2nd-proof}.

\begin{theorem}\label{T-6.6}
Let $\psi,\ffi\in\Me_*^+$ and $z>1/2$. Then we have
\begin{itemize}
\item[(1)] $\alpha\mapsto\log Q_{\alpha,z}(\psi\|\ffi)$ is convex on $(1,2z]$,
\item[(2)] $\alpha\mapsto D_{\alpha,z}(\psi\|\ffi)$ is monotone increasing on $(1,2z]$.
\end{itemize}
\end{theorem}

\begin{proof}
Assume that $z>1/2$ and 
let $p$, $q$, $\psi_0$ and $\ffi_0$ be defined in the same way as in the first paragraph of
Remark \ref{remark:2nd-proof}. For each $\alpha\in(1,2z]$ put $\eta:={2z-\alpha\over2z-1}\in[0,1)$.
Assume that $Q_{\alpha,z}(\psi\|\ffi)<\infty$ (hence $s(\psi)\le s(\ffi)$), so that there exists a unique
$y\in s(\psi)L^p(\Me)s(\ffi)$ such that $h_\psi^{\alpha\over2z}=yh_\ffi^{\alpha-1\over2z}$. Since
\[
h_\psi=h_\psi^{2z-\alpha\over2z}yh_\ffi^{\alpha-1\over2z}
=h_{\psi_0}^{\eta\over q}yh_{\ffi_0}^{1-\eta\over q},
\]
we have $h_\psi\in L^p_\eta$ and
\[
\|h_\psi\|_{p,\psi_0,\ffi_0,\eta}^p=\|y\|_p^p=Q_{\alpha,z}(\psi\|\ffi),
\]
where for $\eta=0$ ($\alpha=2z$) the left-hand side is $\|h_\psi\|_{L^p_L}^p$. Now let
$\alpha_1,\alpha_2\in(1,2z]$ and $\alpha=(1-\theta)\alpha_1+\theta\alpha_2$ for any $\theta\in(0,1)$.
Put $\eta_j:={2z-\alpha_j\over2z-1}$, $j=1,2$, and
$\eta:={2z-\alpha\over2z-1}=(1-\theta)\eta+\theta\eta_1$. To show (1), it suffices to prove that
\eqref{F-6.9} holds in the present situation. For this, we may assume that
$Q_{\alpha_j,z}(\psi\|\ffi)<\infty$, $j=1,2$. Then we can use \eqref{F-6.8} similarly to the discussion
in Remark \ref{remark:2nd-proof} with $h_\psi$ instead of $h_0$. Hence we have \eqref{F-6.9},
and (1) follows.

As for (2), note that $h_\psi=h_{\psi_0}^{1/q}h_\psi^{1/p}\in L^p_R$ (see \eqref{F-C.3}) and
$\|h_\psi\|_{L^p_R}^p=\|h_\psi^{1/p}\|_p^p=\psi(\1)$. Assume that $1<\alpha<\alpha_1\le2z$ and
$Q_{\alpha_1,z}(\psi\|\ffi)<\infty$, so that $\alpha=(1-\theta)\alpha_1+\theta$ for some $\theta\in(0,1)$.
Let $\eta:={2z-\alpha\over2z-1}$ and $\eta_1:={2z-\alpha_1\over2z-1}$. Since
\[
L^p_\eta=C_\theta(L^p_{\eta_1},L^p_R)
\]
by the reiteration theorem, it follows that
\[
Q_{\alpha,z}(\psi\|\ffi)\le Q_{\alpha_1,z}(\psi\|\ffi)^{1-\theta}\psi(\1)^\theta.
\]
Taking the logarithm and noting $\theta={\alpha_1-\alpha\over\alpha_1-1}$, we obtain
$D_{\alpha,z}(\psi\|\ffi)\le D_{\alpha_1,z}(\psi\|\ffi)$, proving (2).
\end{proof}

\subsection{Limits as $\alpha\nearrow1$ and $\alpha\searrow1$}\label{Sec-6.3}

The aim of this last subsection is to show the limits of $D_{\alpha,z}$ as $\alpha\nearrow1$ and
$\alpha\searrow1$, extending the limits in \eqref{F-6.1} and \eqref{F-6.2}.

\begin{theorem}\label{T-6.7}
Let $\psi,\ffi\in\Me_*^+$, $\psi\ne0$. For every $z>0$ we have
\[
\lim_{\alpha\nearrow1}D_{\alpha,z}(\psi\|\ffi)=D_1(\psi\|\ffi).
\]
\end{theorem}

\begin{proof}
Assume first that $z\in(0,1]$ and $0\le1-z<\alpha<1$. Let $\beta:={\alpha-1+z\over z}$; then
$0<\beta<1$ and $\beta\nearrow1$ as $\alpha\nearrow1$. Hence the result follows from Lemma \ref{L-6.8}
below and \eqref{F-5.1} for $D_{\alpha,1}$. On the other hand, for the case $z\in[1,\infty)$ the result
follows from Proposition \ref{P-5.6}(i) as \eqref{F-5.12} does from Proposition \ref{P-5.5}(i).
\end{proof}

\begin{lemma}\label{L-6.8}
Assume that $z\in(0,1]$ and $0\le1-z<\alpha<1$. Let $\beta:={\alpha-1+z\over z}$. Then for any
$\psi,\ffi\in\Me_*^+$, $\psi\ne0$,
\[
D_{\beta,1}(\psi\|\ffi)\le D_{\alpha,z}(\psi\|\ffi)\le D_{\alpha,1}(\psi\|\ffi).
\]
\end{lemma}

\begin{proof}
Since the statement is trivial for $z=1$, we may assume that $z\in(0,1)$. The second inequality follows
from Proposition \ref{P-5.6}(i). For the first inequality, noting that $\beta\in(0,1)$ by assumption and 
by the H\"older inequality with ${1\over2z}={1-z\over2z}+{1\over2}$, we have
\begin{align*}
Q_{\alpha,z}(\psi\|\ffi)&=\Big\|h_\psi^{\alpha\over2z}h_\ffi^{1-\alpha\over2z}\Big\|_{2z}^{2z}
=\Big\|h_\psi^{1-z\over2z}h_\psi^{\beta\over2}h_\ffi^{1-\beta\over2}\Big\|_{2z}^{2z} \\
&\le\Big\|h_\psi^{1-z\over2z}\Big\|_{2z\over1-z}^{2z}
\Big\|h_\psi^{\beta\over2}h_\ffi^{1-\beta\over2}\Big\|_2^{2z}
=\psi(\1)^{1-z}Q_{\beta,1}(\psi\|\ffi)^z,
\end{align*}
which proves the second inequality since $\alpha-1=z(\beta-1)$.
\end{proof}

\begin{theorem}\label{T-6.9}
Let $\psi,\ffi\in\Me_*^+$, $\psi\ne0$, and $z>1/2$. Assume that $D_{\alpha,z}(\psi\|\ffi)<\infty$ for some
$\alpha\in(1,2z]$. Then we have
\[
\lim_{\alpha\searrow1}D_{\alpha,z}(\psi\|\ffi)=D_1(\psi\|\ffi).
\]
\end{theorem}

\begin{proof}
Assume that $z>1/2$ and $D_{\alpha,z}(\psi\|\ffi)<\infty$ for some $\alpha\in(1,2z]$. We may assume that
$\ffi$ is faithful. We utilize the function $f$ on $S$ given in \eqref{F-5.16}, whose values are in $L^{2z}_L$
as seen from the proof of Theorem \ref{T-5.7}. Since $f$ is analytic in a neighborhood of $1/\alpha$, we have
the expansion
\[
f(w)=f\biggl({1\over\alpha}\biggr)+\biggl(w-{1\over\alpha}\biggr)h+o\biggl(w-{1\over\alpha}\biggr),
\]
where $h\in L^{2z}_L$ is the derivative of $f$ at $w=1/\alpha$ and $\|o(\zeta)\|_{L^{2z}_L}/|\zeta|\to0$ as
$|\zeta|\to0$. For each $\alpha'\in(1,\alpha)$ it follows that
\[
f\biggl({\alpha'\over\alpha}\biggr)=f\biggl({1\over\alpha}\biggr)+{\alpha'-1\over\alpha}\,h
+o\biggl({\alpha'-1\over\alpha}\biggr)\quad\mbox{as $\alpha'\searrow1$}.
\]
Furthermore,  since $f(\alpha'/\alpha)\in L^{2z}_L$, we have
$f(\alpha'/\alpha)=h_\psi^{\alpha'\over2z}h_\ffi^{1-{\alpha'\over2z}}=y'h_\ffi^{2z-1\over2z}$ for some
$y'\in L^{2z}(\Me)$, so that $h_\psi^{\alpha'\over 2z}=y'h_\ffi^{\alpha'-1\over 2z}$ and 
$Q_{\alpha',z}(\psi\|\ffi)=\|y'\|_{2z}^{2z}=\|f(\alpha'/\alpha)\|_{L^{2z}_L}$.

Now let us recall that $L^{2z}_L$ is uniformly convex thanks to $2z>1$ (see \cite{haagerup1979lpspaces},
\cite[Theorem 4.2]{kosaki1984applications}), so that the norm $\|\cdot\|_{L^{2z}_L}$ is uniformly Fr\'echet
differentiable (see, e.g., \cite[Part 3, Chap.~II]{beauzamy1982introduction}). We set
$a_0\in L^{2z\over2z-1}_L$ with the unit norm by
\[
a_0:=\biggl({h_\psi\over\psi(\1)}\biggr)^{2z-1\over2z}h_\ffi^{1\over2z},
\]
so that $\<a_0,f(1/\alpha)\>=\|f(1/\alpha)\|_{L^{2z}_L}$, where the dual pairing of $L^{2z\over2z-1}_L$
and $L^{2z}_L$ is given in \eqref{F-C.5} in Appendix \ref{Sec-C} with $p={2z\over2z-1}$. Then
the uniform Fr\'echet differentiability of $\|\cdot\|_{L^{2z}_L}$ at $1/\alpha$ implies that
\begin{align}\label{F-6.11}
\<a_0,h\>=\lim_{\alpha'\searrow1}{\|f(\alpha'/\alpha)\|_{L^{2z}_L}-\|f(1/\alpha)\|_{L^{2z}_L}
\over{\alpha'-1\over\alpha}}
\end{align}
and also
\begin{align}
\<a_0,h\>&=\lim_{t\to0}{1\over it}\<a_0,f((1/\alpha)+it)-f(1/\alpha)\> \nonumber\\
&=\psi(\1)^{-{2z-1\over2z}}\lim_{t\to0}{1\over it}\Bigl\<h_\psi^{2z-1\over2z}h_\ffi^{1\over2z},
h_\psi^{1\over2z}\Bigl(h_{\psi_0}^{{\alpha\over2z}it}h_\ffi^{-{\alpha\over2z}it}-\1\Bigr)
h_\ffi^{2z-1\over2z}\Bigr\> \nonumber\\
&=\psi(\1)^{-{2z-1\over2z}}\lim_{t\to0}{1\over it}\Tr\Bigl[h_\psi^{2z-1\over2z}h_\psi^{1\over2z}
\bigl(h_{\psi_0}^{{\alpha\over2z}it}h_\ffi^{-{\alpha\over2z}it}-\1\bigr)\Bigr] \nonumber\\
&=\psi(\1)^{-{2z-1\over2z}}{\alpha\over2z}
\lim_{t\to0}{1\over it}\Tr\bigl[h_\psi\bigl(h_{\psi_0}^{it}h_\ffi^{-it}-\1\bigr)\bigr] \nonumber\\
&=\psi(\1)^{-{2z-1\over2z}}{\alpha\over2z}\,D(\psi\|\ffi), \label{F-6.12}
\end{align}
where we have used \eqref{F-C.5} for the third equality and \cite[Theorem 5.7]{ohya1993quantum}
for the last equality. Since $\psi(\1)=\|f(1/\alpha)\|_{L^{2z}_L}^{2z}$, it follows from \eqref{F-6.11} and
\eqref{F-6.12} that
\begin{align*}
D_{\alpha',z}(\psi\|\ffi)
&={\log Q_{\alpha',z}(\psi\|\ffi)-\log\psi(\1)\over\alpha'-1}
={2z\log\|f(\alpha'/\alpha)\|_{L^{2z}_L}-2z\log\|f(1/\alpha)\|_{L^{2z}_L}\over\alpha'-1} \\
&=\biggl({\log\|f(\alpha'/\alpha)\|_{L^{2z}_L}-\log\|f(1/\alpha)\|_{L^{2z}_L}\over
\|f(\alpha'/\alpha)\|_{L^{2z}_L}-\|f(1/\alpha)\|_{L^{2z}_L}}\biggr)
{2z\over\alpha}\biggl({\|f(\alpha'/\alpha)\|_{L^{2z}_L}-\|f(1/\alpha)\|_{L^{2z}_L}\over
{\alpha'-1\over\alpha}}\biggr) \\
&\to{1\over\psi(\1)^{1\over2z}}\,{2z\over\alpha}\,\psi(\1)^{-{2z-1\over2z}}{\alpha\over2z}\,D(\psi\|\ffi)
={D(\psi\|\ffi)\over\psi(\1)}=D_1(\psi\|\ffi)
\end{align*}
as $\alpha'\searrow1$, proving the statement.
\end{proof}

\section{Concluding remarks}\label{Sec-7}

We have advanced, after the paper \cite{kato2023onrenyi} by Kato, the study of $\alpha$-$z$-R\'enyi
divergences $D_{\alpha,z}(\psi\|\ffi)$ ($\alpha,z>0$, $\alpha\ne1$) for normal positive functionals $\psi,\ffi$
on general von Neumann algebras, based on the Haagerup $L^p$-spaces (as well as Kosaki's
interpolation $L^p$-spaces for the technical side). In Theorems \ref{thm:variational} and \ref{thm:dpi}
we have finished the proofs of the variational expression and the DPI for $D_{\alpha,z}(\psi\|\ffi)$ in the case
$\alpha>1$, while those were proved in \cite{kato2023onrenyi} only in the case $0<\alpha<1$. Since the
DPI bounds of $(\alpha,z)$ in Theorem \ref{thm:dpi} are the same as those characterized in
\cite{zhang2020fromwyd} in the finite dimensional case, our DPI theorem for $D_{\alpha,z}$ is best possible.
On the other hand, as noted in Remark \ref{remark:variational}, the variational expressions in Theorem
\ref{thm:variational} were proved in both finite and infinite dimensional $\mathcal{B}(\mathcal{H})$ settings
in \cite{zhang2020fromwyd} and \cite{mosonyi2023thestrong} with no restriction on $z>0$,
while we need the assumption $z\ge\alpha/2$ to prove the variational expression in the case
$\alpha>1$. Thus it may be expected that Theorem \ref{thm:variational}(ii) can be also extended
to all $z>0$.

The main theorem of the paper says that the equality
$D_{\alpha,z}(\psi\circ\gamma\|\ffi\circ\gamma)=D_{\alpha,z}(\psi\|\ffi)<\infty$ in the DPI under a normal
$2$-positive unital map $\gamma$ between von Neumann algebras implies the sufficiency (the reversibility)
of $\gamma$ with respect to $\{\psi,\ffi\}$. Although, as noted in Remark \ref{rem:conditions}, some
other necessary conditions for the equality case of the DPI were given in \cite{zhang2020equality}
in the finite dimensional setting, our sufficiency theorem is stated in quite a proper form in the sense
of Definition \ref{defi:reversible}, that seems new even in the finite dimensional case. Moreover,
the theorem seems best possible since it holds for all $(\alpha,z)$ in the DPI bounds and thus includes
those for the Petz-type and the sandwiched R\'enyi divergences proved in
\cite[Theorem 6.19]{hiai2021quantum} and \cite{jencova2018renyi,jencova2021renyi}.

In the second half of the paper, we have discussed the monotonicity properties of $D_{\alpha,z}$ in the
parameters $\alpha,z$. In the finite dimensional case, the monotonicity of $D_{\alpha,z}$ in all $z>0$ for
any fixed $\alpha\in(0,\infty)$, $\alpha\ne1$, as well as the limit $\lim_{\alpha\to1}D_{\alpha,z}=D_1$ for
any $z>0$, was shown in \cite{lin2015investigating} (also \cite{mosonyi2023somecontinuity}).
Proposition \ref{P-5.5} is the extension of this to the finite von Neumann algebra case. Furthermore,
in the general von Neumann algebra setting, the monotonicity of $D_{\alpha,z}$ in $z$ was shown in
\cite{kato2023onrenyi} (as Theorems \ref{T-5.1}) for any fixed $\alpha\in(0,1)$, and in Theorem \ref{T-5.7}
we have presented that for $\alpha>1$ with a restriction $z\ge\alpha/2$. It is left open to remove this
restriction when $\alpha>1$. On the other hand, concerning the monotonicity property in $\alpha>0$ for
a fixed $z>0$, it seems that nothing has been known before the present paper and
\cite{hiai2024log-majorization} even in the finite dimensional case. In the general von Neumann algebra
setting, in Theorems \ref{T-6.1} and \ref{T-6.6} we have presented the monotonicity of $D_{\alpha,z}$ in
$\alpha>0$ with $z$ fixed, though with a restriction $\alpha\le2z$ when $\alpha>1$. 
It is interesting to find whether or not the monotonicity in $\alpha$ holds without this restriction;
note that this is true in the finite dimensional setting \cite[Theorem
3.1(ii)]{hiai2024log-majorization}. In Theorems \ref{T-6.7} and
\ref{T-6.9} we have finally shown the limits $\lim_{\alpha\nearrow1}D_{\alpha,z}=D_1$ and
$\lim_{\alpha\searrow1}D_{\alpha,z}=D_1$ (under the assumption that $D_{\alpha,z}(\psi\|\ffi)<\infty$ for
some $\alpha\in(1,2z]$ for the latter) with $z$ fixed, though with a restriction $z>1/2$ for the latter limit.
It is again left open to remove this restriction.

\subsection*{Acknowledgments}

The authors are grateful to the anonymous referee for very careful reading of the manuscript and
for helpful suggestions. The present work started when the authors met at Nagoya University in
November 2023 and had the opportunity to discuss some open problems in \cite{kato2023onrenyi} with
Shinya Kato. The authors thank Yoshimichi Ueda for his arranging the meeting, and are also grateful to
Ueda and Kato for discussions. AJ wishes to thank Francesco Buscemi for the invitation to Nagoya,
which enabled the meeting, and his kind hospitality during the visit. AJ also acknowledges the support by
the grant VEGA 2/0128/24 and the  Science and Technology Assistance Agency under the
contract No.\ APVV-20-0069.

\bigskip

\noindent
\textbf{Data availability}\enspace  No datasets were generated or analyzed during
the current study.

\medskip 

\noindent
\textbf{Competing interests}\enspace The authors have no competing interests to declare.

\appendix

\section{Haagerup  $L^p$-spaces}\label{Sec-A}

Let $\Me$ be a von Neumann algebra on a Hilbert space $\mathcal{H}$. Let $\omega$ be
a faithful normal semi-finite weight $\omega$ on $\Me$ and $\sigma_t^\omega$, $t\in\bR$, be the
associated \emph{modular automorphism group} on $\Me$. Then we have the crossed product
$\cR:=\Me\rtimes_{\sigma^\omega}\bR$, which is a semi-finite von Neumann algebra on the Hilbert
space $\mathcal{H}\otimes L^2(\bR)$ with the canonical trace $\tau$.
Let $\theta_s$, $s\in\bR$, be the \emph{dual action} on $\cR$ having the $\tau$-scaling property
$\tau\circ\theta_s=e^{-s}\tau$, $s\in\bR$; see \cite[Chap.~X]{takesaki2003theoryof} (also
\cite[Chap.~8]{hiai2021lectures}). Let $\widetilde\cR$ denote the space of $\tau$-measurable operators
affiliated with $\cR$; see \cite{fack1986generalized} (also \cite[Chap.~4]{hiai2021lectures}). For $0<p\le\infty$,
the \emph{Haagerup $L^p$-space} $L^p(\Me)$ \cite{haagerup1979lpspaces,terp1981lpspaces} (also
\cite[Chap.~9]{hiai2021lectures}) is defined by
\[
L^p(\Me):=\{a\in\widetilde\cR:\theta_s(a)=e^{-s/p}a,\ s\in\bR\}.
\]
In particular, $\Me=L^\infty(\Me)$ and we have an order isomorphism $\Me_*\cong L^1(\Me)$ given as
$\psi\in\Me_*\leftrightarrow h_\psi\in L^1(\Me)$, so that $\Tr\,h_\psi=\psi(\1)$, $\psi\in L^1(\Me)$, defines a
positive linear functional $\Tr$ on $L^1(\Me)$. For $0<p<\infty$ the $L^p$-norm (quasi-norm for $0<p<1$)
of $a\in L^p(\Me)$ is defined by $\|a\|_p:=(\Tr\,|a|^p)^{1/p}$, and the $L^\infty$-norm $\|\cdot\|_\infty$ is the
operator norm $\|\cdot\|$ on $\Me$. For $1\le p<\infty$, $L^p(\Me)$ is a Banach space whose dual Banach
space is $L^q(\Me)$, where $1/p+1/q=1$, by the duality pairing
\begin{align}\label{F-A.1}
(a,b)\in L^p(\Me)\times L^q(\Me)\mapsto\Tr(ab)=\Tr(ba).
\end{align}

We will next recall the extensions of conditional expectations to the
Haagerup $L^p$-spaces, obtained in \cite{junge2003noncommutative}.
Here, as in \cite{junge2003noncommutative} we assume that $\Me$ is a $\sigma$-finite von
Neumann algebra and $\omega$ is a faithful normal state on $\Me$.
Let $\Ne\subseteq \Me$ be a subalgebra such that
$\sigma_t^\omega(\Ne)\subseteq \Ne$, $t\in\bR$. Equivalently, there is a faithful normal conditional
expectation $\cE$ on $\Me$ with range $\Ne$. We then have
$\sigma^{\omega|_\Ne}=\sigma^\omega|_\Ne$ for the modular group of the restriction
$\omega|_\Ne$ and we may identify $\cS:=\Ne \rtimes_{\sigma^{\omega|_\Ne}}\bR$
with a subalgebra in $\cR$. It can be seen that 
\[
\tilde \cE:=\cE\otimes \mathrm{id}_{L^2(\bR)}
\]
defines a conditional expectation of $\cR$ onto $\cS$. Further, let $\tau_0$ be the
canonical trace for $\cS$, then we have $\tau_0=\tau|_\cS$ and
$\tau=\tau\circ \tilde \cE=\tau_0\circ \tilde \cE$. It follows that $\tilde{\cS}$
can be identified with the subspace in $\tilde \cR$ of elements affiliated with $\cS$ and
since the dual action for $\cS$ is just a restriction of $\theta_s$, $s\in \bR$, we have a
similar identification of $L^p(\Ne)$ with a subspace in $L^p(\Me)$. By the results of
\cite[Sec.~2]{junge2003noncommutative}, $\cE$ extends to a contractive projection
$\cE_p$ from $L^p(\Me)$ onto $L^p(\Ne)$ for any $1\le p\le \infty$, where 
\[
\cE_\infty=\cE,\qquad \cE_1=\cE_*.
\]
Moreover, for any $x\in L^p(\Me)$,
\[
\cE_p(x)^*=\cE_p(x^*),\qquad  x\ge 0 \implies \cE_p(x)\ge 0
\]
and for $1\le p,q,r\le \infty$ such that $\tfrac1p+\tfrac1q+\tfrac1r=\tfrac1s\le 1$, we have
\begin{equation}\label{eq:cond}
\cE_s(axb)=a\cE_r(x)b,\qquad a\in L^p(\Ne),\ b\in L^q(\Ne),\ x\in L^r(\Me)
\end{equation}
(see \cite[Proposition 2.3]{junge2003noncommutative}).

We finish this section by some well-known lemmas. The proofs are given for completeness.

\begin{lemma}\label{lemma:cone}
For any $0<p<\infty$ and $\varphi\in \Me_*^+$, 
$h_\varphi^{\frac1{2p}}\Me^+h_\varphi^{\frac1{2p}}$ is dense in $L^p(\Me)^+$ with respect
to the (quasi)-norm $\|\cdot\|_p$.
\end{lemma}

\begin{proof} We may assume that $\varphi$ is faithful. By \cite[Lemma 1.1]{junge2003noncommutative}, $\Me
h_\varphi^{\frac1{2p}}$ is dense in $L^{2p}(\Me)$ for any $0<p<\infty$. Let $y\in L^p(\Me)^+$, then
$y^{\frac12}\in L^{2p}(\Me)$, hence there is a sequence $a_n\in \Me$ such that
$\|a_nh^{\frac1{2p}}_\varphi-y^{\frac12}\|_{2p}\to 0$. Then also 
\[
\Big\|h^{\frac1{2p}}_\varphi a_n^*-y^{\frac12}\Big\|_{2p}
=\Big\|(a_nh^{\frac1{2p}}_\varphi-y^{\frac12})^*\Big\|_{2p}
=\Big\|a_nh^{\frac1{2p}}_\varphi-y^{\frac12}\Big\|_{2p}\to 0
\]
and 
\[
\Big\|h^{\frac1{2p}}_\varphi a_n^*a_nh^{\frac1{2p}}_\varphi-y\Big\|_p
=\Big\|(h^{\frac1{2p}}_\varphi a_n^*-y^{\frac12})a_nh^{\frac1{2p}}_\varphi
+y^{\frac12}(a_nh^{\frac1{2p}}_\varphi-y^{\frac12})\Big\|_p.
\]
Since $\|\cdot\|_p$ is a (quasi)-norm, the above expression goes to 0 by the H\"older
inequality.
\end{proof}

\begin{lemma}\label{lemma:order1}
Let $0<p\le \infty$ and let $h,k\in L^p(\Me)^+$ be such
that $h\le k$. Then 
$\|h\|_p\le \|k\|_p$. Moreover, if $1\le p<\infty$, then 
\[
\|k-h\|_p^p\le \|k\|_p^p-\|h\|_p^p.
\]
\end{lemma}

\begin{proof} The first statement follows from \cite[Lemmas 2.5(iii) and 4.8]{fack1986generalized}.
The second statement is from \cite[Lemma 5.1]{fack1986generalized}.
\end{proof}

\begin{lemma}\label{lemma:order}
Let $\psi,\varphi\in \Me_*^+$ with $\psi\le \varphi$.
Then for any $a\in \Me$ and $p\in [1,\infty)$,
\[
\Tr\Bigl(\bigl(a^*h_\psi^{1/p}a\bigr)^p\Bigr)\le\Tr\Bigl(\bigl(a^*h_\varphi^{1/p}a\bigr)^p\Bigr).
\]
\end{lemma}

\begin{proof} Since $1/p\in (0,1]$, it follows (see \cite[Lemma B.7]{hiai2021quantum} and
\cite[Lemma 3.2]{hiai2021connections}) that $h_\psi^{1/p}\le h_\varphi^{1/p}$.  Hence
$a^*h_\psi^{1/p}a\le a^*h_\varphi^{1/p}a$. Therefore, by Lemma \ref{lemma:order1} we have
the statement.
\end{proof}

\section{Haagerup's reduction theorem}\label{Sec-B}

In this appendix let us recall Haagerup's reduction theorem, which was presented in
\cite[Sec.~2]{haagerup2010areduction} (a compact survey is also found in
\cite[Sec.~2.5]{fawzi2023asymptotic}). Let $\Me$ be a general $\sigma$-finite von Neumann algebra.
Let $\omega$ be a faithful normal state of $\Me$ and $\sigma_t^\omega$ ($t\in\bR$) be the associated
modular automorphism group. Consider the discrete additive group $G:=\bigcup_{n\in\bN}2^{-n}\bZ$ and
define $\hat\Me:=\Me\rtimes_{\sigma^\omega}G$, the crossed product of $\Me$ by the action
$\sigma^\omega|_G$. Then the dual weight $\hat\omega$ is a faithful normal state of $\hat\Me$, and
we have $\hat\omega=\omega\circ E_\Me$, where $E_\Me:\hat\Me\to\Me$ is the canonical conditional
expectation (see, e.g., \cite[Sec.~8.1]{hiai2021lectures}, also \cite[Sec.~2.5]{fawzi2023asymptotic}).

Haagerup's reduction theorem is summarized as follows:

\begin{theorem}[\cite{haagerup2010areduction}]\label{T-B.1}
In the above setting, there exists an increasing sequence $\{\Me_n\}_{n\ge1}$ of von Neumann
subalgebras of $\hat\Me$, containing the unit of $\hat\Me$, such that the following hold:
\begin{itemize}
\item[(i)] Each $\Me_n$ is finite with a faithful normal tracial state $\tau_n$.
\item[(ii)] $\left(\bigcup_{n\ge1}\Me_n\right)''=\hat\Me$.
\item[(iii)] For every $n$ there exists a (unique) faithful normal conditional expectation
$E_{\Me_n}:\hat\Me\to\Me_n$ satisfying
\[
\hat\omega\circ E_{\Me_n}=\hat\omega,\qquad
\sigma_t^{\hat\omega}\circ E_{\Me_n}=E_{\Me_n}\circ\sigma_t^{\hat\omega},\quad t\in\bR.
\]
Moreover, for any $x\in\hat\Me$, $E_{\Me_n}(x)\to x$ in the $\sigma$-strong topology.
\end{itemize}
\end{theorem}

Furthermore, for any $\psi\in\Me_*^+$ define $\hat\psi:=\psi\circ E_\Me$. Then by
\cite[Theorem 4]{hiai1984strong} we have $\hat\psi\circ E_{\Me_n}\to\hat\psi$ in the norm. Here
we give the next lemma, which is used in Sec.~\ref{Sec-5.2}.

\begin{lemma}\label{L-B.2}
In the above situation, for any $\psi,\ffi\in\Me_*^+$ with $\psi\ne0$ let $\hat\psi:=\psi\circ E_\Me$
and $\hat\ffi:=\ffi\circ E_\Me$. If $\alpha$ and $z$ satisfy the DPI bounds (i.e., condition (i) or (ii) in
Theorem \ref{thm:dpi}), then we have
\begin{align}
D_{\alpha,z}(\psi\|\ffi)&=D_{\alpha,z}(\hat\psi\|\hat\ffi)
=\lim_{n\to\infty}D_{\alpha,z}(\hat\psi|_{\Me_n}\|\hat\ffi|_{\Me_n})\quad\mbox{increasingly},
\label{F-6.16}\\
D_1(\psi\|\ffi)&=D_1(\hat\psi\|\hat\ffi)
=\lim_{n\to\infty}D_1(\hat\psi|_{\Me_n}\|\hat\ffi|_{\Me_n})\quad\mbox{increasingly}. \label{F-6.17}
\end{align}
\end{lemma}

\begin{proof}
Apply the DPI for $D_{\alpha,z}$ proved in Theorem \ref{thm:dpi} to the injection
$\Me\hookrightarrow\hat\Me$ and to the conditional expectation $E_\Me:\hat\Me\to\Me$. We then have
the first equality in \eqref{F-6.16}. By Theorem \ref{T-B.1} we can apply the martingale convergence
in Theorem \ref{thm:martingale} to obtain the latter equality in \eqref{F-6.16} with increasing convergence.
The assertion of $D_1$ in \eqref{F-6.17} is included in \cite[Proposition 2.2]{fawzi2023asymptotic},
while this is an immediate consequence of the well-known martingale convergence and the DPI
of the relative entropy \cite{kosaki1986relative}.
\end{proof}

\section{Kosaki's interpolation $L^p$-spaces}\label{Sec-C}

Assume that $\Me$ is a $\sigma$-finite von Neumann algebra and let faithful
$\psi_0,\ffi_0\in\Me_*^+$ be given. For each $\eta\in[0,1]$ consider an embedding
$\Me\hookrightarrow L^1(\Me)$ by $x\mapsto h_{\psi_0}^\eta xh_{\ffi_0}^{1-\eta}$. Defining
$\|h_{\psi_0}^\eta xh_{\ffi_0}^{1-\eta}\|_\infty:=\|x\|$ (the operator norm of $x$) on
$h_{\psi_0}^\eta\Me h_{\ffi_0}^{1-\eta}$ we have a pair
$\bigl(h_{\psi_0}^\eta\Me h_{\ffi_0}^{1-\eta},L^1(\Me)\bigr)$ of compatible Banach spaces (see, e.g.,
\cite{bergh1976interpolation}). For $1<p<\infty$, \emph{Kosaki's interpolation $L^p$-space}
with respect to $\psi_0,\ffi_0$ and $\eta$ \cite{kosaki1984applications} (also see
\cite[Sec.~9.3]{hiai2021lectures} for a compact survey) is defined as the complex interpolation
Banach space:
\begin{align}\label{F-C.1}
L^p(\Me,\psi_0,\ffi_0)_\eta:=C_{1/p}\bigl(h_{\psi_0}^\eta\Me h_{\ffi_0}^{1-\eta},L^1(\Me)\bigr)
\end{align}
equipped with the interpolation norm $\|\cdot\|_{p,\psi_0,\ffi_0,\eta}:=\|\cdot\|_{C_{1/p}}$. Then,
 Kosaki's theorem \cite[Theorem 9.1]{kosaki1984applications} says that for every $\eta\in[0,1]$ and
 $p\in(1,\infty)$ with $1/p+1/q=1$,
\begin{align*}
&L^p(\Me,\psi_0,\ffi_0)_\eta=h_{\psi_0}^{\eta/q}L^p(\Me)h_{\ffi_0}^{(1-\eta)/q}\ (\subset L^1(\Me)), \\
&\|h_{\psi_0}^{\eta/q}ah_{\ffi_0}^{(1-\eta)/q}\|_{p,\psi_0,\ffi_0,\eta}=\|a\|_p,\qquad a\in L^p(\Me),
\end{align*}
that is, $L^p(\Me)\cong L^p(\Me,\psi_0,\ffi_0)_\eta$ by the isometry
$a\mapsto h_{\psi_0}^{\eta/q}ah_{\ffi_0}^{(1-\eta)/q}$. In the main body of this paper we use
the special cases where $\eta=0,1$, that is,
\begin{align}
L^p(\Me,\ffi_0)_L&:=C_{1/p}\bigl(\Me h_{\ffi_0},L^1(\Me)\bigr)=L^p(\Me)h_{\ffi_0}^{1/q},
\label{F-C.2}\\
L^p(\Me,\psi_0)_R&:=C_{1/p}\bigl(h_{\psi_0}\Me,L^1(\Me)\bigr)=h_{\psi_0}^{1/q}L^p(\Me),
\label{F-C.3}
\end{align}
which are called Kosaki's \emph{left and right $L^p$-spaces}, respectively. Another special case
we use is the \emph{symmetric $L^p$-space} $L^p(\Me,\ffi_0)$ where $\eta=1/2$ and $\psi_0=\ffi_0$, i.e.,
\begin{align}\label{F-C.4}
L^p(\Me,\ffi_0)=C_{1/p}\bigl(h_{\ffi_0}^{1/2}\Me h_{\ffi_0}^{1/2},L^1(\Me)\bigr)
=h_{\ffi_0}^{1/2q}L^p(\Me)h_{\ffi_0}^{1/2q},
\end{align}
whose interpolation norm is denoted by $\|\cdot\|_{p,\ffi_0}$. The $L^p$-$L^q$ duality of Kosaki's
$L^p$-spaces can be given by transforming the duality paring in \eqref{F-A.1}; in particular,
the duality pairing between $L^p(\Me,\ffi_0)_L$ and $L^q(\Me,\ffi_0)_L$ for $1\le p<\infty$ and
$1/p+1/q=1$ is written as
\begin{align}\label{F-C.5}
\<ah_{\ffi_0}^{1/q},bh_{\ffi_0}^{1/p}\>=\Tr(ab),\qquad
a\in L^p(\Me),\ b\in L^q(\Me).
\end{align}
Kosaki's non-commutative Stein--Weiss interpolation theorem \cite[Theorem 11.1]{kosaki1984applications}
says that for each $\eta\in(0,1)$ and $p\in(1,\infty)$, Kosaki's $L^p$-space $L^p(\Me,\psi_0,\ffi_0)_\eta$
given in \eqref{F-C.1} is the complex interpolation space of the left and right $L^p$-spaces in \eqref{F-C.2}
and \eqref{F-C.3} with equal norms, that is,
\begin{align}\label{F-C.6}
L^p(\Me,\psi_0,\ffi_0)_\eta
=C_{1/p}\bigl(h_{\psi_0}^\eta\Me h_{\ffi_0}^{1-\eta},L^1(\Me)\bigr)
=C_\eta(L^p(\Me,\ffi_0)_L,L^p(\Me,\psi_0)_R).
\end{align}


\end{document}